\documentclass[a4paper]{article}
\usepackage[UKenglish]{babel}
\usepackage{amsmath,amsfonts,amssymb,amsthm}
\usepackage{amsbsy}
\usepackage{tikz}
\usepackage{times}
\usepackage{bm,bbm}
\usepackage{a4wide}
\usepackage[sans]{dsfont}
\usepackage[mathscr]{euscript}
\usepackage{cite}
\usepackage{graphicx}
\usepackage{enumerate}
\usepackage{hyperref}
\usepackage{centernot}

\theoremstyle{plain}
\newtheorem{theorem}{Theorem}
\newtheorem{proposition}[theorem]{Proposition}
\newtheorem{corollary}[theorem]{Corollary}
\newtheorem{lemma}[theorem]{Lemma}
\theoremstyle{remark}
\newtheorem{remark}{Remark}
\newtheorem{example}{Example}
\theoremstyle{definition}
\newtheorem{definition}{Definition}

\newcommand{\rmd}{\mathrm{d}}
\newcommand{\rmi}{\mathrm{i}}
\newcommand{\rme}{\mathrm{e}}

\newcommand{\Rbb}{{\mathbb{R}}}
\newcommand{\Cbb}{{\mathbb{C}}}
\newcommand{\Zbb}{{\mathbb{Z}}}
\newcommand{\Tr}[1]{{\rm Tr}\,#1}
\newcommand{\norm}[1]{\left\Vert#1\right\Vert}
\newcommand{\abs}[1]{\left|#1\right|}
\newcommand{\id}{{\mathbbm{1}}}
\newcommand{\argmin}{\operatornamewithlimits{arg\,min}}

\newcommand{\Ao}{\mathsf{A}}
\newcommand{\Bo}{\mathsf{B}}
\newcommand{\Co}{\mathsf{C}}
\newcommand{\Mo}{\mathsf{M}}
\newcommand{\Qo}{\mathsf{Q}}
\newcommand{\Po}{\mathsf{P}}
\newcommand{\To}{\mathsf{T}}
\newcommand{\Uo}{\mathsf{U}}
\newcommand{\Xo}{\mathsf{X}}
\newcommand{\Yo}{\mathsf{Y}}
\newcommand{\Zo}{\mathsf{Z}}

\newcommand{\Srel}[2]{S\big(#1\|#2\big)}
\newcommand{\ii}{c}
\newcommand{\Icomp}{\ii_{\rm inc}}
\newcommand{\Iad}{\ii_{\rm ed}}
\newcommand{\Iprep}{\ii_{\rm prep}}
\newcommand{\Div}[2]{D\big(#1\|#2\big)}
\newcommand{\Mad}{\Mscr_{\rm ed}}
\newcommand{\Mcomp}{\Mscr_{\rm inc}}
\newcommand{\sddiv}[2]{S[#1 \| #2]}

\newcommand{\Ccal}{\mathcal{C}}
\newcommand{\Jcal}{\mathcal{J}}
\newcommand{\Scal}{\mathcal{S}}
\newcommand{\Ucal}{\mathcal{U}}

\newcommand{\Hscr}{\mathscr{H}}
\newcommand{\Jscr}{\mathscr{J}}
\newcommand{\Lscr}{\mathscr{L}}
\newcommand{\Mscr}{\mathscr{M}}
\newcommand{\Pscr}{\mathscr{P}}
\newcommand{\Sscr}{\mathscr{S}}
\newcommand{\Tscr}{\mathscr{T}}
\newcommand{\Uscr}{\mathscr{U}}
\newcommand{\Zscr}{\mathscr{Z}}
\newcommand{\Xscr}{\mathscr{X}}
\newcommand{\Yscr}{\mathscr{Y}}

\def\centerarc[#1](#2,#3)(#4:#5:#6){ \draw[#1,domain=#4:#5] plot ({#2+#6*cos(\x)}, {#3+#6*sin(\x)});}

\newcommand{\ma}[1]{\mathsf{#1}}

\newcommand{\ip}[2]{\left\langle\,#1\,|\,#2\,\right\rangle} 
\newcommand{\kb}[2]{\left|\,#1\,\right\rangle\left\langle\,#2\,\right|} 
\newcommand{\F}{\mathbb{F}}
\newcommand{\sh}{\Sscr(\Hscr)} 
\newcommand{\tr}{{\rm tr}\,}
\newcommand{\lam}{\lambda}

\begin{document}
\title{Measurement uncertainty relations for discrete observables: Relative entropy
formulation}

\author{Alberto Barchielli$^{1,2,3}$,
Matteo Gregoratti$^{1,2}$, Alessandro Toigo$^{1,3}$
\\
\\
$^1$ Politecnico di Milano, Dipartimento di Matematica, \\
Piazza Leonardo da Vinci 32, I-20133 Milano, Italy,\\
$^2$ Istituto Nazionale di Alta Matematica (INDAM-GNAMPA),\\
$^3$ Istituto Nazionale di Fisica Nucleare (INFN), Sezione di Milano}

\maketitle

\begin{abstract}
We introduce a new information-theoretic formulation of quantum measurement uncertainty relations, based on the notion of relative entropy between measurement probabilities. In the case of a finite-dimensional system and for any approximate joint measurement of two target discrete observables, we define the entropic divergence as the maximal total loss of information occurring in the approximation at hand. For fixed target observables, we study the joint measurements minimizing the entropic divergence, and we prove the general properties of its minimum value. Such a minimum is our uncertainty lower bound: the total information lost by replacing the target observables with their optimal approximations, evaluated at the worst possible state. The bound turns out to be also an entropic incompatibility degree, that is, a good information-theoretic measure of incompatibility: indeed, it vanishes if and only if the target observables are compatible, it is state-independent, and it enjoys all the invariance properties which are desirable for such a measure.
In this context, we point out the difference between general approximate joint
measurements and sequential approximate joint measurements; to do this, we
introduce a separate index for the tradeoff between the error of the first
measurement and the disturbance of the second one. By exploiting the symmetry
properties of the target observables, exact values, lower bounds and optimal approximations are
evaluated in two different concrete examples: (1) a couple of spin-1/2
components (not necessarily orthogonal); (2) two Fourier conjugate mutually
unbiased bases in prime power dimension. Finally, the entropic
incompatibility degree straightforwardly generalizes to the case of many observables, still maintaining all its relevant properties; we explicitly compute it for three orthogonal spin-1/2 components.
\end{abstract}

\section{Introduction}

In the foundations of Quantum Mechanics, a remarkable achievement of the last
years has been the clarification of the differences between \emph{preparation
uncertainty relations} (PURs) and \emph{measurement uncertainty relations} (MURs)
\cite{Oza02,Oza03a,Oza03b,Oza04,Wer04,BHL07,AAHB16,Wer16,BLW14a,BLW14b,DamSW15,Oza15,BusLW14},
both of them arising from Heisenberg's heuristic considerations about the
precision with which the position and the momentum of a quantum particle can be
determined \cite{Hei27}.

One speaks of PURs when some lower bound is given on the ``spreads'' of the
distributions of two observables $\Ao$ and $\Bo$ measured in the same state
$\rho$. The most known formulation of PURs, due to Robertson \cite{Rob29},
involves the product of the two standard deviations; more recent formulations
are given in terms of distances among probability distributions \cite{Wer16} or
entropies \cite{Kra87,MaaU88,KriParthas02,MaaWer,WehW10,KTW14,ColesBTW17,AAHB16}.

On the other hand, one refers to MURs when some lower bound is given on the
``errors'' of any approximate joint measurement $\Mo$ of two target observables
$\Ao$ and $\Bo$. When $\Mo$ is realized as a sequence of two measurements, one
for each target observable, MURs are regarded also as relations between the
``error'' allowed in an approximate measurement of the first observable and the
``disturbance'' affecting the successive measurement of the second one.

Although the recent developments of the theory of approximate quantum
measurements \cite{Hol01,BGL97,BHL07,BusH08,BLPY16} and nondisturbing quantum
measurements \cite{HeiW10,HeiM15}  have generated a considerable renewed
interest in MURs, no agreement has yet been reached about the proper
quantifications of the ``error'' or ``disturbance'' terms. Here, the main
problem is how to compare the target observables $\Ao$ and $\Bo$ with their
approximate or perturbed versions provided by the marginals $\Mo_{[1]}$ and
$\Mo_{[2]}$ of $\Mo$; indeed, $\Ao$, $\Mo_{[1]}$, $\Mo_{[2]}$ and $\Bo$ may
typically be incompatible. The proposals then range from operator formulations
of the error \cite{Oza02,Oza03a,Oza03b,Oza04, App98, App16} to distances for
probability distributions \cite{Wer04,Wer16,BLW14a,BLW14b,DamSW15,BusLW14,BHL07} and
conditional entropies \cite{BusHOW,AbbB16,CF15}.

In this paper, we propose and develop a new approach to MURs based on the notion
of {\em relative entropy}. Here we deal with the case of discrete
observables for a finite dimensional quantum system.
The extension to position and momentum is given in \cite{BGT17}.

In the spirit of Busch, Lahti, Werner \cite{Wer04,Wer16,BLW14a,BLW14b,BusLW14}, we quantify the ``error'' in the approximation $\Ao\simeq\Mo_{[1]}$ by comparing the respective outcome
distributions $\Ao^\rho$ and $\Mo_{[1]}^\rho$ in every possible state $\rho$;
however, differently from \cite{Wer04,Wer16,BLW14a,BLW14b,BusLW14}, the comparison is
done from the point of view of information theory. Then, the natural choice is
to consider $\Srel{\Ao^\rho}{\Mo_{[1]}^\rho}$, the relative entropy of
$\Ao^\rho$ with respect to $\Mo_{[1]}^\rho$, as a quantification of the
information loss when $\Ao^\rho$ is approximated with $\Mo_{[1]}^\rho$.
Similarly, in order to quantify either the ``error'' or -- if $\Ao$ and $\Bo$
are measured in sequence -- the ``disturbance'' related to the approximation
$\Bo\simeq\Mo_{[2]}$, we employ the relative entropy
$\Srel{\Bo^\rho}{\Mo_{[2]}^\rho}$. Relative entropy appears to be the
fundamental quantity from which the other entropic notions can be derived, cf.\
\cite{BA02,CovT06,OhP93}.
It should be noticed that relative entropy, of
classical or quantum type, has already been used in quantum measurement theory
to give proper measures of information gains and losses in various scenarios
\cite{OhP93,BarL05,BarL06a,BarL06b,Mac07}.

The relative entropy formulation of MURs,
given in Section \ref{sec:MUR}, is: for every approximate joint
measurement $\Mo$ of $\Ao$ and $\Bo$, there exists a state $\rho$ such that
\begin{equation}\label{S0}
\Srel{\Ao^{\rho}}{\Mo^{\rho}_{[1]}}+\Srel{\Bo^{\rho}}{\Mo^{\rho}_{[2]}}\geq c(\Ao,\Bo),
\end{equation}
where the uncertainty lower bound
\begin{equation}\label{c0}
c(\Ao,\Bo)=\inf_\Mo\sup_\rho\Big\{\Srel{\Ao^{\rho}}{\Mo^{\rho}_{[1]}}+\Srel{\Bo^{\rho}}{\Mo^{\rho}_{[2]}}\Big\}
\end{equation}
depends on the allowed joint measurements $\Mo$. In the above
definition, the same state $\rho$ appears in both error terms
$\Srel{\Ao^{\rho}}{\Mo^{\rho}_{[1]}}$ and
$\Srel{\Bo^{\rho}}{\Mo^{\rho}_{[2]}}$; thus, by making their sum, all possible error compensations are taken into account in the maximization. The quantity $\sup_\rho\left\{\Srel{\Ao^{\rho}}{\Mo^{\rho}_{[1]}}+\Srel{\Bo^{\rho}}{\Mo^{\rho}_{[2]}}\right\}$ gives a state-independent quantification of the total inefficiency of the approximate joint measurement $\Mo$ at hand, and we call it entropic divergence of $\Mo$ from $(\Ao, \Bo)$.

By considering any possible approximate joint measurement in the definition of
$c(\Ao,\Bo)$, we get an uncertainty lower bound $\Icomp(\Ao,\Bo)$ that turns out to be a proper measure of the incompatibility of $\Ao$ and $\Bo$. On the
other hand, by considering only sequential measurements, we derive an
uncertainty lower bound $\Iad(\Ao,\Bo)$ that provides a suitable quantification
of the error/disturbance tradeoff for the two (sequentially ordered) target
observables. Indeed, such lower bounds share a lot of desirable properties:
they are zero if and only if the target observables are compatible
(respectively, sequentially compatible); they are invariant under unitary
transformations and relabelling of the output values of the measurements; and
finally, they are bounded from above by a value that is independent of both the
dimension of the Hilbert space and the number of the possible outcomes. As a
main result, we show also that, for a generic couple of observables $\Ao$ and
$\Bo$, considering only their sequential measurements is a real restriction,
because in general $\Iad(\Ao,\Bo)$ may be larger than $\Icomp(\Ao,\Bo)$;
actually, the two indexes are guaranteed to coincide only if one makes some
extra assumptions on $\Ao$ and $\Bo$ (e.g.\  if the second observable
$\Bo$ is supposed to be sharp).

Thus, every time $\Ao$ and $\Bo$ are incompatible, the total loss of
information $\Srel{\Ao^{\rho}}{\Mo^{\rho}_{[1]}}$ $+$ $\Srel{\Bo^{\rho}}{\Mo^{\rho}_{[2]}}$ in
the approximations $\Ao\simeq\Mo_{[1]}$ and $\Bo\simeq\Mo_{[2]}$ depends on
both the joint measurement $\Mo$ and the state $\rho$; however, since
$\Icomp(\Ao,\Bo)>0$, inequality \eqref{S0} states that there is a minimum potential loss that no
joint measurement $\Mo$ can avoid. Similar remarks hold for sequential measurements and the corresponding error/disturbance coefficient. Note that, even if $\Ao$ and $\Bo$ are incompatible, the left hand side of \eqref{S0} can vanish if the state $\rho$ and the approximate joint measurement $\Mo$ are suitably chosen (see Section \ref{sec:MUR}). Of course, this is not a contradiction, as the formulation \eqref{S0}, \eqref{c0} of MURs is about the size of the total information loss in the worst -- but not all -- input states. In this sense, the bound \eqref{c0} is a state-independent quantification of the minimal inefficiency of the approximations $\Ao\simeq\Mo_{[1]}$ and $\Bo\simeq\Mo_{[2]}$.

Our MURs directly compare with those of \cite{Wer04,Wer16,BLW14a,BLW14b,BusLW14}, from which however they differ in one essential aspect: the latter quantify the inaccuracy of the approximate joint measurement $\Mo$ by maximizing the errors of the approximations $\Ao^{\rho_1}\simeq\Mo^{\rho_1}_{[1]}$ and $\Bo^{\rho_2}\simeq\Mo^{\rho_2}_{[2]}$ over independently chosen states $\rho_1$ and $\rho_2$; instead, in \eqref{c0} we maximize the total approximation error $\Srel{\Ao^{\rho}}{\Mo^{\rho}_{[1]}} + \Srel{\Bo^{\rho}}{\Mo^{\rho}_{[2]}}$ over a single state $\rho$.
On the conceptual level, this amounts to say that our MURs are a statement about the inaccuracy of the approximation $(\Ao,\Bo)\simeq (\Mo_{[1]},\Mo_{[2]})$ that occurs in one preparation of the system; those of \cite{Wer04,Wer16,BLW14a,BLW14b,BusLW14} rather refer to the inefficiencies of two separate uses of the approximate joint measurement $\Mo$, namely, for approximating $\Ao\simeq\Mo_{[1]}$ in a first preparation, and $\Bo\simeq\Mo_{[2]}$ in a second one. Similar considerations hold for the conditional entropy approach of \cite{BusHOW,AbbB16,CF15}, where the ``noise'' and ``disturbance'' terms are defined through different preparations in a sort of calibration procedure.
In this respect, our MURs are reminiscent of the traditional entropic PURs, which relate the spreads of the distributions $\Ao^\rho$ and $\Bo^\rho$ evaluated at the same state $\rho$ (see Section \ref{sec:PUR}).

Whenever $\Ao$
and $\Bo$ are incompatible, we will look for the exact value of
$\Icomp(\Ao,\Bo)$, or at least some lower bound for it, as well as we will try
to determine the optimal approximate joint measurements $\Mo$ which saturate
the minimum. In particular, we will prove that in some relevant applications
there is actually a unique such $\Mo$, thus showing that in these cases the
entropic optimality criterium unambiguously fixes the best approximate joint
measurement.

The generalization of our MURs to the case of
more than two target observables is rather straightforward by the very structure of the relative entropy formulation.
It is worth noticing that there are
triples of observables whose optimal approximate joint measurements are not
unique, even if all their possible pairings do have the corresponding binary
uniqueness property (see e.g. the two and three orthogonal spin-1/2 components
in Sections \ref{ex:D2} and \ref{ex:3spin}).

Now, we summarize the structure of the paper. In Section \ref{sec:indices}, we state our entropic MURs for two target observables, and we introduce and study the main mathematical objects which are involved in their formulation.
In Section \ref{sec:symm}, we undertake the explicit computation of the incompatibility indexes $\Icomp(\Ao,\Bo)$ and
$\Iad(\Ao,\Bo)$ and their respective optimal approximate joint measurements
$\Mo$ for several examples of incompatible target observables. Some general
results are proved, which show how the symmetry properties of the quantum
system can help in the task. Then, two cases are studied: two spin-1/2
components, which we do not assume to be necessarily orthogonal, and two
Fourier conjugate observables associated with a pair of mutually unbiased bases
(MUBs) in prime power dimension. In Section \ref{sec:manyobs}, we
generalize the relative entropy formulation of MURs to the case of many target
observables. As an example, the case of three orthogonal spin-1/2 components is
completely solved. Finally, Section \ref{sec:CO} contains a conclusive discussion and presents some open problems. Three further appendices are provided at the end of the paper: in Appendix \ref{app:HW}, a couple of examples show that the coefficients $\Icomp(\Ao,\Bo)$, $\Iad(\Ao,\Bo)$ and $\Iad(\Bo,\Ao)$ may be different in general; Appendices \ref{app:spin} and \ref{app:MUB} collect all the technical details and proofs for the cases studied in Sections \ref{ex:D2}, \ref{sec:cov.MUBs}, and \ref{ex:3spin}.

\subsection{Observables and instruments}\label{Ob+In}
We start by fixing our quantum system and recalling the notions and basic facts
on observables and measurements that we will use in the article
\cite{DavQTOS,Hol01,Hol12,HeZi12,BGL97,BLPY16,HeiW10,HeiMZ16}.

\paragraph{The Hilbert space $\Hscr$ and the spaces $\Lscr(\Hscr),\, \Tscr(\Hscr),\, \Sscr(\Hscr)$}
We consider a quantum system described by a finite-dimensional complex Hilbert
space $\Hscr$, with $\dim \Hscr=d$; then, the spaces $\Lscr(\Hscr)$ of all
linear bounded operators on $\Hscr$  and the trace-class $\Tscr(\Hscr)$
coincide. Let $\Sscr(\Hscr)$ denote the convex set of all states on $\Hscr$
(positive, unit trace operators), which is a compact subset of $\Tscr(\Hscr)$.
The extreme points of $\Sscr(\Hscr)$ are the pure states (rank-one projections)
$\rho = \kb{\psi}{\psi}$, with $\psi\in\Hscr$ and $\norm{\psi} = 1$.

\paragraph{The space of observables $\Mscr(\Xscr)$ and the space of probabilities
$\Pscr(\Xscr)$} In the general formulation of quantum mechanics, an
\emph{observable} is identified with a \emph{positive operator valued measure}
(POVM). We will consider only observables with outcomes in a finite set $\Xscr$. Then, a POVM on $\Xscr$ is
identified with its discrete density $\Ao:\Xscr\to\Lscr(\Hscr)$, whose values
$\Ao(x)$ are positive operators on $\Hscr$ such that $\sum_{x\in\Xscr}\Ao(x) =
\id$; here, the sum involves a finite number $\abs{\Xscr}$ of terms ($\abs{\Xscr}$
denotes the cardinality of $\Xscr$). Similarly, a probability on $\Xscr$ is
identified with its discrete probability density (or mass function)
$p:\Xscr\to\Rbb$, where $p(x)\geq0$ and $\sum_{x\in\Xscr}p(x) = 1$.

For $\rho\in\sh$, the function $\Ao^\rho(x) = \Tr\{\rho\Ao(x)\}$ is the
discrete probability density on $\Xscr$ which gives the outcome distribution in
a measurement of the observable $\Ao$ performed on the quantum system prepared
in the state $\rho$.

We denote by $\Mscr(\Xscr)$ the set of the observables which are associated with
the system at hand and have outcomes in $\Xscr$; $\Mscr(\Xscr)$ is a convex,
compact subset of $\Lscr(\Hscr)^{\Xscr}$, the finite dimensional linear space
of all functions from $\Xscr$ to $\Lscr(\Hscr)$. Both mappings
$\rho\mapsto\Ao^\rho$ and $\Ao\mapsto\Ao^\rho$ are continuous and affine (i.e.\
preserving convex combinations) from the respective domains into the convex set
$\Pscr(\Xscr)$ of the probabilities on $\Xscr$. As a subset of $\Rbb^{\Xscr}$,
the set $\Pscr(\Xscr)$ is convex and compact. The extreme points of
$\Pscr(\Xscr)$ are the (Kronecker) delta distributions  $\delta_x$, with
$x\in\Xscr$.

\paragraph{Trivial and sharp observables}
An observable $\Ao$ is {\em trivial} if $\Ao = p\id$ for some probability $p$,
where $\id$ is the identity of $\Hscr$. In particular, we will make use of the
uniform distribution $u_\Xscr$ on $\Xscr$, $u_\Xscr(x)=1/\abs\Xscr$, and the
\emph{trivial uniform observable} $\Uo_\Xscr=u_\Xscr \id$.

An observable $\Ao$ is \emph{sharp} if $\Ao(x)$ is a projection $\forall x \in
\Xscr$. Note that we allow $\Ao(x)=0$ for some $x$, which is required when
dealing with sets of observables sharing the same outcome space. Of course, for
every sharp observable we have $|\{x:\Ao(x)\neq0\}|\leq d$.

\paragraph{Bi-observables and compatible observables}
When the outcome set has the product form $\Xscr\times\Yscr$, we speak of
bi-observables. In this case, given the POVM $\Mo\in\Mscr(\Xscr\times \Yscr)$,
we can introduce also the \emph{marginal observables}
$\Mo_{[1]}\in\Mscr(\Xscr)$ and $\Mo_{[2]}\in\Mscr(\Yscr)$ by
\[
\Mo_{[1]}(x) = \sum_{y\in\Yscr}
\Mo(x,y) , \qquad \Mo_{[2]}(y) = \sum_{x\in\Xscr}
\Mo(x,y).
\]
In the same way, for $p\in\Pscr(\Xscr\times \Yscr)$, we get the \emph{marginal probabilities} $p_{[1]}\in\Pscr(\Xscr)$ and
$p_{[2]}\in\Pscr(\Yscr)$. Clearly, $(\Mo_{[i]})^\rho = (\Mo^\rho)_{[i]}$; hence
there is no ambiguity in writing $\Mo^\rho_{[i]}$ for both probabilities.

Two observables $\Ao\in\Mscr(\Xscr)$ and
$\Bo\in\Mscr(\Yscr)$ are \emph{jointly measurable} or \emph{compatible} if
there exists a bi-observable $\Mo\in \Mscr(\Xscr\times \Yscr)$ such that
$\Mo_{[1]}=\Ao$ and $\Mo_{[2]}=\Bo$; then, we call $\Mo$ a \emph{joint
measurement} of $\Ao$ and $\Bo$.

Two classical probabilities $p\in\Pscr(\Xscr)$ and $q\in\Pscr(\Yscr)$ are
always compatible, as they can be seen as the marginals of at least one joint
probability in $\Pscr(\Xscr\times\Yscr)$. Indeed, one can take the {\em product
probability} $p\otimes q$ given by $(p\otimes q)(x,y) = p(x)q(y)$. Clearly,
nothing similar can be defined for two non-commuting quantum observables, for which
instead compatibility usually is a highly nontrivial requirement.

\paragraph{The space of instruments $\Jscr(\Xscr)$} Given a pre-measurement
state $\rho$, a POVM allows to compute the probability distribution of the
measurement outcome. In order to describe also the state change produced
by the measurement, we need the more general mathematical notion of
\emph{instrument}, i.e.\ a measure $\Jcal$ on the outcome set $\Xscr$ taking
values in the set of the completely positive maps on $\Lscr(\Hscr)$. In our
case of finitely many outcomes, an instrument is described by its discrete
density $x\mapsto\Jcal_x$, $x\in \Xscr$, whose general structure is
$\Jcal_x[\rho]=\sum_\alpha J^\alpha_x\rho J^{\alpha\,*}_x$,
$\forall\rho\in\Sscr(\Hscr)$; here, the Kraus operators $J^\alpha_x\in
\Lscr(\Hscr)$ are such that $\sum_{x\in\Xscr}\sum_\alpha J^{\alpha\,*}_x
J^\alpha_x=\id$ and, since $\Hscr$ is finite-dimensional, the index $\alpha$
can be restricted to finitely many values. The \emph{adjoint instrument} is
given by $\Jcal^*_x[F]=\sum_\alpha J^{\alpha\,*}_x FJ^\alpha_x$, $\forall F\in
\Lscr(\Hscr)$. The sum $\Jcal_{\Xscr}=\sum_{x\in\Xscr}\Jcal_x$ is a
\emph{quantum channel}, i.e.\ a completely positive trace preserving map on
$\Sscr(\Hscr)$. We denote by $\Jscr(\Xscr)$ the convex and compact set of all
$\Xscr$-valued instruments for our quantum system.

By setting $\Ao(x)=\Jcal_x^*[\id]=\sum_\alpha J^{\alpha\,*}_x J^\alpha_x$, a
POVM $\Ao\in \Mscr(\Xscr)$ is defined, which is the observable measured by the
instrument $\Jcal$; we say that \emph{the instrument $\Jcal$ implements the
observable $\Ao$}. The state of the system after the measurement, conditioned
on the outcome $x$, is $\Jcal_x[\rho]/\Ao^\rho(x)$. We recall that, given an
observable $\Ao$, one can always find an instrument $\Jcal$ implementing $\Ao$,
but $\Jcal$ is not uniquely determined by $\Ao$, i.e.\ different instruments
$\Jcal$, with different actions on the quantum system, may be used to measure
the same observable $\Ao$.

\paragraph{Sequential measurements and sequentially compatible observables}
Employing the notion of instrument, we can describe a measurement of an
observables $\Ao\in\Mscr(\Xscr)$ followed by a measurement of an observable
$\Bo\in\Mscr(\Yscr)$: a \emph{sequential measurement} of $\Ao$ followed by
$\Bo$ is a bi-observable $\Mo(x,y)=\Jcal^*_x[\Bo(y)]$, where $\Jcal$ is any
instrument implementing $\Ao$.  Its marginals are
$\Mo_{[1]}(x)=\Jcal^*_x[\id]=\Ao(x)$ and
$\Mo_{[2]}(y)=\Jcal^*_{\Xscr}[\Bo(y)]$. We write $\Mo=\Jcal^*(\Bo)$, which is a
measurement in which one first applies the instrument $\Jcal$ to measure $\Ao$,
and then he measures the observable $\Bo$ on the resulting output state; in
this way, he obtains a {joint measurement} of $\Ao$ and
$\Jcal^*_{\Xscr}[\Bo(\cdot)]$, a perturbed version of $\Bo$.

An observable $\Ao\in\Mscr(\Xscr)$ \emph{can be measured without disturbing}
$\Bo\in\Mscr(\Yscr)$  \cite{HeiW10}, or shortly $\Ao$ and $\Bo$ are
\emph{sequentially compatible observables}, if there exists a sequential
measurement $\Mo=\Jcal^*(\Bo)$ such that
\[
\Mo_{[1]}\equiv\Jcal^*_\cdot[\id]=\Ao,\qquad\Mo_{[2]}\equiv\Jcal^*_{\Xscr}[\Bo(\cdot)]=\Bo.
\]
So, a measurement of $\Bo$ at time 1 (i.e.\ after the measurement of
$\Ao$) has the same outcome distribution as a measurement of $\Bo$ at time 0
(i.e.\ before the measurement of $\Ao$).

If $\Ao$ and $\Bo$ are sequentially compatible observables,  they clearly
are also jointly measurable. However, the opposite is not true; two
counterexamples are shown in \cite{HeiW10} and are reported in Appendix
\ref{app:HW}. This happens because we demand to measure just $\Bo$ at time 1,
i.e.\ we do not content ourselves with getting at time 1 the same outcome
distribution of a measurement of $\Bo$ performed at time 0. Indeed, this second
requirement is weaker: it can be satisfied by any couple of jointly measurable
observables $\Ao$ and $\Bo$, by measuring a suitable third observable $\Co$ after
$\Ao$ (with $\Ao$ implemented by an instrument $\Jcal$ which possibly increases
the dimension of the Hilbert space). The definition of sequentially compatible
observables is not symmetric, and indeed there exist couples of observables
such that $\Ao$ can be measured without disturbing $\Bo$, but for which the
opposite is not true. This asymmetry is also reflected in the remarkable fact that,
if the second observable is sharp, then the compatibility of $\Ao$ and $\Bo$
turns out to be equivalent to their sequential compatibility.

\paragraph{Target observables} In this paper, we fix two target observables
with finitely many values, $\Ao\in\Mscr(\Xscr)$ and $\Bo\in\Mscr(\Yscr)$, and
we study how to characterize their uncertainty relations. For
any $\rho\in\sh$, the associated probability distributions $\Ao^\rho$ and
$\Bo^\rho$ can be estimated  by measuring either $\Ao$ or $\Bo$ in many
identical preparations of the quantum system in the state $\rho$. No joint or
sequential measurement of $\Ao$ and $\Bo$ is required at this stage. In Section
\ref{sec:indices} we develop a general theory to quantify the error made
by approximating $\Ao$ and $\Bo$ with compatible observables and we introduce the
notion of optimal approximate joint measurement for $\Ao$ and $\Bo$.

\subsection{Relative and Shannon entropies}\label{sec:c-ent}

In this paper, we will be concerned with entropic quantities of classical type
\cite{CovT06,BA02}; we express them in ``bits'', which means to use
logarithms with base 2: $\log \equiv \log_2$.

The fundamental quantity is the \emph{relative entropy}; although it can be defined for general probability measures, here we only recall the discrete case. Given two probabilities $p,\,q\in \Pscr(\Xscr)$, the relative
entropy of $p$ with respect to $q$ is
\begin{equation}\label{eq:defS}
\Srel pq=
\begin{cases}\displaystyle
\sum_{x\in {\rm supp}\,p} p(x)\log \frac{p(x)}{q(x)} & \text{ if ${\rm supp}\,p\subseteq{\rm supp}\,q$,}
\\ {}+\infty & \text{ otherwise;}
\end{cases}
\end{equation}
it defines an extended real valued function on the product set
$\Pscr(\Xscr)\times\Pscr(\Xscr)$. Also the terms \emph{Kullback-Leibler
divergence} and \emph{information for discrimination} are used for $\Srel pq$.

The relative entropy $\Srel pq$ is a measure of the inefficiency of assuming that
the probability is $q$ when the true probability is $p$ \cite[Sect.\ 2.3]{CovT06}; in other words, it is the amount of information lost when $q$ is used to approximate $p$ \cite[p.\ 51]{BA02}. It appears in data compression theory \cite[Theor.\ 5.4.3]{CovT06}, model selection problems \cite{BA02}, and it is related to the error probability in the context of hypothesis tests that discriminate the two distributions $p$ and $q$ \cite[Theor.\ 11.8.3]{CovT06}. We stress that $\Srel pq$
compares $p$ and $q$, but it is not a distance since it is not symmetric.  As
such, the use of $S$ is particularly convenient when the two probabilities have
different roles; for instance, if $p$ is the true distribution of a given
random variable, while $q$ is the distribution actually used as an
approximation of $p$. This will be our case, where the role of $p$ is played by
the distribution $\Ao^\rho$ (or $\Bo^\rho$) of the target observable $\Ao$ (or
$\Bo$) and  $q$ will be the distribution of some allowed approximation;
in particular, no joint distribution of $p$ and $q$ is involved.

In comparing our results with entropic PURs, we need also the \emph{Shannon entropy} of a probability $p\in\Pscr(\Xscr)$. It is defined by
\begin{equation}\label{def:H}
H(p)=- \sum_{x\in\Xscr}p(x)\log p(x) ,
\end{equation}
and it provides a measure of the uncertainty of a random variable with distribution $p$ \cite[Sect.\ 2.1]{CovT06}.

We collect in the following proposition the main properties of the
relative and Shannon entropies \cite{CovT06,BA02,OhP93,Hol12,Top01}. For the
definition and main properties of lower semicontinuous (LSC) functions, we
refer to \cite[Sect.~1.5]{Ped89}.

\begin{proposition}\label{prop:propHSrel}
The following properties hold.
\begin{enumerate}[(i)]
\item \label{HSbound} $0\leq H(p)\leq \log \abs\Xscr$ and $\Srel pq \geq
    0$, for all $p,\, q\in \Pscr(\Xscr)$.
\item \label{HSnull} $H(p)=0$ if and only if $p=\delta_x$ for some $x$,
    where $\delta_x$ is the delta distribution at $x$. $\Srel pq = 0$ if
    and only if $p=q$.
\item \label{HSrel} $H(u_\Xscr)=\log \abs\Xscr$, and $H(p)=\log \abs\Xscr -\Srel p {u_\Xscr}
    $ for all $p\in \Pscr(\Xscr)$, where $u_\Xscr$ is the
    uniform probability on $\Xscr$.
\item \label{HSinv} $H$ and $S$ are invariant for relabelling of the
    outcomes; that is, if $f:\Xscr'\to\Xscr$ is a bijective map, then
    $H(p\circ f) = H(p)$ and $\Srel{p\circ f}{q\circ f} = \Srel{p}{q}$.
\item \label{HSconvex} $H$ is a concave function on $\Pscr(\Xscr)$, and $S$
    is jointly convex on $\Pscr(\Xscr)\times\Pscr(\Xscr)$, namely
\[
\Srel{\lambda p_1 + (1-\lambda) p_2}{\lambda q_1 + (1-\lambda) q_2}
\leq \lambda \Srel{p_1}{q_1} + (1-\lambda) \Srel{p_2}{q_2}, \ \ \forall\lambda\in [0,1].
\]

\item \label{HScont} The function $p\mapsto H(p)$ is continuous on
    $\Pscr(\Xscr)$. The function $(p,q)\mapsto \Srel pq$ is LSC on $\Pscr(\Xscr)\times\Pscr(\Xscr)$.

\item \label{HSprod} If $p_1,p_2\in\Pscr(\Xscr)$ and
    $q_1,q_2\in\Pscr(\Yscr)$, then $ \Srel{p_1\otimes q_1}{p_2\otimes q_2}
    = \Srel{p_1}{p_2} + \Srel{q_1}{q_2}$.
\end{enumerate}
\end{proposition}

In order to derive some further specific properties of the relative entropy that
will be needed in the following, it is useful to introduce the extended real
function $s:[0,1]\times[0,1]\to[-(\log\rme)/\rme,+\infty]$, with
\begin{equation}\label{eq:defs}
s(u,v)=\begin{cases}\displaystyle u\log\frac{u}{v} &\text{if }0<u\leq 1 \text{ and } 0<v\leq 1,\\
0 &\text{if }u=0 \text{ and } 0\leq v\leq 1,\\
+\infty &\text{if }u>0 \text{ and } v=0.
\end{cases}
\end{equation}
In terms of $s$, the relative entropy can be rewritten as $\Srel{p}{q} = \sum_{x\in\Xscr} s(p(x),q(x))$. Note that, unlike the relative entropy, the function $s$ can take also negative values, and
its minimum is $s(1/\rme,1) = -(\log \rme)/\rme$. As a function of $(u,v)$, $s$ is continuous
at all the points of the square $[0,1]\times [0,1]$ except at the origin
$(0,0)$, where it is  easily proved to be LSC.

\begin{proposition}\label{prop:propHSrel2}
For all $\lambda\in (0,1]$ and $q\in\Pscr(\Xscr)$, the map $g_\lambda(p)=\Srel{p}{\lambda p+(1-\lambda)q}$ is finite and continuous in $p\in\Pscr(\Xscr)$. It attains the maximum value
\begin{equation}\label{eq:bAlb}
\max_{p\in\Pscr(\Xscr)}
\Srel{p}{\lambda p + (1-\lambda)q} = \log\frac{1}{\lambda + (1-\lambda) \min_{x\in\Xscr} q(x)} ,
\end{equation}
which is a strictly decreasing function of $\lambda\in(0,1]$.
\end{proposition}
\begin{proof}
Let $\lambda\in (0,1]$. For all $u,v\in [0,1]$, the condition $u>0$
implies $\lambda u + (1-\lambda)v >0$, hence
$$
s(u,\lambda u + (1-\lambda)v) = \begin{cases}\displaystyle u\log\frac{u}{\lambda u
+ (1-\lambda)v} &\text{if }0<u\leq 1 ,\\
0 &\text{if }u=0 .
\end{cases}
$$
Clearly, this is a continuous function of $u\in (0,1]$. To see that it is
continuous also at $0$, we take the limit
\begin{multline*}
\lim_{u\to 0^+} u\log\frac{u}{\lambda u + (1-\lambda)v}  = \lim_{u\to 0^+} u\log u
- \lim_{u\to 0^+} u\log[\lambda u + (1-\lambda)v] \\
{} = - \lim_{u\to 0^+} u\log[\lambda u + (1-\lambda)v]
= \begin{cases}\displaystyle 0 &\text{if } v\neq 0 ,\\
- \frac{1}{\lambda}\lim_{u\to 0^+} \lambda u\log(\lambda u) = 0 &\text{if } v = 0 .
\end{cases}
\end{multline*}
Since $g_\lambda(p) = \sum_x s\left(p(x),\lambda p(x) + (1-\lambda) q(x)\right)$, the
continuity of $g_\lambda$ then follows. Since $g_\lambda$ is also convex on
$\Pscr(\Xscr)$ by Proposition \ref{prop:propHSrel}, item (\ref{HSconvex}),
and the set $\Pscr(\Xscr)$ is compact, the function $g_\lambda$ takes its maximum
at some extreme point $\delta_x$ of $\Pscr(\Xscr)$. It follows that
\begin{multline*}
\sup_{p\in\Pscr(\Xscr)} \Srel{p}{\lambda p + (1-\lambda)q} = \max_{x\in\Xscr}
\Srel{\delta_x}{\lambda \delta_x + (1-\lambda)q}
\\ {} = \log\frac{1}{\lambda + (1-\lambda)\min_{x\in\Xscr} q(x)} .
\end{multline*}
Setting $q_{\rm min} = \min_{x\in\Xscr} q(x)$, the derivative in $\lambda$ of the
last expression is
\[
\frac{\rmd\ }{\rmd\lambda}\left(\log\frac{1}{\lambda + (1-\lambda) q_{\rm min}}\right) =
\frac{\left(q_{\rm min}-1\right)\log \rme}{(1-q_{\rm min})\lambda + q_{\rm min}}\, ,
\]
which is negative for all $\lambda\in (0,1]$ since $q_{\rm min}\leq 1/|\Xscr|<1$.
Thus, the right hand side of \eqref{eq:bAlb} is strictly decreasing in $\lambda$.
\end{proof}

Note that, if $\lambda = 0$, then $g_0(p)=\Srel{p}{q}$ is an extended real LSC
function on $\Pscr(\Xscr)$ by Proposition
\ref{prop:propHSrel}, item (\ref{HScont}). However, it is not difficult to show along the lines of
the previous proof that the maximum in \eqref{eq:bAlb} is still attained, and
$$
\max_{p\in\Pscr(\Xscr)}
\Srel{p}{q} = \begin{cases} \displaystyle \log\frac{1}{\min_x q(x)} & \text{ if ${\rm supp}\,q = \Xscr$,}
\\
+\infty & \text{ otherwise} .
\end{cases}
$$

\section{Entropic measurement uncertainty relations}\label{sec:indices}

In general, the two target  observables $\Ao$ and $\Bo$, introduced at the
end of Section \ref{Ob+In}, are incompatible, and only ``approximate'' joint
measurements are possible for them.
Moreover, any measurement of $\Ao$ may disturb a subsequent measurement of
$\Bo$, in a way that the resulting distribution of $\Bo$ can be very far from
its unperturbed version; this disturbance may be present even when the two
observables are compatible. Typically, such a disturbance of $\Ao$ on $\Bo$ can
not be removed, nor just made arbitrarily small, unless we drop the requirement
of exactly measuring $\Ao$. However, in both cases, the measurement
uncertainties on $\Ao$ and $\Bo$ can not always be made equally small. The
quantum nature of $\Ao$ and $\Bo$ relates their measurement uncertainties, so
that improving the approximation of $\Ao$ affects the quality of the
corresponding approximation of $\Bo$ and vice versa. Incompatibility of $\Ao$
and $\Bo$ on the one hand, and the disturbance induced on $\Bo$ by a
measurement of $\Ao$ on the other hand, are alternative manifestations of the
quantum relation between the two observables, and as such deserve
different approaches.

Our aim is now to quantify both these types of measurement uncertainty
relations between $\Ao$ and $\Bo$ by means of suitable informational
quantities. In the case of incompatible observables, we will find an \emph{entropic incompatibility degree}, encoding the minimum total error affecting
any approximate joint measurement of $\Ao$ and $\Bo$. Similarly, when the
observable $\Bo$ is measured after an approximate version of $\Ao$, the resulting
uncertainties on both observables will produce an \emph{error/disturbance tradeoff} for $\Ao$ and $\Bo$. In both cases, we will look
for an optimal bi-observable $\Mo$ whose marginals $\Mo_{[1]}$ and $\Mo_{[2]}$
are the best approximations of the two target observables $\Ao$ and $\Bo$.
However, the different points of view will be reflected in the fact that we will
optimize over $\Mo$ in two different sets, according to the case at hand.

\subsection{Error function and entropic divergence for observables}

We now regard any bi-observable $\Mo\in \Mscr(\Xscr\times\Yscr)$ as an
\emph{approximate joint measurement} of $\Ao$ and $\Bo$ and we want an
informational quantification of how far its marginals $\Mo_{[1]}$ and
$\Mo_{[2]}$ are from correctly approximating the two target observables
$\Ao$ and $\Bo$.
Following \cite{Wer04,BLW14a,BLW14b}, these two approximations will be judged
by comparing (within our entropic approach) the distribution $\Mo_{[1]}^\rho$
with $\Ao^\rho$, and the distribution $\Mo_{[2]}^\rho$ with $\Bo^\rho$,
for all states $\rho$. Note that we can not compare the output of
$\Mo_{[1]}$ with that of $\Ao$, and the output of $\Mo_{[2]}$ with that of
$\Bo$, in one and the same experiment. Indeed, although our bi-observable $\Mo$ is a
joint measurement of $\Mo_{[1]}$ and $\Mo_{[2]}$, there is no way to turn
it into a joint measurement of the four observables $\Ao$, $\Mo_{[1]}$,
$\Mo_{[2]}$ and $\Bo$, when $\Ao$ and $\Bo$ are not compatible.
Nevertheless,
even if $\Ao$ and $\Bo$ are incompatible, each of them can be measured in
independent repetitions of a preparation (state) $\rho$ of the system.
Similarly, any bi-observable $\Mo$ can be measured in other independent
repetitions of the same preparation. So, all the three probability
distributions $\Ao^\rho$, $\Bo^\rho$, $\Mo^\rho$ can be estimated from
independent experiments, and then they can be compared without any hypothesis of
compatibility among $\Ao$, $\Bo$ and $\Mo$.

The first step is to quantify the inefficiency of the distribution
approximations $\Ao^\rho\simeq\Mo_{[1]}^\rho$ and
$\Bo^\rho\simeq\Mo_{[2]}^\rho$, given the bi-observable $\Mo$. According to the
discussion in Section \ref{sec:c-ent}, the natural way to quantify the loss of
information in each approximation is to use the relative entropy. Remarkably, the relative entropy properties allow us to give a single quantification for the whole couple approximation $(\Ao^\rho,\Bo^\rho)\simeq(\Mo_{[1]}^\rho,\Mo_{[2]}^\rho)$: since $\Srel{\Ao^\rho}{\Mo_{[1]}^\rho}$ and $\Srel{\Bo^\rho}{\Mo_{[2]}^\rho}$ are homogeneous and dimensionless, they can be added to give the total amount of information loss.

\begin{definition}\label{def:erf}
For any bi-observable $\Mo\in\Mscr(\Xscr\times\Yscr)$, the \emph{error function} of the approximation $(\Ao,\Bo)$ $\simeq (\Mo_{[1]},\Mo_{[2]})$ is the state-dependent quantity
\begin{equation}\label{statedependentdivergence}
\sddiv{\Ao,\Bo}{\Mo}(\rho)=\Srel{\Ao^\rho}{ \Mo^\rho_{[1]}}+ \Srel{\Bo^\rho}{ \Mo^\rho_{[2]}}.
\end{equation}
\end{definition}

Note that the approximating distributions $\Mo^\rho_{[i]}$ appear in the second entry of the relative entropy, consistently with the discussion following its definition \eqref{eq:defS}.

By Proposition \ref{prop:propHSrel}, item (\ref{HSprod}), we can rewrite
\eqref{statedependentdivergence} in the form
\begin{equation}\label{eq:Soprod}
\sddiv{\Ao,\Bo}{\Mo}(\rho)= \Srel{\Ao^\rho\otimes \Bo^\rho}{ \Mo^\rho_{[1]}\otimes \Mo^\rho_{[2]}}.
\end{equation}
It is important to note that the error function itself is a relative entropy; this can be mathematically useful in some situations (see e.g.~the proof of Theorem \ref{prop:invarD^}). Note that, whether $\Ao$ and $\Bo$ are compatible or not, $\Ao^\rho\otimes \Bo^\rho$ is the distribution of their measurements in two independent preparations of the same state $\rho$.

The second step is to quantify the inefficiency of the observable
approximations $\Ao\simeq\Mo_{[1]}$ and $\Bo\simeq\Mo_{[2]}$ by means of the
marginals of a given bi-observable $\Mo$, without reference to any particular state. In order to construct a state-independent quantity, we take the worst
case in \eqref{statedependentdivergence} with respect to the system state
$\rho$.
\begin{definition}\label{def:D}
The \emph{entropic divergence of $\Mo\in \Mscr(\Xscr\times\Yscr)$ from
$(\Ao,\Bo)$} is the quantity
\begin{equation}\label{eqdef:D}
\Div{\Ao,\Bo}{\Mo}= \sup_{\rho\in\Sscr(\Hscr)} \sddiv{\Ao,\Bo}{\Mo}(\rho) \equiv \sup_{\rho\in\Sscr(\Hscr)}
\left\{\Srel{\Ao^\rho}{ \Mo^\rho_{[1]}}+ \Srel{\Bo^\rho}{ \Mo^\rho_{[2]}}\right\}.
\end{equation}
\end{definition}
The entropic divergence $\Div{\Ao,\Bo}{\Mo}$ quantifies the worst
total loss of information due to the couple approximation $(\Ao,\Bo)\simeq(\Mo_{[1]},\Mo_{[2]})$. Note that there is a unique supremum over $\rho$, so
that $\Div{\Ao,\Bo}{\Mo}$ takes into account any
possible balancing and compensation between the information losses in the first and in the
second approximation. The entropic divergence depends only on $\Mo_{[1]}$ and $\Mo_{[2]}$, and so it is the same for different bi-observables with
equal marginals. If $\Ao$ and $\Bo$ are compatible and $\Mo$ is any of
their joint measurements, then $\Div{\Ao,\Bo}{\Mo} = 0$ by Proposition
\ref{prop:propHSrel}, item (\ref{HSnull}).

\begin{theorem}\label{prop:2A}
Let $\Ao\in\Mscr(\Xscr)$, $\Bo\in\Mscr(\Yscr)$ be the target observables.
The error function and the entropic divergence defined above have the following properties.
\begin{enumerate}[(i)]
\item \label{slsc} The function
    $\sddiv{\Ao,\Bo}{\Mo}:\Sscr(\Hscr)\to[0,+\infty]$ is convex and LSC,
    $\forall \Mo\in \Mscr(\Xscr\times \Yscr)$.

\item \label{Dlsc} The function $\Div{\Ao,\Bo}{\cdot} : \Mscr(\Xscr\times
    \Yscr)\to [0,+\infty]$ is convex and LSC.

\item \label{Dfinite} For any $\Mo\in\Mscr(\Xscr\times\Yscr)$, the
    following three statements are equivalent:
\begin{enumerate}[(a)]
\item $\Div{\Ao,\Bo}{\Mo}<+\infty$,
\item $\operatorname{ker}\Mo_{[1]}(x)\subseteq\operatorname{ker}\Ao(x),\
    \forall x, \quad\text{and}\quad
    \operatorname{ker}\Mo_{[2]}(y)\subseteq\operatorname{ker}\Bo(y),\
    \forall y$,
\item $\sddiv{\Ao,\Bo}{\Mo}$ is bounded and continuous.
\end{enumerate}

\item \label{worstrho} $\displaystyle\Div{\Ao,\Bo}{\Mo}=
    \max_{\rho\in\Sscr(\Hscr),\ \rho\    \mathrm{pure}}
    \sddiv{\Ao,\Bo}{\Mo}(\rho) $, where the maximum can be any
    value in the extended interval $[0,+\infty]$.

\item \label{Sddivinv} The error $\sddiv{\Ao,\Bo}{\Mo}(\rho)$ is invariant under an overall unitary conjugation of $\Ao$, $\Bo$, $\Mo$ and $\rho$, and a relabelling of the outcome spaces $\Xscr$ and $\Yscr$.

\item \label{Dinv} The entropic divergence $\Div{\Ao,\Bo}{\Mo}$ is invariant under an overall unitary conjugation of $\Ao$, $\Bo$ and $\Mo$, and a relabelling of the outcome spaces $\Xscr$ and $\Yscr$.
\end{enumerate}
\end{theorem}
\begin{proof}
(\ref{slsc}) The function $\sddiv{\Ao,\Bo}{\Mo}$ is the sum of two terms which
are convex, because the mapping $\rho\mapsto \Xo^\rho$ is affine for any
observable $\Xo$ and by Proposition
\ref{prop:propHSrel}, item (\ref{HSconvex}); hence $\sddiv{\Ao,\Bo}{\Mo}$ is convex. Moreover, each
term is LSC, since $\rho\mapsto \Xo^\rho$ is continuous and because of Proposition \ref{prop:propHSrel}, item (\ref{HScont}); so the sum
$\sddiv{\Ao,\Bo}{\Mo}$ is LSC by \cite[Prop.~1.5.12]{Ped89}.

(\ref{Dlsc}) Each mapping $\Mo\mapsto\Mo_{[i]}^\rho$ is affine and continuous,
and the functions $\Srel{\Ao^\rho}{\cdot}$, $\Srel{\Bo^\rho}{\cdot}$ are convex
and LSC by Proposition \ref{prop:propHSrel}, items (\ref{HSconvex}) and
(\ref{HScont}). It follows that $\Mo\mapsto \Srel{\Ao^\rho}{\Mo_{[1]}^\rho}$
and $\Mo\mapsto \Srel{\Bo^\rho} {\Mo_{[2]}^\rho}$ are also convex and LSC
functions on $\Mscr(\Xscr\times\Yscr)$; hence, such are their sum and the
supremum $\Div{\Ao,\Bo}{\cdot}$ \cite[Prop.~1.5.12]{Ped89}.

(\ref{Dfinite}) Let us show (a)$\Rightarrow$(b)$\Rightarrow$(c)$\Rightarrow$(a).
\\
(a)$\Rightarrow$(b). If
$\operatorname{ker}\Mo_{[1]}(x)\nsubseteq\operatorname{ker}\Ao(x)$ for some
$x$, then we could take a pure state $\rho=\kb{\psi}{\psi}$ with $\psi$
belonging to $\operatorname{ker}\Mo_{[1]}(x)$ but not to
$\operatorname{ker}\Ao(x)$, so that $\Mo_{[1]}^\rho(x)=0$ while
$\Ao^\rho(x)>0$; thus, we would get $\Srel{\Ao^\rho}{
\Mo^\rho_{[1]}}=+\infty$ and the contradiction $\Div{\Ao,\Bo}{\Mo}=+\infty$.
\\
(b)$\Rightarrow$(c). The function $\sddiv{\Ao,\Bo}{\Mo}$ is a finite sum of terms
of the kind $ s\big(\Ao^\rho(x), \Mo^\rho_{[1]}(x)\big)$ or
$s\big(\Bo^\rho(y),\Mo^\rho_{[2]}(y)\big)$, where $s$ is the function
defined in \eqref{eq:defs}. Under the hypothesis (b), each of these terms is a
bounded and continuous function of $\rho$ by Lemma \ref{lem:conts} below. We
thus conclude that $\sddiv{\Ao,\Bo}{\Mo}$ is bounded and continuous.
\\
(c)$\Rightarrow$(a). Trivial, as $\Div{\Ao,\Bo}{\Mo} =
\sup_{\rho\in\sh}\sddiv{\Ao,\Bo}{\Mo}(\rho)$.

(\ref{worstrho}) If $\Div{\Ao,\Bo}{\Mo}<+\infty$, then $\sddiv{\Ao,\Bo}{\Mo}$
is a bounded and continuous function on the compact set $\Sscr(\Hscr)$ by item
(\ref{Dfinite}) above, and thus it attains a maximum; moreover,
$\sddiv{\Ao,\Bo}{\Mo}$ is convex, hence it has at least a maximum point among
the extreme points of $\Sscr(\Hscr)$, which are the pure states. If instead
$\Div{\Ao,\Bo}{\Mo}=+\infty$, then
$\operatorname{ker}\Mo_{[1]}(x)\nsubseteq\operatorname{ker}\Ao(x)$ for some
$x$, or $\operatorname{ker}\Mo_{[2]}(y)\nsubseteq\operatorname{ker}\Bo(y)$ for
some $y$ again by item (\ref{Dfinite}). In this case, every pure state
$\rho=\kb{\psi}{\psi}$ with
$\psi\in\operatorname{ker}\Mo_{[1]}(x)\setminus\operatorname{ker} \Ao(x)$, or
$\psi\in\operatorname{ker}\Mo_{[2]}(y)\setminus\operatorname{ker}\Bo(y)$, is
such that $\sddiv{\Ao,\Bo}{\Mo}(\rho) = +\infty$, and thus it is a maximum point
of $\sddiv{\Ao,\Bo}{\Mo}$.

(\ref{Sddivinv}) For any unitary operator $U$ on $\Hscr$, we have $(U^*\Ao U)^{U^*\rho U} = \Ao^\rho$, $(U^*\Bo U)^{U^*\rho U} = \Bo^\rho$, and, since $(U^*\Mo U)_{[i]} = U^*\Mo_{[i]}U$, also $(U^*\Mo U)_{[i]}^{U^*\rho U} = \Mo_{[i]}^\rho$. Therefore, by the definition \eqref{statedependentdivergence} of the error function, we get the equality
$$
\sddiv{U^*\Ao U,\,U^*\Bo U}{U^*\Mo U}(U^*\rho U) = \sddiv{\Ao,\Bo}{\Mo}(\rho).
$$
The invariance under relabelling of the outcomes is an immediate consequence of the analogous property of the relative entropy (Proposition \ref{prop:propHSrel}, item (\ref{HSinv})).

(\ref{Dinv}) The two invariances immediately follow by the previous item. We check only the first one:
\begin{align*}
& \Div{U^*\Ao U,\, U^*\Bo U}{U^*\Mo U} = \sup_{\rho\in\sh} \sddiv{U^*\Ao U,\,U^*\Bo U}{U^*\Mo U}(\rho) \\
& \qquad \qquad = \sup_{\rho\in\sh} \sddiv{U^*\Ao U,\,U^*\Bo U}{U^*\Mo U}(U^*\rho U) = \sup_{\rho\in\sh} \sddiv{\Ao,\Bo}{\Mo}(\rho) \\
& \qquad \qquad = \Div{\Ao,\Bo}{\Mo} ,
\end{align*}
where in the second equality we have used the fact that $U\sh U^* = \sh$.
\end{proof}

An essential step in the last proof is the following lemma.
\begin{lemma}\label{lem:conts}
Suppose $A,B\in\Lscr(\Hscr)$ are such that $0\leq A\leq\id$ and $0\leq
B\leq\id$, and assume that $\ker B\subseteq\ker A$. Let $A^\rho = \Tr\{A\rho\}$, $B^\rho = \Tr\{B\rho\}$, and let $s$ be the function defined in \eqref{eq:defs}. Then, the function $s_{A,B}
: \sh\to [0,+\infty]$, with
$s_{A,B}(\rho) = s(A^\rho,B^\rho)$, is bounded and continuous.
\end{lemma}
\begin{proof}
We will show that $s_{A,B}$ is a continuous function on $\sh$; since $\sh$ is
compact, this will also imply that $s_{A,B}$ is bounded. The case $B=0$ is
trivial, hence we will suppose $B\neq 0$. By the hypotheses, the condition
$B^\rho = 0$ implies that $A^\rho = 0$. The definition \eqref{eq:defs} of $s$
then gives
$$
s_{A,B}(\rho) = \begin{cases}\displaystyle A^\rho\log\frac{A^\rho}{B^\rho}
&\text{if } A^\rho > 0 \text{ and } B^\rho > 0 \\
0 &\text{if } A^\rho = 0 \text{ and } B^\rho > 0 \\
0 &\text{if } A^\rho = 0 \text{ and } B^\rho = 0
\end{cases}
\ = \ \begin{cases}\displaystyle B^\rho h\left(\frac{A^\rho}{B^\rho}\right) &\text{if } B^\rho > 0\\
0 &\text{if } B^\rho = 0
\end{cases}
$$
where we have introduced the continuous function $h:[0,+\infty)\to [-(\log\rme)/\rme,+\infty)$,
with $ h(t) =  t\log t$ if $ t > 0$, and $h(0)=0$.
The function $s_{A,B}$ is clearly continuous on the open subset $\Uscr =
\{\rho\in\sh : B^\rho > 0\}$ of the state space $\sh$. It remains to show that
it is also continuous at all the points of the set $\Uscr^{\rm c} =
\{\rho\in\sh : B^\rho = 0\}$. To this aim, observe that
\[
A\leq c_{\rm max}(A) P_A \leq c_{\rm max}(A) P_B \qquad \text{and} \qquad B\geq c_{\rm min}(B) P_B ,
\]
where $c_{\rm max}(A)$ is the maximum eigenvalue of $A$, $c_{\rm min}(B)$ is
the minimum positive eigenvalue of $B$, and we denote by $P_A$ and $P_B$ the
orthogonal projections onto $\ker A^\perp$ and $\ker B^\perp$, respectively.
Since $P_B^\rho \neq 0$ for all $\rho$ such that $B^\rho>0$, it follows that
$$
0\leq \frac{A^\rho}{B^\rho}\leq \frac{c_{\rm max}(A)}{c_{\rm min}(B)}\,, \qquad \forall\rho\in \Uscr .
$$
Hence, by continuity of $h$ and boundedness of the interval $[0,c_{\rm
max}(A)/c_{\rm min}(B)]$, there is a constant $M>0$ such that
$$
|s_{A,B}(\rho)| = \left|B^\rho h\left(\frac{A^\rho}{B^\rho}\right)\right|
\leq M B^\rho ,\qquad \forall\rho\in \Uscr .
$$
On the other hand, for $\rho\in \Uscr^c$ we have $s_{A,B}(\rho) = 0$. If
$(\rho_k)_k$ is a sequence in $\sh$ converging to $\rho_0\in \Uscr^c$, then
$|s_{A,B}(\rho_k)-s_{A,B}(\rho_0)|\leq M B^{\rho_k}
\displaystyle\operatornamewithlimits{\longrightarrow}_{k\to\infty} 0$, which
shows that $s_{A,B}$ is continuous at $\rho_0$. \end{proof}

\subsection{Incompatibility degree, error/disturbance coefficient, and
optimal approximate joint measurements}\label{sef:def_c}

After introducing the error function $\sddiv{\Ao,\Bo}{\Mo}(\rho)$, which
describes the total information lost by measuring the bi-observable $\Mo$ in
place of $\Ao$ and $\Bo$ in the state $\rho$, and after defining its maximum
value $\Div{\Ao,\Bo}{\Mo}$ over all states, the third step is to quantify
the intrinsic measurement uncertainties between $\Ao$ and $\Bo$, dropping any
reference to a particular state or approximating joint measurement. When we are
interested in incompatibility, this is done by taking the minimum of the
divergence $\Div{\Ao,\Bo}{\Mo}$ over all possible bi-observables
$\Mo\in\Mscr(\Xscr\times\Yscr)$. The resulting quantity is the minimum
inefficiency which can not be avoided when the (possibly incompatible)
observables $\Ao$ and $\Bo$ are approximated by the compatible marginals
$\Mo_{[1]}$ and $\Mo_{[2]}$ of any bi-observable $\Mo$. This minimum can be
understood as an ``incompatibility degree'' of the two observables $\Ao$ and
$\Bo$.

\begin{definition}\label{def:Cinc} The \emph{entropic incompatibility degree}
$\Icomp(\Ao,\Bo)$ of the observables $\Ao$ and $\Bo$ is
\begin{equation}\label{def:inc}
\Icomp(\Ao,\Bo)=\inf_{\Mo\in\Mscr(\Xscr\times \Yscr)} \Div{\Ao,\Bo}{\Mo}
\equiv \inf_{\Mo\in\Mscr(\Xscr\times \Yscr)}\sup_{\rho\in\Sscr(\Hscr)} \sddiv{\Ao,\Bo}{\Mo}(\rho) .
\end{equation}
\end{definition}

The definition is consistent, as obviously $\Icomp(\Ao,\Bo)\geq 0$, and
$\Icomp(\Ao,\Bo) = 0$ when $\Ao$ and $\Bo$ are compatible.
As the notion of incompatibility is symmetric by exchanging the observables
$\Ao$ and $\Bo$, we would expect that also the incompatibility degree satisfies
the property $\Icomp(\Ao,\Bo) = \Icomp(\Bo,\Ao)$. Indeed, this is actually
true, as $\Div{\Ao,\Bo}{\Mo} = \Div{\Bo,\Ao}{\Mo'}$ for all
$\Mo\in\Mscr(\Xscr\times\Yscr)$, where $\Mo'\in\Mscr(\Yscr\times\Xscr)$ is
defined by $\Mo'(y,x) = \Mo(x,y)$. Note that the symmetry of $\Icomp$ comes
from the fact that, in defining the error function $\sddiv{\Ao,\Bo}{\Mo}$, we have
chosen equal weights for the contributions of the two approximation errors of
$\Ao$ and $\Bo$.

On the other hand, when we deal with the error/disturbance uncertainty
relation, our analysis is restricted to the bi-observables describing
sequential measurements of an approximate version $\Ao'$ of $\Ao$, followed by
an exact measurement of $\Bo$. In other words, we focus on
\begin{multline}\label{def:seq}
\Mscr(\Xscr;\Bo)  = \{\Jcal^*(\Bo) : \Jcal\in\Jscr(\Xscr)\} \\
 = \{\Mo\in\Mscr(\Xscr\times\Yscr) : \Mo(x,y)
= \Jcal^*_x[\Bo(y)] \ \forall x,y, \ \text{for some} \ \Jcal\in\Jscr(\Xscr)\},
\end{multline}
the subset of
$\Mscr(\Xscr\times\Yscr)$ consisting of the sequential measurements where the first outcome set $\Xscr$
and the second observable $\Bo$ are fixed. If $\Mo =
\Jcal^*(\Bo)\in\Mscr(\Xscr;\Bo)$, then $\Ao'=\Mo_{[1]} = \Jcal^*_\cdot[\id]$ is
the observable approximating $\Ao$, and $\Bo'=\Jcal^*_\Xscr[\Bo(\cdot)]$ is the
version of $\Bo$ perturbed by the measurement of $\Ao'$. In
general, it may equally well be $\Ao'\neq\Ao$ and $\Bo'\neq\Bo$, unless the
observable $\Ao$ can be measured without disturbing $\Bo$ \cite{HeiW10}.

In order to quantify the measurement uncertainties due to the error/disturbance
trade\-off, we then consider the minimum of the entropic divergence
$\Div{\Ao,\Bo}{\Mo}$ for $\Mo\in\Mscr(\Xscr;\Bo)$. If we read $\Srel{\Ao^\rho}{
\Mo^\rho_{[1]}}$ as the error made by $\Jcal$ in measuring $\Ao$ in the state
$\rho$, and $\Srel{\Bo^\rho}{ \Mo^\rho_{[2]}}$ as the amount of disturbance
introduced by $\Jcal$ on the subsequent measurement of $\Bo$, then the
divergence $\Div{\Ao,\Bo}{\Mo}$ expresses the sum error $+$ disturbance
maximized over all states for the sequential measurement $\Mo$. Minimizing
$\Div{\Ao,\Bo}{\Mo}$ over all sequential measurements, we then obtain the
following entropic quantification of the error/disturbance tradeoff
between $\Ao$ and $\Bo$.

\begin{definition}\label{def:Ced}
The \emph{entropic error/disturbance coefficient}  $\Iad(\Ao,\Bo)$ of $\Ao$
followed by $\Bo$ is
\begin{equation}\label{def:ed}
\Iad(\Ao,\Bo)=\inf_{\Mo\in\Mscr(\Xscr;\Bo)} \Div{\Ao,\Bo}{\Mo} \equiv \inf_{\Mo\in\Mscr(\Xscr;\Bo)} \sup_{\rho\in\Sscr(\Hscr)} \sddiv{\Ao,\Bo}{\Mo}(\rho).
\end{equation}
\end{definition}

Similarly to the incompatibility degree, the error/disturbance coefficient is
always nonnegative, and $\Iad(\Ao,\Bo)=0$ when $\Ao$ can be measured without
disturbing $\Bo$, i.e.\ $\Ao$ and $\Bo$ are sequentially compatible. Contrary to $\Icomp$, we stress that in general the two
indexes $\Iad(\Ao,\Bo)$ and $\Iad(\Bo,\Ao)$ can be different, as shown in
Remark \ref{ex:12} below.

When the approximate measurement of the first observable $\Ao$ is described by
the instrument $\Jcal$, the measurement of the second fixed observable $\Bo$
could be preceded by any kind of correction taking into account the observed
outcome $x$ \cite{BLW14a}. This can be formalized by inserting a quantum
channel $\Ccal_x$ in between the measurements of $\Ao$ and $\Bo$. As the
composition  $\Jcal'_x = \Ccal_x\circ\Jcal_x$ gives again an instrument
$\Jcal'\in\Jscr(\Xscr)$, we then see that any possible correction is considered
when we take the infimum in $\Mscr(\Xscr;\Bo)$. The latter fact shows that
Definition \ref{def:Ced} is consistent, since only by taking into account all
possible corrections we can properly speak of pure unavoidable disturbance and
of error/disturbance tradeoff.

Comparing the two indexes $\Icomp$ and $\Iad$, the inequality
$\Icomp(\Ao,\Bo)\leq \Iad(\Ao,\Bo)$ trivially follows from the inclusion
$\Mscr(\Xscr;\Bo)\subseteq\Mscr(\Xscr\times\Yscr)$. This means that, even if
one is interested in $\Iad$, the most symmetric index $\Icomp$ is at least a
lower bound for it.
We stress that the inclusion $\Mscr(\Xscr;\Bo)\subseteq\Mscr(\Xscr\times
\Yscr)$ may be strict in general. For example, there may exist observables
which are compatible with $\Bo$, but can not be measured before $\Bo$ without
disturbing it. Then, taken such an observable $\Ao$, a joint measurement of
$\Ao$ and $\Bo$ clearly belongs to $\Mscr(\Xscr\times\Yscr)$ but can not be in
$\Mscr(\Xscr;\Bo)$. When $\Mscr(\Xscr;\Bo)\subsetneq\Mscr(\Xscr\times\Yscr)$,
the incompatibility and error/disturbance approaches definitely are not
equivalent. Nevertheless, there is one remarkable situation in which they are the same.

\begin{proposition}\label{prop:sharp}
If $\Bo\in\Mscr(\Yscr)$ is a sharp observable, then $\Mscr(\Xscr;\Bo) =
\Mscr(\Xscr\times\Yscr)$.
\end{proposition}
\begin{proof}
The proof directly follows from the argument at the end of
\cite[Sect.~II.D]{HeiW10}. Indeed, for any $\Mo\in\Mscr(\Xscr\times\Yscr)$, we
can define the instrument $\Jcal\in\Jscr(\Xscr)$ with
\[
\Jcal_x[\rho] = \sum_{y\in\Yscr:\Bo(y)\neq 0}
\Tr\left\{\rho\Mo(x,y)\right\} \frac{\Bo(y)}{\Tr\left\{\Bo(y)\right\}} .
\]
For such an instrument, the equality $\Mo(x,y) = \Jcal^*_x[\Bo(y)]$ is
immediate. \end{proof}

As an immediate consequence of this result, we have
$\Iad(\Ao,\Bo) = \Icomp(\Ao,\Bo)$ when\-ever the second measured observable $\Bo$
is sharp.

By Theorem \ref{prop:2B} below, the two infima in the definitions of $\Icomp(\Ao,\Bo)$ and $\Iad(\Ao,\Bo)$ are actually two minima. It is convenient to give a name to the corresponding sets
of minimizing bi-observables:
\begin{equation*}
\Mcomp(\Ao,\Bo) = \argmin_{\Mo\in\Mscr(\Xscr\times \Yscr)} \Div{\Ao,\Bo}{\Mo},
\qquad \Mad(\Ao,\Bo) = \argmin_{\Mo\in\Mscr(\Xscr;\Bo)} \Div{\Ao,\Bo}{\Mo}.
\end{equation*}
We can say that $\Mcomp(\Ao,\Bo)$ is the set of the {\em optimal approximate
joint measurements of $\Ao$ and $\Bo$}. Similarly, $\Mad(\Ao,\Bo)$ contains the
sequential measurements optimally approximating  $\Ao$ and $\Bo$.

The next theorem summarizes the main properties of $\Icomp$ and $\Iad$
contained in the above discussion, and states some further relevant facts about
the two indexes.

\begin{theorem}\label{prop:2B}
Let $\Ao\in\Mscr(\Xscr)$, $\Bo\in\Mscr(\Yscr)$ be the target observables.
For the entropic coefficients defined above the following properties hold.
\begin{enumerate}[(i)]
\item \label{cinv} The coefficients $\Icomp(\Ao,\Bo)$ and $\Iad(\Ao,\Bo)$
    are invariant under an overall unitary conjugation of the observables
    $\Ao$ and $\Bo$, and they do not depend on the labelling of the outcomes in
    $\Xscr$ and $\Yscr$.
		
\item \label{cincsymm} The incompatibility degree has the exchange symmetry
    $\Icomp(\Ao,\Bo)=\Icomp(\Bo,\Ao)$.

\item \label{cbound} \ \ We have \ \ \ $\displaystyle0\leq \Icomp(\Ao,\Bo)\leq \Iad(\Ao,\Bo)\leq
    \log \abs\Xscr-\inf_{\rho\in\Sscr(\Hscr)}H(\Ao^\rho)$ \ \ and \ \ \ \goodbreak
    $\displaystyle\Icomp(\Ao,\Bo)\leq \log
    \abs\Yscr-\inf_{\rho\in\Sscr(\Hscr)}H(\Bo^\rho)$.

\item \label{optM} The sets $\Mcomp(\Ao,\Bo)$ and $\Mad(\Ao,\Bo)$ are
    nonempty convex compact subsets of $\Mscr(\Xscr\times\Yscr)$.

\item \label{jmeas} $\Icomp(\Ao,\Bo)=0$ if and only if the observables
    $\Ao$ and $\Bo$ are compatible, and in this case
    $\Mcomp(\Ao,\Bo)$ is the set of all their joint measurements.

\item \label{smeas} $\Iad(\Ao,\Bo)=0$ if and only if the observables $\Ao$
    and $\Bo$ are sequentially compatible, and in this case $\Mad(\Ao,\Bo)$ is the set of all the sequential measurements of $\Ao$ followed by $\Bo$.

\item \label{sharpc} If $\Bo$ is sharp, then
   $\Mcomp(\Ao,\Bo)=\Mad(\Ao,\Bo)$ and  $\Icomp(\Ao,\Bo)=\Iad(\Ao,\Bo)$.
\end{enumerate}
\end{theorem}
\begin{proof}
(\ref{cinv}) The invariance under unitary conjugation follows from the
corresponding property of the entropic divergence (Theorem \ref{prop:2A}, item
(\ref{Dinv})). We will prove it only for $\Iad$, the case of $\Icomp$ being
even simpler. We have
\[ 
\Iad(U^*\Ao U,U^*\Bo U)=\inf_{\Mo \in \Mscr(\Xscr;U^*\Bo U)} \Div{U^*\Ao U,
U^*\Bo U}{\Mo} = \inf_{\Mo' \in U\Mscr(\Xscr;U^*\Bo U)U^*}\Div{\Ao,\Bo}{\Mo'} ,
\]
and, in order to show that $\Iad(U^*\Ao U,U^*\Bo U)=\Iad(\Ao,\Bo)$, it only
remains to prove the set equality $U\Mscr(\Xscr;U^*\Bo U)U^* =
\Mscr(\Xscr;\Bo)$. If $\Mo = \Jcal^*(U^*\Bo U) \in \Mscr(\Xscr;U^*\Bo U)$,
then, defining the instrument $\Jcal'_x[\rho] = U\Jcal_x[U^*\rho U]U^*$,
$\forall\rho,x$, we have $U\Mo U^* = \Jcal^{\prime *}(\Bo)\in\Mscr(\Xscr;\Bo)$,
as claimed. In a similar way, the invariance under relabelling of the outcomes
is a consequence of the analogous property of the entropic divergence.

(\ref{cincsymm}) This property has already been noticed.

(\ref{cbound}) The positivity and the inequality between the two indexes have already been noticed. Then, let
$\Ucal\in\Jscr(\Xscr)$ be the trivial uniform instrument $\Ucal_x[\rho] =
u_{\Xscr}(x)\rho$. Taking the sequential
measurement $\Ucal^*(\Bo)\in\Mscr(\Xscr;\Bo )$, we get $\Ucal^*(\Bo)^\rho =
u_{\Xscr}\otimes\Bo^\rho$ and
\[
\Srel{\Ao^\rho}{\Ucal^*(\Bo)^\rho_{[1]}}+\Srel{\Bo^\rho}{\Ucal^*(\Bo)^\rho_{[2]}}
= \Srel{\Ao^\rho}{u_{\Xscr}}= \log \abs\Xscr - H\big(\Ao^\rho\big),
\]
where the last equality follows from Proposition \ref{prop:propHSrel}, item (\ref{HSrel}). By taking the supremum over all the states, we get \
$\Div{\Ao,\Bo}{\Ucal^*(\Bo)}= \log
\abs\Xscr-\inf_{\rho\in\Sscr(\Hscr)}H(\Ao^\rho)$, \ hence \ $\Iad(\Ao,\Bo)\leq \log
\abs\Xscr-\inf_{\rho\in\Sscr(\Hscr)}H(\Ao^\rho)$ by definition. The last
inequality then follows by item (\ref{cincsymm}).

(\ref{optM}) By item (\ref{Dlsc}) of Theorem \ref{prop:2A} and item
(\ref{cbound}) just above, $\Div{\Ao,\Bo}{\cdot}$ is a convex LSC proper (i.e.\
not identically $+\infty$) function on the compact set
$\Mscr(\Xscr\times\Yscr)$. This implies that $\Mcomp(\Ao,\Bo) \neq \emptyset$
\cite[Exerc. E.1.6]{Ped89}. Closedness and convexity of $\Mcomp(\Ao,\Bo)$ are
then easy and standard consequences of $\Div{\Ao,\Bo}{\cdot}$ being convex and
LSC. On the other hand, the set $\Mscr(\Xscr;\Bo)$ is a convex and compact
subset of $\Mscr(\Xscr\times\Yscr)$; indeed, this follows from convexity and
compactness of $\Jscr(\Xscr)$ and continuity of the mapping $\Jcal\mapsto
\Jcal^*(\Bo)$ in the definition \eqref{def:seq}. The proof that the subset
$\Mad(\Ao,\Bo)\subseteq\Mscr(\Xscr;\Bo)$ is nonempty, convex and compact then
follows along the same lines of $\Mcomp(\Ao,\Bo)$.

(\ref{jmeas}) Assume $\Icomp(\Ao,\Bo) = 0$. Then $\Mcomp(\Ao,\Bo)$ exactly
consists of all the joint measurements of $\Ao$ and $\Bo$, which therefore turn
out to be compatible, as $\Mcomp(\Ao,\Bo)\neq\emptyset$ by (\ref{optM}).
Indeed, if $\Mo\in\Mcomp(\Ao,\Bo)$, then
$0=\Icomp(\Ao,\Bo)=\Div{\Ao,\Bo}{\Mo}$, which gives
$\Srel{\Ao^\rho}{\Mo_{[1]}^\rho} = \Srel{\Bo^\rho}{\Mo_{[2]}^\rho}=0$ for all
$\rho$. By Proposition \ref{prop:propHSrel}, item (\ref{HSnull}), this yields
$\Ao^\rho = \Mo_{[1]}^\rho$, $\Bo = \Mo_{[2]}^\rho$, $\forall \rho$, and so
$\Ao = \Mo_{[1]}$, $\Bo = \Mo_{[2]}$, which means that $\Mo$ is a joint
measurement of $\Ao$ and $\Bo$. The converse implication was already noticed in
the text.

(\ref{smeas}) Similarly to the previous item, if $\Iad(\Ao,\Bo) = 0$, then
$\Mad(\Ao,\Bo)$ consists exactly of all the sequential measurements of $\Ao$ followed by $\Bo$. Indeed, by the same argument of (\ref{jmeas}), if
$\Mo\in\Mad(\Ao,\Bo)$, then $\Mo$ is a joint measurement of $\Ao$ and $\Bo$;
since $\Mad(\Ao,\Bo)\subseteq\Mscr(\Xscr;\Bo)$, such a $\Mo$
is also a sequential measurement. As $\Mad(\Ao,\Bo)\neq\emptyset$ by
(\ref{optM}), this proves that $\Ao$ and $\Bo$ are sequentially compatible.
The other implication is trivial and was already remarked.

(\ref{sharpc}) As observed above, if $\Bo$ is sharp, then by Proposition
\ref{prop:sharp} we have $\Mscr(\Xscr;\Bo) = \Mscr(\Xscr\times\Yscr)$, which
implies the claim. \end{proof}

Item (\ref{cbound}) implies that the two indexes $\Icomp$
and $\Iad$ are always finite, although the relative entropy $\Srel pq$ is infinite whenever ${\rm supp}\,q\nsupseteq{\rm supp}\,p$. Actually, such a feature of $S$ has a role: because of Theorem \ref{prop:2A}, item (\ref{Dfinite}), a bi-observable $\Mo$ is immediately discarded as a very bad approximation of $\Ao$ and $\Bo$ whenever $\operatorname{ker}\Mo_{[1]}(x)\nsubseteq\operatorname{ker}\Ao(x)$ for some $x$, or $\operatorname{ker}\Mo_{[2]}(y)\nsubseteq\operatorname{ker}\Bo(y)$ for some $ y$.

We see in items (\ref{jmeas}) and (\ref{smeas}) that $\Icomp$
and $\Iad$ have the desirable feature of being zero exactly when the two
observables $\Ao$ and $\Bo$ satisfy the corresponding compatibility or
nondisturbance property. We also stress that, by their very definitions,
$\Icomp(\Ao,\Bo)$ and $\Iad(\Ao,\Bo)$ are independent of both the preparations
$\rho$ and the approximating bi-observables $\Mo$, as well as they
satisfy the natural invariance properties of item (\ref{cinv}). In view of these facts, we are allowed once more to say that the two bounds
$\Icomp(\Ao,\Bo)$ and $\Iad(\Ao,\Bo)$ are proper quantifications of the
intrinsic incompatibility and error/disturbance affecting the two observables
$\Ao$ and $\Bo$.

We stress that the definitions of $\Icomp(\Ao,\Bo)$
and $\Iad(\Ao,\Bo)$ are rather implicit. Indeed, even if we proved that they
are strictly positive when $\Ao$ and $\Bo$ are incompatible (or sequentially
incompatible), their evaluation requires the two optimizations ``sup'' on the
states and ``inf'' on the measurements. Nevertheless, in some cases explicit
computations are possible (even including the evaluation of the optimal
approximate joint measurements) or explicit lower bounds can be exhibited, see
Sections \ref{ex:D2} and \ref{sec:cov.MUBs}.

\begin{remark}\label{ex:12}
Item (\ref{sharpc}) of Theorem \ref{prop:2B} says that the two indexes
coincide in the important case in which $\Bo$ is sharp. However, this is not
true in general, as shown e.g.~by the two examples in Appendix \ref{app:HW}
(taken from \cite{HeiW10}). In the first example, $\dim \Hscr=3$,
$\abs\Xscr=2$, $\abs\Yscr=5$, and we have
$\Iad(\Ao,\Bo)>\Iad(\Bo,\Ao)=\Icomp(\Ao,\Bo)=0$. The second example is more
symmetric and simpler ($\abs\Xscr=\abs\Yscr=2$), and it yields
$\Iad(\Ao,\Bo)>\Icomp(\Ao,\Bo)=0$ and also $\Iad(\Bo,\Ao)>0$.
\end{remark}

\subsection{Entropic MURs}\label{sec:MUR}
By definition, the two coefficients \eqref{def:inc} and \eqref{def:ed} are
lower bounds for the entropic divergence \eqref{eqdef:D} of every bi-observable
$\Mo$ from $(\Ao,\Bo)$:
\begin{equation}\label{MUR1}\begin{split}
&\Div{ \Ao,\Bo}{\Mo}\geq \Icomp(\Ao,\Bo), \quad \forall\Mo\in \Mscr(\Xscr \times \Yscr);
\\ &\Div{\Ao,\Bo}{\Mo}\geq \Iad(\Ao,\Bo), \quad \forall \Mo\in \Mscr(\Xscr;\Bo).
\end{split}
\end{equation}
By items (\ref{jmeas}) and (\ref{smeas}) of Theorem \ref{prop:2B},
the two inequalities are non trivial and, by item (\ref{optM}), both bounds are
tight. As $\Div{ \Ao,\Bo}{\Mo}$ is a state-independent quantification of the inefficiency of the observable approximations $\Ao\simeq\Mo_{[1]}$ and $\Bo\simeq\Mo_{[2]}$, the inequalities \eqref{MUR1} are two state-independent formulations of entropic MURs.

Since the definition of $\Div{ \Ao,\Bo}{\Mo}$ involves a unique supremum over $\rho$, by Theorem \ref{prop:2A}, item (\ref{worstrho}), we can also reformulate the entropic MURs \eqref{MUR1} as statements about the total loss of information that occurs in one preparation of the system:
\begin{equation}\label{MUR2}\begin{split}
\forall \Mo\in \Mscr(\Xscr\times\Yscr), \ \
\exists \rho\in \Sscr(\Hscr) : \ \
&\Srel{\Ao^{\rho}}{\Mo^{\rho}_{[1]}}+\Srel{\Bo^{\rho}}{\Mo^{\rho}_{[2]}}\geq \Icomp(\Ao,\Bo);
\\
\forall \Mo\in \Mscr(\Xscr;\Bo), \ \
\exists \rho\in \Sscr(\Hscr) : \ \
&\Srel{\Ao^{\rho}}{\Mo^{\rho}_{[1]}}+\Srel{\Bo^{\rho}}{\Mo^{\rho}_{[2]}}\geq \Iad(\Ao,\Bo).
\end{split}
\end{equation}
So, in an approximate joint measurement of $\Ao$ and $\Bo$, the total loss of
information can not be arbitrarily reduced: it depends on the state $\rho$, but
potentially it can be as large as $\Icomp(\Ao,\Bo)$. Similarly, in a sequential
measurement of $\Ao$ and $\Bo$, there is a tradeoff between the information
lost in the first measurement (because of the approximation error) and the
information lost in the second measurement (because of the disturbance): they
both depend on the state $\rho$, but potentially their sum can be as large as
$\Iad(\Ao,\Bo)$.

The indexes $\Icomp(\Ao,\Bo)$ and $\Iad(\Ao,\Bo)$ are state-independent by their very definitions; however, the corresponding MURs \eqref{MUR2} only refer to the worst possible state $\rho$ for the measurement $\Mo$ at hand. Such a state-dependency is a general feature of MURs \cite[Sect.\ C]{BLW14a}: no MUR can provide a non trivial bound for the error of the approximation $(\Ao^\rho,\Bo^\rho) \simeq (\Mo_{[1]}^\rho,\Mo_{[2]}^\rho)$, holding for all states $\rho$ in any approximate joint measurement $\Mo$. Indeed, for a fixed $\rho \in \Sscr(\Hscr)$, the trivial bi-observable $\Mo(x,y)=\Ao^{\rho}(x) \Bo^{\rho}(y)\id$ gives $(\Ao^\rho,\Bo^\rho) = (\Mo_{[1]}^\rho,\Mo_{[2]}^\rho)$; hence, it perfectly approximates the target observables in the state $\rho$ whatever criterion one chooses for defining the error.

Here, in some detail, let us compare our MURs with Busch, Lahti and Werner's approach based on Wasserstein (or transport) distances (in the following, BLW approach; see \cite{Wer04,BLW14a,BLW14b}).
As for BLW, our starting point is just giving a quantification of the error in the distribution approximation $\Ao^\rho\simeq\Mo_{[1]}^\rho$ (or $\Bo^\rho\simeq\Mo_{[2]}^\rho$).
Anyway, employing the relative entropy in place of a Wasserstein distance reflects a different point of view, with some immediate consequences.
BLW use a Wasserstein distance $d(\Ao^{\rho},\Mo^{\rho}_{[1]})$ because they want that the error reflects the metric structure of the underlying outputs $\Xscr$; since the units of measurement of $\Xscr$ and $\Yscr$ may not be homogeneous, this essentially leads to  quantifying the error of the
whole couple approximation $(\Ao^{\rho},\Bo^{\rho})\simeq(\Mo_{[1]}^{\rho},\Mo_{[2]}^{\rho})$
with the dimensional pair $\big(d(\Ao^{\rho},\Mo^{\rho}_{[1]}),d(\Bo^{\rho},\Mo^{\rho}_{[2]})\big)$. On the contrary, the relative entropy is homogeneous and scale invariant; thus, it allows us to quantify the error of the couple approximation $(\Ao^\rho,\Bo^\rho)\simeq(\Mo_{[1]}^\rho,\Mo_{[2]}^\rho)$ with the single, dimensionless and scalar total error $\Srel{\Ao^\rho}{ \Mo^\rho_{[1]}}+ \Srel{\Bo^\rho}{ \Mo^\rho_{[2]}}$.

A second difference arises in the quantification of the inefficiency of the observable approximations $\Ao\simeq\Mo_{[1]}$ and $\Bo\simeq\Mo_{[2]}$.
The BLW approach naturally leads to using the two \emph{deviations} $d(\Ao,\Mo_{[1]}) = \sup_{\rho} d(\Ao^{\rho},\Mo^{\rho}_{[1]})$ and $d(\Bo,\Mo_{[2]}) = \sup_{\rho}d(\Bo^{\rho},\Mo^{\rho}_{[2]})$, that is, the dimensional couple $\big(d(\Ao,\Mo_{[1]}),d(\Bo,\Mo_{[2]})\big)$.
Instead, the entropic approach gives the entropic divergence $\Div{\Ao,\Bo}{\Mo}$ as a natural, dimensionless and scalar measure of the approximation inefficiency.

Note that, for fixed $\Mo$, the divergence $\Div{\Ao,\Bo}{\Mo}$ tells us how badly $\Mo^\rho$ can approximate the probabilities $\Ao^\rho$ and $\Bo^\rho$ when the three observables are measured in one state $\rho$, but the same is not true for $\big(d(\Ao,\Mo_{[1]}),d(\Bo,\Mo_{[2]})\big)$. Indeed, BLW evaluate the worst possible errors separately, so that the two suprema for the Wasserstein distances $d(\Ao^{\rho_1},\Mo^{\rho_1}_{[1]})$ and $d(\Bo^{\rho_2},\Mo^{\rho_2}_{[2]})$ are attained at possibly different states $\rho_1$ and $\rho_2$.

Now, when MURs are derived, the difference of the two approaches is reflected in the distinct aims of the respective statements.

For BLW, proving a MUR means showing that the two deviations $d(\Ao,\Mo_{[1]})$ and $d(\Bo,\Mo_{[2]})$ can not both be too small; that is, all the couples $\big(d(\Ao,\Mo_{[1]}),d(\Bo,\Mo_{[2]})\big)$ must lie above some curve in the real plane, away from the origin. One can even look for the exact characterisation of all the admissible points
\[
\Big\{\big(d(\Ao,\Mo_{[1]}),d(\Bo,\Mo_{[2]})\big): \Mo \in\Mscr(\Xscr\times \Yscr)\Big\};
\]
this is the {\em uncertainty region} (or diagram) of $\Ao$ and $\Bo$. Then, any constraint on the shape of the uncertainty region yields a relation between the worst errors occurring in two separate uses of an approximate joint measurement $\Mo$: namely, for approximating $\Ao\simeq\Mo_{[1]}$ in a first preparation, and $\Bo\simeq\Mo_{[2]}$ in a second one.

On the other hand, in our entropic approach, proving a MUR amounts to giving a strictly positive lower bound for $\Div{\Ao,\Bo}{\Mo}$; the sharpest statements are achieved when $\Icomp(\Ao,\Bo)$ or $\Iad(\Ao,\Bo)$ are explicitly evaluated. This is the state-independent formulation \eqref{MUR1}; it can be further rephrased as the statement \eqref{MUR2} about the inefficiency of an arbitrary approximation $(\Ao,\Bo)\simeq(\Mo_{[1]},\Mo_{[2]})$ that occurs in one preparation of the system, the same for both observables.

\subsection{Noisy observables and uncertainty upper bounds}\label{sec:noise}

Before trying to exactly compute $\Icomp(\Ao,\Bo)$ and $\Iad(\Ao,\Bo)$ in some concrete examples, let us improve their general upper bound given in Theorem \ref{prop:2B}, item
(\ref{cbound}). For this task, we introduce an important class of
bi-observables $\Mo$ that are known to give good approximations of $\Ao$ and
$\Bo$. Even if these $\Mo$ were not optimal, we expect
that they should have a small divergence from $(\Ao,\Bo)$ and thus they should
give a good upper bound for its minimum.

Two incompatible observables $\Ao$ and $\Bo$ can always be turned into a
compatible pair by adding enough classical noise to their measurements. Indeed,
for any choice of trivial observables $\To_\Ao=p_\Ao\id$,
$p_\Ao\in\Pscr(\Xscr)$, and $\To_\Bo=p_\Bo\id$, $p_\Bo\in\Pscr(\Yscr)$, the
observables $\lam\Ao+(1-\lam)\To_\Ao$ and $\gamma\Bo+(1-\gamma)\To_\Bo$, which
are {\em noisy versions} of $\Ao$ and $\Bo$ with {\em noise intensities}
$1-\lam$ and $1-\gamma$, are compatible for all $\lam,\gamma\in [0,1]$ such
that $\lam+\gamma \leq 1$ (sufficient condition) \cite[Prop.\ 1]{BHSS13}. A
bi-observable with the given marginals is
\begin{equation*}
\Mo(x,y)=\lam \Ao(x)p_\Bo(y)+\gamma p_\Ao(x)\Bo(y)+\Big(1-\lam-\gamma\Big)p_\Ao(x)p_\Bo(y)\id.
\end{equation*}
Anyway, depending on $\Ao$, $\Bo$, $p_\Ao$ and $p_\Bo$, it may be possible to go outside the region $\lambda+\gamma \leq 1$, and so reduce the noise intensities. In the following, for every $0\leq\lambda\leq1$, we will consider the couple of equally noisy
observables
\begin{align}\label{eq:defOlam}
\begin{aligned}
\Ao_\lam(x) & = \lam\Ao(x)+(1-\lam)\Ao^{\rho_0}(x)\id, \\
\Bo_\lam(y) & = \lam\Bo(y)+(1-\lam)\Bo^{\rho_0}(y)\id,
\end{aligned}
\end{align}
where $\rho_0=(1/d)\id$ is the maximally chaotic state. Note that, if $\Ao$ is a rank-one sharp observable, then $\Ao^{\rho_0}=u_\Xscr$; a similar consideration holds for $\Bo$. If $\lambda\leq1/2$, the two observables are compatible, but,
depending on the specific $\Ao$ and $\Bo$, they could be compatible also for
larger $\lambda$. In any case, by  \eqref{eq:bAlb} and \eqref{eqdef:D} we get the bound
\begin{multline}\label{eq:bound1}
\Icomp(\Ao,\Bo) \leq \Div{\Ao,\Bo}{\Mo} \\ {}\leq \log\frac{1}{\lam +(1-\lam)
\min_{x\in \Xscr} \Ao^{\rho_0}(x)}
+\log\frac{1}{\lam + (1-\lam)\min_{y\in\Yscr} \Bo^{\rho_0}(y)}
\end{multline}
for all $\lam\in[0,1]$ such that $\Ao_\lam$ and $\Bo_\lam$ are compatible, and
any joint measurement $\Mo$ of $\Ao_\lam$ and $\Bo_\lam$. Since the two terms
in the right hand side of \eqref{eq:bound1} are decreasing functions of $\lam$,
in order to obtain the best bound we are led to find the maximal value
$\lam_{\rm max}$ of $\lam$ for which the noisy observables $\Ao_{\lam}$ and
$\Bo_{\lam}$ are compatible. This problem was addressed in \cite{HSTZ14}, where
a complete solution was given for a couple of Fourier conjugate sharp
observables. Moreover, it was shown that in the general case a nontrivial lower
bound for $\lam_{\rm max}$ can always be achieved by means of {\em optimal
approximate cloning} \cite{KW99}.

Following the same idea, we are going to find a nontrivial upper bound for
$\Icomp(\Ao,\Bo)$ by means of the optimal approximate $2$-cloning channel
\begin{equation*}
\Phi : \sh\to\Scal(\Hscr\otimes\Hscr), \qquad \Phi(\rho) = \frac{2}{d+1}\, S_2 (\rho\otimes\id) S_2 ,
\end{equation*}
where $S_2:\Hscr\otimes\Hscr\to\Hscr\otimes\Hscr$ is the orthogonal projection
of $\Hscr\otimes\Hscr$ onto its symmetric subspace Sym$(\Hscr\otimes\Hscr)$,
defined by $S_2(\phi_1\otimes\phi_2) = (\phi_1\otimes\phi_2 +
\phi_2\otimes\phi_1)/2$. Performing a measurement of the tensor product
observable $\Ao\otimes\Bo$ in the state $\Phi(\rho)$ then amounts to measure
the bi-observable $\Mo_{\rm cl}=\Phi^*(\Ao\otimes\Bo)$ in $\rho$; its marginals
are (see \cite{Wer98})
\[
\Mo_{{\rm cl}\,[1]} = \Ao_{\lam_{\rm cl}} \qquad \text{and} \qquad
\Mo_{{\rm cl}\,[2]} = \Bo_{\lam_{\rm cl}} \qquad \text{where} \qquad \lam_{\rm cl} = \frac{d+2}{2(d+1)}.
\]
Of course $\lam_{\rm cl}\leq\lam_{\rm max}$, but the important point is that
$\lam_{\rm cl}>1/2$. Inserting the above $\lam_{\rm cl}$ in the bound
\eqref{eq:bound1} and using $d\Ao^{\rho_0}(x)=\Tr\left\{\Ao(x)\right\}$, we obtain
\[ 
\Icomp(\Ao,\Bo) \leq \Div{\Ao,\Bo}{\Mo_{\rm cl}} \leq \log\frac{2(d+1)}
{d+2 + \min_x \Tr\left\{\Ao(x)\right\}}+\log\frac{2(d+1)}{d+2 + \min_y \Tr\left\{\Bo(y)\right\}} ,
\] 
holding for all observables $\Ao$ and $\Bo$.

It is worth noticing that the bi-observable $\Mo_{\rm cl}$ describes a
sequential measurement having $\Bo$ as second measured observable. Indeed,
define the instrument $\Jcal\in\Jscr(\Xscr)$, with
\[
\Jcal_x[\rho] = {\rm Tr}_1 \left\{(\Ao(x)\otimes\id)\Phi(\rho)\right\} ,
\]
where ${\rm Tr}_1$ denotes the partial trace with respect to the first factor.
It is easy to check that $\Mo_{\rm cl} = \Jcal^*(\Bo)$, so that $\Mo_{\rm
cl}\in\Mad(\Xscr;\Bo)$. Therefore, the upper bound we have found for
$\Div{\Ao,\Bo}{\Mo_{\rm cl}}$ actually provides a bound also for the entropic
error/disturbance coefficient $\Iad(\Ao,\Bo)$.

Summarizing the above discussion, we thus arrive at the main conclusion of this
section.

\begin{theorem}\label{prop:3}
For any couple of observables $\Ao$ and $\Bo$, we have
\begin{multline}\label{eq:bound3}
\Icomp(\Ao,\Bo) \leq \Iad(\Ao,\Bo) \leq \log\frac{2(d+1)}{d+2 + \min_{x\in\Xscr} \Tr\left\{\Ao(x)\right\}} \\
 + \log\frac{2(d+1)}{d+2 + \min_{y\in\Yscr} \Tr\left\{\Bo(y)\right\}} \leq 2\log\frac{2(d+1)}{d+k} \leq 2 ,
\end{multline}
where in the second to last expression, $k=2$ in general, or even $k=3$ if $|\Xscr|=|\Yscr|=d$ and both $\Ao$ and $\Bo$ are sharp with ${\rm rank}\,\Ao(x) = {\rm rank}\,\Bo(y) = 1$ for all $x,y$.
\end{theorem}

The striking result is that the two uncertainty indexes lie between 0 and 2,
independently of the target
observables $\Ao$ and $\Bo$, the numbers $\abs\Xscr$
and $\abs\Yscr$ of the possible outcomes, and the Hilbert space dimension $d$. Note that the bound $2\log [2(d+1)]/(d+k)$ tends to $2$ from below as $d\to\infty$.

For sharp observables, the bound \eqref {eq:bound3} is much better than the
bound given in Theorem \ref{prop:2B}, item (\ref{cbound}). However, the case
of two trivial uniform observables $\Ao=\Uo_\Xscr$ and $\Bo=\Uo_\Yscr$ is an
example where the bound of Theorem \ref{prop:2B} is better than the bound
\eqref{eq:bound3}.

As a final consideration, we will later show that there are observables $\Ao$
and $\Bo$ such that their compatible noisy versions \eqref{eq:defOlam}  do not
optimally approximate $\Ao$ and $\Bo$. Equivalently, for these observables all
the elements $\Mo\in\Mcomp(\Ao,\Bo)$ (or $\Mo\in\Mad(\Ao,\Bo)$) have marginals
$\Mo_{[1]}\neq\Ao_\lam$ and $\Mo_{[2]}\neq\Bo_\lam$ for all $\lam\in [0,1]$.
Indeed, an example is provided by the two nonorthogonal sharp spin-1/2
observables in Section \ref{ex:D2}. The motivation of this feature comes from
the fact that we are not making any extra assumption about our approximate
joint measurements, as we optimize over the whole sets $\Mscr(\Xscr\times
\Yscr)$ or $\Mscr(\Xscr;\Bo)$, according to the case at hand. This is the main
difference with the approach e.g.\ of \cite{HSTZ14,HeiMZ16}, where a degree of
compatibility is defined by considering the minimal noise which one needs to
add to $\Ao$ and $\Bo$ in order to make them compatible. It should also be
remarked that the non-optimality of the noisy versions is true also in other
contexts \cite{BusH08}.

\subsection{Connections with preparation uncertainty}\label{sec:PUR}
The entropic incompatibility degree and error/disturbance coefficient are the
non trivial and tight lower bounds of the entropic MURs stated in Section
\ref{sec:MUR}. As we recalled in the Introduction, MURs are different from
PURs, which have been formulated in the information-theoretic framework by
using different types of entropies (Shannon, R\'enyi,\ldots)
\cite{Kra87,MaaU88,KriParthas02,MaaWer,WehW10,KTW14,ColesBTW17}. Here we consider only the
Shannon entropy \eqref{def:H}, and, to facilitate the connections with our indexes, we
introduce the \emph{entropic preparation uncertainty coefficient}
\begin{equation}\label{cprep}
\Iprep(\Ao,\Bo)=
\inf_{\rho\in\Sscr(\Hscr)}\left[H(\Ao^\rho)+H(\Bo^\rho)\right] .
\end{equation}
According to the previous sections, the target observables $\Ao$ and $\Bo$
are general POVMs. With this definition, the lower bound proved in
\cite[Cor.~2.6]{KriParthas02} can be written as
\begin{equation}\label{boundprep}
\Iprep(\Ao,\Bo)\geq-\log \max_{x\in \Xscr,\ y\in\Yscr}\norm{\Ao(x)^{1/2}\Bo(y)^{1/2}}^2.
\end{equation}
When the observables are sharp, this lower bound reduces to the one conjectured
in \cite{Kra87} and proved in \cite{MaaU88}.

Note that the infimum in \eqref{cprep} actually is a minimum, because the two
entropies are continuous in $\rho$. Moreover, the equality $\Iprep(\Ao,\Bo)=0$
is attained if and only if there exist two outcomes $x$ and $y$ such that both
positive operators $\Ao(x)$ and $\Bo(y)$ have at least one common
eigenvector with eigenvalue $1$.

For sharp observables, we immediately deduce that the absence of measurement
uncertainty implies the absence of preparation uncertainty. Indeed,
$\Icomp(\Ao,\Bo)=0$ is the same as $\Ao$ and $\Bo$ being compatible, which in
turn is equivalent to the existence of a whole basis of common eigenvectors
$\{\psi_i : i=1,\ldots,d\}$ for which both distributions
$\langle\psi_i|\Ao(x)\psi_i\rangle$ and $\langle\psi_i|\Bo(y)\psi_i\rangle$
reduce to Kronecker deltas \cite[Cor.~5.3]{LA03}. Therefore, we have the
implication $\Icomp(\Ao,\Bo)=0\implies\Iprep(\Ao,\Bo)=0$. However, the same
relation fails for general POVMs: for any couple of trivial observables $\Ao$
and $\Bo$ such that $\Ao\neq\delta_x\id$ or $\Bo\neq\delta_y\id$, we have
$\Icomp(\Ao,\Bo)=0$ and $\Iprep(\Ao,\Bo)>0$. On the converse direction, the
example of two non commuting sharp observables with a common eigenspace shows
that in general $\Iprep(\Ao,\Bo)=0\centernot\implies\Icomp(\Ao,\Bo)=0$. The
failure of this implication exhibits a striking difference between
preparation and measurement uncertainties: actually, the entropic
incompatibility degree vanishes if and only if the two observables are
compatible (Theorem \ref{prop:2B}, item (\ref{jmeas})), while in the
preparation case nothing similar happens.

Nevertheless, there exists a link between the entropic incompatibility degree
$\Icomp$ and the preparation uncertainty coefficient $\Iprep$. Indeed, let us
consider the trivial uniform bi-observable $\Uo\in \Mscr(\Xscr\times \Yscr)$, with
$\Uo = (u_\Xscr\otimes u_\Yscr)\id$ and $ \Uo_{[1]}
= u_\Xscr\id$, \ $\Uo_{[2]} = u_\Yscr\id$.
By Proposition \ref{prop:propHSrel}, item (\ref{HSrel}), we have
\[
\Srel{\Ao^\rho}{ \Uo^\rho_{[1]}}+\Srel{\Bo^\rho}{\Uo^\rho_{[2]}}=\log \abs\Xscr+\log \abs\Yscr -
H(\Ao^\rho)-H(\Bo^\rho).
\]
By taking the supremum over all states, Definitions \ref{def:D} and
\ref{def:Cinc} give
\[
\Icomp(\Ao,\Bo)\leq \Div{\Ao,\Bo}{\Uo}=\log \abs\Xscr+\log \abs\Yscr-\Iprep(\Ao,\Bo).
\]
The final result is the following tradeoff bound:
\begin{equation}\label{jointbound}
\Icomp(\Ao,\Bo)+\Iprep(\Ao,\Bo)\leq \log \abs\Xscr+\log \abs\Yscr.
\end{equation}

Note that this bound is saturated at least in the trivial case $\Ao=u_\Xscr
\id$, $\Bo=u_\Yscr \id$, for which we have $\Iprep(\Ao,\Bo)= \log
\abs\Xscr+\log \abs\Yscr$ and $\Icomp(\Ao,\Bo)=0$. We also remark that
\eqref{jointbound} is not the trivial sum of the two upper bounds
$\Icomp(\Ao,\Bo)\leq 2$ (Theorem \ref{prop:3}) and $\Iprep(\Ao,\Bo)\leq \log
\abs\Xscr+\log \abs\Yscr$ (following from the definition \eqref{cprep} of
$\Iprep$ and the bound for the Shannon entropy of Proposition \ref{prop:propHSrel}, item (\ref{HSbound})).

\section{Symmetries and uncertainty lower bounds}\label{sec:symm}

In quantum mechanics, many fundamental observables are directly related to symmetry properties of the quantum system at hand. That is, in many concrete
situations there is some symmetry group $G$ acting on both the measurement
outcome space and the set of system states, in such a way that the
two group actions naturally intertwine. The observables that preserve the
symmetry structure are usually called {\em $G$-covariant}.

In the present setting, covariance will help us to find the
incompatibility degree $\Icomp(\Ao,\Bo)$ and characterize the optimal set
$\Mcomp(\Ao,\Bo)$ for a couple of sharp observables $\Ao$ and $\Bo$ sharing
suitable symmetry properties. In Section \ref{sec:cov.gen} below we provide a
general result in this sense, which we then apply to the cases of two spin-1/2
components (Section \ref{ex:D2}) and two observables that are conjugated by the
Fourier transform of a finite field (Section \ref{sec:cov.MUBs}).

\subsection{Symmetries and optimal approximate joint measurments}\label{sec:cov.gen}

We now suppose that the joint outcome space $\Xscr\times\Yscr$ carries the
action of a finite group $G$, acting on the left, so that each $g\in G$ is
associated with a bijective map on the finite set $\Xscr\times\Yscr$. Moreover,
we also assume that there is a projective unitary representation $U$ of $G$ on
$\Hscr$. The following natural left actions are then defined for all $g\in G$:
\begin{enumerate}[-]
\item on $\sh$: $g\rho = U(g)\rho U(g)^*$;
\item on $\Pscr(\Xscr\times\Yscr)$: $gp(x,y) = p(g^{-1}(x,y))$ for all
    $(x,y)\in\Xscr\times\Yscr$;
\item on $\Mscr(\Xscr\times\Yscr)$: $g\Mo(x,y) =
    U(g)\Mo(g^{-1}(x,y))U(g)^*$ for all $(x,y)\in\Xscr\times\Yscr$.
\end{enumerate}
While the two actions on $\sh$ and $\Pscr(\Xscr\times\Yscr)$ have a clear
physical interpretation, the action on $\Mscr(\Xscr\times\Yscr)$ is understood
by means of the fundamental relation
\begin{equation}\label{eq:gMrho}
g(\Mo^\rho) = (g\Mo)^{g\rho} ,
\end{equation}
which asserts that $g\Mo$ is defined in such a way that measuring it on the
transformed state $g\rho$ just gives the translated probability $g(\Mo^\rho)$.
Note that the parenthesis order actually matters in \eqref{eq:gMrho}.

A fixed point $\Mo$ for the action of $G$ on $\Mscr(\Xscr\times\Yscr)$ is a
{\em $G$-covariant observable}, i.e.\ $U(g)\Mo(x,y)U(g)^*=\Mo(g(x,y))$ for all
$(x,y)\in\Xscr\times\Yscr$ and $g\in G$. On the other hand, if
$\Mo\in\Mscr(\Xscr\times\Yscr)$ is any observable, then
\begin{equation}\label{eq:covariantization}
\Mo_G = \frac{1}{|G|} \sum_{g\in G} g\Mo
\end{equation}
is a $G$-covariant element in $\Mscr(\Xscr\times\Yscr)$, which we call the {\em
covariant version} of $\Mo$.

Now we state some sufficient conditions on the observables $\Ao,\Bo$ and the
action of the group $G$ ensuring that the entropic divergence
$\Div{\Ao,\Bo}{\cdot}$ is $G$-invariant, and then we derive their consequences
on the optimal approximate joint measurements of $\Ao$ and $\Bo$.

Note that the relative entropy is always invariant for a group action, that is,
\begin{equation}\label{eq:invS}
\Srel{gp}{gq} = \Srel pq, \qquad \forall p,q\in\Pscr(\Xscr\times\Yscr),\quad g\in G,
\end{equation}
by Proposition \ref{prop:propHSrel}, (\ref{HSinv}). Note also that, for
$p\in\Pscr(\Xscr\times\Yscr)$, the expression $gp_{[i]} = (gp)_{[i]}$ is
unambiguous, as the action of $g$ is defined on $\Pscr(\Xscr\times\Yscr)$ and
not on $\Pscr(\Xscr)$ or $\Pscr(\Yscr)$.

\begin{theorem}\label{prop:invarD^}
Let $\Ao\in\Mscr(\Xscr)$, $\Bo\in\Mscr(\Yscr)$ be the target observables.
Let $G$ be a finite group, acting on $\Xscr\times\Yscr$ and with a projective
unitary representation $U$ on $\Hscr$. Suppose the group $G$ is generated by a
subset $S_G\subseteq G$, such that each $g\in S_G$ satisfies either one
condition between:
\begin{enumerate}[(i)]
\item \label{8i} there exist maps $f_{g,\Xscr}:\Xscr\to\Xscr$ and
    $f_{g,\Yscr}:\Yscr\to\Yscr$ such that, for all $x\in\Xscr$ and
    $y\in\Yscr$,
\begin{enumerate}
\item \label{8ia} $g(x,y) = (f_{g,\Xscr}(x),f_{g,\Yscr}(y))$
\item \label{8ib} $U(g)\Ao(x)U(g)^* = \Ao(f_{g,\Xscr}(x))$ and
    $U(g)\Bo(y)U(g)^* = \Bo(f_{g,\Yscr}(y))$;
\end{enumerate}
\item \label{8ii} there exist maps $f_{g,\Xscr}:\Xscr\to\Yscr$ and
    $f_{g,\Yscr}:\Yscr\to\Xscr$ such that, for all $x\in\Xscr$ and
    $y\in\Yscr$,
\begin{enumerate}
\item \label{8iia} $g(x,y) = (f_{g,\Yscr}(y),f_{g,\Xscr}(x))$
\item \label{8iib} $U(g)\Ao(x)U(g)^* = \Bo(f_{g,\Xscr}(x))$ and
    $U(g)\Bo(y)U(g)^* = \Ao(f_{g,\Yscr}(y))$ .
\end{enumerate}
\end{enumerate}
Then, $\Div{\Ao,\Bo}{g\Mo} = \Div{\Ao,\Bo}{\Mo}$ for all
$\Mo\in\Mscr(\Xscr\times\Yscr)$ and $g\in G$.
\end{theorem}
\begin{proof}
If two elements $g_1,g_2\in G$ satisfy the above hypotheses, so does their
product $g_1 g_2$. Since $S_G$ generates $G$, we can then assume that $S_G=G$.
In this case, condition (\ref{8i}.a) or (\ref{8ii}.a) easily implies the
relation
\begin{equation}\label{eq:a}
gp_{[1]}\otimes gp_{[2]} = g(p_{[1]}\otimes p_{[2]}), \qquad \forall p\in\Pscr(\Xscr\times\Yscr), \, g\in G .
\end{equation}
On the other hand, by condition (\ref{8i}.b) or (\ref{8ii}.b), we get
\begin{equation}\label{eq:b}
\Ao^{g\rho}\otimes\Bo^{g\rho} = g(\Ao^\rho\otimes\Bo^\rho), \qquad \forall \rho\in\sh, \, g\in G .
\end{equation}
For any $\Mo\in\Mscr(\Xscr\times\Yscr)$, we then have
\begin{align*}
\sddiv{\Ao,\Bo}{g^{-1}\Mo}(\rho) & = \Srel{\Ao^\rho\otimes\Bo^\rho}
{(g^{-1}\Mo)_{[1]}^\rho\otimes(g^{-1}\Mo)_{[2]}^\rho} && \qquad \text{by \eqref{eq:Soprod}} \\
& = \Srel{\Ao^\rho\otimes\Bo^\rho}{g^{-1}(\Mo^{g\rho})_{[1]}\otimes g^{-1}(\Mo^{g\rho})_{[2]}}
&& \qquad \text{by \eqref{eq:gMrho}} \\
& = \Srel{\Ao^\rho\otimes\Bo^\rho}{g^{-1}(\Mo^{g\rho}_{[1]}\otimes \Mo^{g\rho}_{[2]})}
&& \qquad \text{by \eqref{eq:a}} \\
& = \Srel{g(\Ao^\rho\otimes\Bo^\rho)}{\Mo^{g\rho}_{[1]}\otimes \Mo^{g\rho}_{[2]}}
&& \qquad \text{by \eqref{eq:invS}} \\
& = \Srel{\Ao^{g\rho}\otimes\Bo^{g\rho}}{\Mo^{g\rho}_{[1]}\otimes \Mo^{g\rho}_{[2]}}
&& \qquad \text{by \eqref{eq:b}} \\
& = \sddiv{\Ao,\Bo}{\Mo}(g\rho) . &&
\end{align*}
Taking the supremum over $\rho$ and observing that $\sh=g\sh$, it follows that
\\ $\Div{\Ao,\Bo}{g^{-1}\Mo} = \Div{\Ao,\Bo}{\Mo}$. \end{proof}
\begin{remark}
\begin{enumerate}\label{rem:invarD^}
\item The occurrence of either hypothesis (\ref{8i}) or (\ref{8ii}) may depend on the generator $g\in S_G$; however, in both cases $g$ does not mix the $\Xscr$ and $\Yscr$ outcomes together.

\item Conditions (\ref{8i}.a), (\ref{8ii}.a) are hypotheses about the
    action of $G$ on the outcome space $\Xscr\times\Yscr$. Note that each
    one implies that the maps $f_{g,\Xscr}$ and $f_{g,\Yscr}$ are
    bijective. In particular, one can have some generator $g$ satisfying
    (\ref{8ii}.a) only if $|\Xscr|=|\Yscr|$.

\item Conditions (\ref{8i}.b), (\ref{8ii}.b) involve also the observables
    $\Ao$ and $\Bo$. Even if $\Ao$ and $\Bo$ are not compatible, they are
    required to behave as if they were the marginals of a covariant
    bi-observable.
\item \label{it:perm}The symmetries allowed in hypothesis (\ref{8ii}) of
    Theorem \ref{prop:invarD^} essentially are of permutational nature.
    They directly follow from the exchange symmetry of the error function
    \eqref{statedependentdivergence}, in which the approximation errors
    $\Srel{\Ao^\rho}{ \Mo^\rho_{[1]}}$ and $\Srel{\Bo^\rho}{
    \Mo^\rho_{[2]}}$ are equally weighted.
\end{enumerate}
\end{remark}

\begin{corollary}\label{prop:invarI^}
Under the hypotheses of Theorem \ref{prop:invarD^},
\begin{enumerate}[-]
\item the set $\Mcomp(\Ao,\Bo)$ is $G$-invariant;
\item for any $\Mo\in\Mcomp(\Ao,\Bo)$, we have $\Mo_G\in\Mcomp(\Ao,\Bo)$;
\item there exists a $G$-covariant observable in $\Mcomp(\Ao,\Bo)$.
\end{enumerate}
\end{corollary}
\begin{proof}
Since $\Div{\Ao,\Bo}{\cdot}$ is $G$-invariant by Theorem \ref{prop:invarD^},
then the set $\Mcomp(\Ao,\Bo)$ is $G$-invariant. This fact and the convexity of
$\Mcomp(\Ao,\Bo)$ implies that $\Mo_G\in\Mcomp(\Ao,\Bo)$ for all
$\Mo\in\Mcomp(\Ao,\Bo)$. Since the latter set is nonempty by Theorem
\ref{prop:2B}, item (\ref{optM}), it then always contains a $G$-covariant
observable.
\end{proof}

\begin{remark}\label{rem:covariant}
Since the covariance requirement reduces the many degrees of freedom in the
choice of a bi-observable $\Mo\in\Mscr(\Xscr\times\Yscr)$, we expect that the
larger is the symmetry group $G$, the fewer amount of free parameters will be
needed to describe a $G$-covariant element $\Mo$.  This
will be a big help in the computation of $\Icomp(\Ao,\Bo)$, as Corollary
\ref{prop:invarI^} allows to minimize $\Div{\Ao,\Bo}{\cdot}$ just on the set of
$G$-covariant bi-observables. More precisely, under the hypotheses of Theorem
\ref{prop:invarD^},
\begin{equation*}
\Icomp(\Ao,\Bo)=\min_{\substack{\Mo\in\Mscr(\Xscr\times\Yscr)\\
\Mo\text{ $G$-covariant}}}\max_{\substack{\rho\in\Sscr(\Hscr)\\ \rho\text{ pure}}}
\left\{\Srel{\Ao^\rho}{\Mo_{[1]}^\rho} + \Srel{\Bo^\rho}{\Mo_{[2]}^\rho}\right\},
\end{equation*}
where the minimum has to be computed only with respect to the parameters
describing a $G$-covariant bi-observable $\Mo$. In particular, it is only the
dependence of the marginals $\Mo_{[1]}$ and $\Mo_{[2]}$ on such parameters that
comes into play. Of course, solving this double optimization problem yields the
value of $\Icomp(\Ao,\Bo)$ and all the covariant optimal joint measurement of
$\Ao$ and $\Bo$, but not the whole optimal set $\Mcomp(\Ao,\Bo)$.
\end{remark}

In the cases of two othogonal spin-1/2 components (Section \ref{sec:2spins}) and two Fourier conjugate observables (Section \ref{sec:cov.MUBs}), covariance will reduce the number of
parameters to just a single one.

If $\Bo$ is not sharp, the two sets $\Mcomp(\Ao,\Bo)$ and $\Mad(\Ao,\Bo)$ may
be different, and we need a specific corollary for $\Mad(\Ao,\Bo)$. Indeed,
stronger hypotheses are required to ensure that the sequential measurement set $\Mscr(\Xscr;\Bo)$  is $G$-invariant.

\begin{corollary}\label{prop:invarIad^}
Under the hypotheses of Theorem \ref{prop:invarD^}, and supposing in addition
that all the generators $g\in S_G$ enjoy only condition (\ref{8i}) of that theorem,
\begin{enumerate}[-]
\item the set $\Mscr(\Xscr;\Bo)$ is $G$-invariant;
\item the set $\Mad(\Ao,\Bo)$ is $G$-invariant;
\item for any $\Mo\in\Mad(\Ao,\Bo)$, we have $\Mo_G\in\Mad(\Ao,\Bo)$;
\item there exists a $G$-covariant observable in $\Mad(\Ao,\Bo)$.
\end{enumerate}
\end{corollary}
\begin{proof}
We know that $\Div{\Ao,\Bo}{\cdot}$ is $G$-invariant by Theorem
\ref{prop:invarD^}, and so we only have to prove that $\Mscr(\Xscr;\Bo)$ is
$G$-invariant; then the subsequent claims follow as in Corollary
\ref{prop:invarI^}. Since we can assume $S_G=G$, any element $g\in G$ maps a
sequential measurement $\Mo=\Jcal^*(\Bo)$ to another sequential measurement
$\Jcal'^*(\Bo)$, due to condition (\ref{8i}) of Theorem \ref{prop:invarD^}:
\begin{multline*}
g\Mo(x,y) =  U(g)\Mo\Big(g^{-1}(x,y)\Big)U(g)^* =
U(g)\Mo\Big(f_{g,\Xscr}^{-1}(x),
f_{g,\Yscr}^{-1}(y)\Big)U(g)^* \\
{}=  U(g)\Jcal^*_{f_{g,\Xscr}^{-1}(x)}\Big[\Bo(f_{g,\Yscr}^{-1}(y))\Big]U(g)^* =
U(g)\Jcal^*_{f_{g,\Xscr}^{-1}(x)}\Big[U(g)^*\Bo(y)U(g)\Big]U(g)^*=:\Jcal'^*(\Bo)(x,y).
\end{multline*}
\end{proof}

\begin{remark}
Corollary \ref{prop:invarIad^} does not admit elements $g$ satisfying condition
(\ref{8ii}) of Theorem \ref{prop:invarD^} because this hypothesis alone can not
guarantee the $G$-invariance of the set $\Mscr(\Xscr;\Bo)$. Of course, it works
for a sharp $\Bo$, but it could fail, for example, for a trivial $\Bo$. Indeed,
take $\Xscr=\Yscr$ and $\Ao=\Bo=\Uo_\Xscr$; then
$\Mscr(\Xscr;\Bo)=\{\Mo\in\Mscr(\Xscr\times\Yscr) : \Mo(x,y)=\Mo_1(x)
u_\Xscr(y),\ \forall x,y, \text{ for some } \Mo_1 \in\Mscr(\Xscr)\}$, and
$\Mo_{[2]}(y)=\Bo(y)$ has rank $d$ for every $\Mo\in\Mscr(\Xscr;\Bo)$ and
$y\in\Yscr$. Nevertheless, if $g$ satisfies (\ref{8ii}.a), then (\ref{8ii}.b)
is obvious, but $g$ could send a sequential measurement $\Mo$ outside
$\Mscr(\Xscr;\Bo)$. Indeed,
$(g\Mo)_{[2]}(y)=U(g)\Mo_1\Big(f_{g,\Xscr}^{-1}(y)\Big)U(g)^*$ has rank equal
to the rank of $\Mo_1\Big(f_{g,\Xscr}^{-1}(y)\Big)$, which can be chosen
smaller than $d$.
\end{remark}

\subsection{Two spin-1/2 components} \label{ex:D2}

As a first application of Theorem \ref{prop:invarD^} and its corollaries, we take as target observables two spin-1/2 components along the directions
defined by two unit vectors $\vec  a$ and $\vec b$ in $\Rbb^3$. They are
represented by the sharp observables
\begin{equation}\label{ABspin}
\Ao(x) = \frac{1}{2}
\left(\id+x\,\vec {a} \cdot \vec {\sigma}\right),  \quad \Bo(y)
= \frac{1}{2} \left(\id+ y\,\vec {b} \cdot \vec {\sigma}\right),
\end{equation}
where $\vec {\sigma} = (\sigma_1\,,\sigma_2\,,\sigma_3)$ is the vector of the
three Pauli matrices on $\Hscr=\Cbb^2$, and $\Xscr=\Yscr=\{-1,+1\}$.
Let $\alpha\in
[0,\pi]$ be the angle formed by $\vec {a}$ and $\vec {b}$; by item (\ref{cinv}) of Theorem \ref{prop:2B}, the coefficient $\Icomp(\Ao,\Bo)$ does not depend on
the choice of the values of the outcomes, and this allows us to take $\alpha\in
[0,\pi/2]$. Indeed, when $\alpha>\pi/2$, it is enough to change $y\to -y$ and
$\vec {b} \to -\vec {b}$ to recover the previous case.
Without loss of generality, we take the two spin directions in the $\vec i\vec j$-plane
and choose the $\vec  i$- and $\vec  j$-axes in such a way that the bisector of
the angle formed by $\vec {a}$ and $\vec {b}$ coincides with the bisector $\vec
n$ of the first quadrant; $\vec m$ is the bisector of the second quadrant.
This choice is illustrated in Figure \ref{abfig},
where $\alpha\in \left[0,\pi/2\right]$,  \ $ a_1^{\,2}+a_2^{\,2}=1$, and
\begin{equation}\label{a+b}
\begin{split}
&\vec a = a_1\vec i + a_2 \vec j,  \qquad\vec b =
a_2\vec i + a_1\vec j,  \qquad\vec
n=\frac{\vec i+\vec j}{\sqrt 2} ,\qquad \vec  m=\frac{\vec  j-\vec i}{\sqrt 2},
\\
&a_1= \sqrt{\frac{1+\sin \alpha}2}\in \left[\frac
1{\sqrt 2}\,,\, 1\right],\qquad a_2=\frac{\cos \alpha}{\sqrt{2(1+\sin
\alpha)}}\in \left[0,\, \frac 1{\sqrt 2} \right] .
\end{split}
\end{equation}
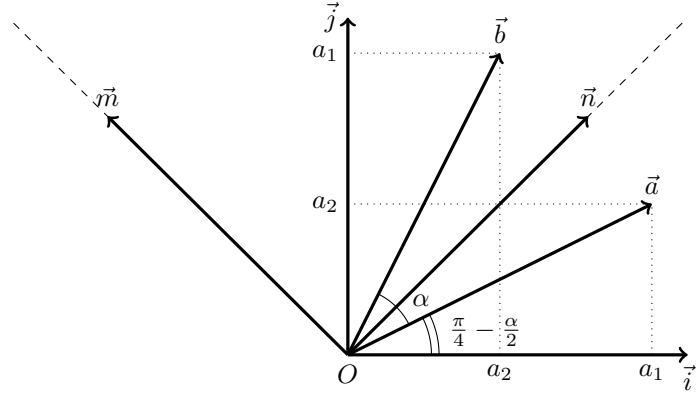
\begin{figure}[h]
\centering
\begin{tikzpicture}[domain=-4:4]
\draw[->, very thick](0,0)node[anchor=north]{$O$}--(4.472,0)node[anchor=north]{$\vec {i}$};
\draw[-, dashed](0,0)--(4.472,4.472);
\draw[->, very thick](0,0)--(3.162,3.162)node[anchor=south]{$\vec {n}$};
\draw[-, dashed](0,0)--(-4.472,4.472);
\draw[->, very thick](0,0)--(-3.162,3.162)node[anchor=south]{$\vec {m}$};
\draw[->, very thick](0,0)--(0,4.472)node[anchor=east]{$\vec {j}$};
\draw[->, very thick](0,0)--(4,2)node[anchor=south]{$\vec {a}$};
\draw[->, very thick](0,0)--(2,4)node[anchor=south]{$\vec {b}$};
\centerarc[](0,0)(26.565:63.435:0.9);
\draw[thick](0.72,0.72)node[anchor=west]{$\alpha$};
\centerarc[](0,0)(0:26.565:1.1);
\centerarc[](0,0)(0:26.565:1.2);
\draw[thick](1.2,0.30)node[anchor=west]{$\frac{\pi}{4}-\frac{\alpha}{2}$};
\draw[-,dotted](4,0)node[anchor=north]{$a_1$} --(4,2);
\draw[-,dotted](0,2)node[anchor=east]{$a_2$}--(4,2);
\draw[-,dotted](2,0)node[anchor=north]{$a_2$} --(2,4);
\draw[-,dotted](0,4)node[anchor=east]{$a_1$}--(2,4);
\end{tikzpicture}
\caption{The unit vectors $\vec  a$ and $\vec  b$ characterizing the target spin-1/2 observables \eqref{ABspin}.}\label{abfig}
\end{figure}

In the next part, we will see that the compatible observables optimally approximating the two target spins \eqref{ABspin} are noisy versions of another two spin-1/2 components; however, in general their directions may be different from the original $\vec{a}$ and $\vec{b}$. For this reason, we need to introduce the family of observables $\Ao_{\vec {c}},\Bo_{\vec {c}}\in\Mscr(\{+1,-1\})$, with
\begin{equation}\label{eq:marginstorte}
\Ao_{\vec {c}}(x) = \frac{1}{2}\left[\id + x\left(c_1 \sigma_1
+ c_2 \sigma_2\right)\right],
\quad \Bo_{\vec {c}}(y)= \frac{1}{2} \left[\id + y\left(c_2 \sigma_1
+ c_1 \sigma_2\right)\right],
\end{equation}
where $\vec  c = c_1\vec {i} + c_2\vec {j}$, $c_i\in \Rbb$. Note that the components of $\vec  c$ appear in $\Ao_{\vec  c}$ and $\Bo_{\vec c}$ in the reverse order; moreover, $\Ao=\Ao_{\vec {a}}$ and $\Bo=\Bo_{\vec {a}}$. Formula \eqref{eq:marginstorte} defines two observables if and only if $\abs{\vec {c}}\leq 1$, that is, $\vec{c}$ belongs to the disk
\begin{equation}\label{diskC}
C=\big\{c_1\vec {i}+c_2\vec {j} : c_1^2+c_2^2\leq 1\big\}.
\end{equation}
Note that, for $\abs{\vec{c}} = 1$, the observable $\Ao_{\vec{c}}$ is sharp, and it is the spin-1/2 component along the direction $\vec{c}$; on the other hand, for $\abs{\vec{c}} \in (0,1)$, $\Ao_{\vec {c}}$ is a noisy version of $\Ao_{\vec{c}/\abs{\vec{c}}}$ with noise intensity $1-\lam = 1-\abs{\vec{c}}$ (cf.~\eqref{eq:defOlam}). Analogue considerations hold for $\Bo_{\vec{c}}$.

\subsubsection{Entropic incompatibility degree and optimal measurements}\label{sec:2spins}
\ When the angle between the spin directions $\vec {a}$ and $\vec {b}$ is $\alpha=\pi/2$, the target observables become the two orthogonal spin-1/2 components along the $\vec i$- and $\vec j$-axes:
\begin{equation}\label{def:XY}
\Ao(x)=\Xo(x) = \frac{1}{2}\,(\id + x\sigma_1), \qquad \Bo(y)=\Yo(y) = \frac{1}{2}\,(\id + y\sigma_2).
\end{equation}
In Appendix \ref{app:orth}, we use Theorem \ref{prop:invarD^} and the many rotational symmetries of these observables to drastically simplify the problem of finding both the value of $\Icomp(\Xo,\Yo)$ and the explicit expression of a bi-observable in $\Mcomp(\Xo,\Yo)$. Remarkably, it also turns out that  $\Mcomp(\Xo,\Yo)$ is a singleton set. Indeed, the following theorem is proved.

\begin{theorem}\label{teo:incoORTHSPIN}
Let $\Xo$ and $\Yo$ be the two orthogonal spin-1/2 components \eqref{def:XY}. Then, there is a unique optimal approximate joint measurement of $\Xo$ and $\Yo$, that is the bi-observable
\begin{equation}\label{orthcov}
\Mo_0(x,y) = \frac{1}{4}\left(\id + \frac x {\sqrt2}\,\sigma_1
+ \frac y {\sqrt2}\,\sigma_2\right) ,
\end{equation}
i.e.\  $\Mcomp(\Xo,\Yo) = \{\Mo_0\}$. If $\rho_e$ is the projection on any eigenvector of $\sigma_1$ or $\sigma_2$, then
\begin{equation}\label{corth}
\Icomp(\Xo,\Yo)=\sddiv{\Xo,\Yo}{\Mo_0}(\rho_{e}) =\log \left[2\left(2-\sqrt 2\right)\right] \simeq 0.228447\,.
\end{equation}
\end{theorem}

Note that $\Mo_0(x,y)$ is a rank-one operator for all $(x,y)\in\Xscr\times\Yscr$, and its marginals
\[
\Mo_{0\,[1]}(x) = \frac{1}{2}\left( \id + \frac x {\sqrt2}\,\sigma_1\right) , \qquad \Mo_{0\,[2]}(y) = \frac{1}{2}\left( \id + \frac y {\sqrt2}\,\sigma_2\right)
\]
turn out to be the noisy versions $\Xo_{1/\sqrt{2}}$, $\Yo_{1/\sqrt{2}}$ of the target observables $\Xo$, $\Yo$ (cf.~\eqref{eq:defOlam}).

When the two spin directions $\vec {a}$ and $\vec {b}$ are not orthogonal, the system loses the $180^\circ$ rotational symmetries around the $\vec i$- and $\vec j$-axes. According to Remark \ref{rem:covariant}, this makes the evaluation of $\Icomp(\Ao,\Bo)$ a more difficult task. The best we can do is to express $\Icomp(\Ao,\Bo)$ as the solution of a maximization/minimization (minimax) problem for an explicit function of two parameters. The analysis of the symmetries of two nonorthogonal spin-1/2 components, and the consequent proof of the next theorem are given in Appendix \ref{ex:D4}.

\begin{theorem}\label{teo:incoNONORTHSPIN}
Let $\Ao$ and $\Bo$ be the spin-1/2 components \eqref{ABspin} with the angle $\alpha\in[0,\pi/2]$. For all $\phi\in [0,2\pi)$, $\gamma\in [-1,1]$ and $x,y\in\{-1,+1\}$, define
\begin{gather}
\rho(\phi) = \frac{1}{2}\left(\id + \cos\phi\, \sigma_1 + \sin\phi\, \sigma_2\right) , \qquad \vec{c}(\gamma) = \frac{\vec{i} + \gamma \vec{j}}{\sqrt{2}} , \label{eq:rho(phi)c(gamma)} \\
\Mo_\gamma(x,y) = \frac{1}{4}\left[ \left(1 + \gamma xy\right)\id + \frac{1}{\sqrt{2}} \left(x\sigma_1 + y\sigma_2\right) + \frac{\gamma}{\sqrt{2}} \left(y\sigma_1 + x\sigma_2\right)\right] . \label{eq:optispinobs}
\end{gather}
Then, $\Mo_\gamma\in\Mscr(\Xscr\times\Yscr)$, and we have
\begin{gather}\label{eq:optispin2}
\Icomp(\Ao,\Bo)  = \min_{\gamma\in [-1,1]} \max_{\phi\in [0,2\pi)} \sddiv{\Ao,\Bo}{\Mo_\gamma}(\rho(\phi)),
\\ \label{eq:dopo_optispin2}
\sddiv{\Ao,\Bo}{\Mo_\gamma}(\rho) = \Srel{\Ao^{\rho}}{\Ao_{\vec{c}(\gamma)}^{\rho}} + \Srel{\Bo^{\rho}}{\Bo_{\vec{c}(\gamma)}^{\rho}}.
\end{gather}
Moreover, $\gamma$ solves the minimization problem \eqref{eq:optispin2} if and only if $\Mo_\gamma\in\Mcomp(\Ao,\Bo)$.
\end{theorem}

In Section \ref{sez:nuova}, we provide a numerical evaluation of the entropic incompatibility degree \eqref{eq:optispin2} for some angles $\alpha\in [0,\pi/2]$. Moreover, using the family of approximate joint measurements in \eqref{eq:optispinobs}, we analytically find a lower bound for $\Icomp(\Ao,\Bo)$.
Note that, for $\alpha\in (0,\pi/2)$, it is not clear if there is a unique $\gamma$ solving \eqref{eq:optispin2}, and if the set $\Mcomp(\Ao,\Bo)$ is only made up of the corresponding bi-observables $\Mo_\gamma$ (see also Remark \ref{rem:covgamma} in Appendix \ref{ex:D4}).

The noisy spin-1/2 components $\Ao_{\vec{c}(\gamma)}$ and $\Bo_{\vec{c}(\gamma)}$ appearing in \eqref{eq:dopo_optispin2} are the two marginals of the bi-observable $\Mo_\gamma$ in \eqref{eq:optispinobs}. When $\Mo_\gamma$ is optimal, we stress that for $\alpha \neq \pi/2$ they may not be noisy versions of the target observables $\Ao$ and $\Bo$. Indeed, in Table \ref{tab} below and the subsequent discussion, we numerically show this for the case $\alpha = \pi/4$.

It is worth noticing that every bi-observable \eqref{orthcov} or \eqref{eq:optispinobs} can be rewritten as a \emph{mixture} (convex combination) of two  sharp joint measurements of compatible spin components,
along the bisector $\vec n$ in the case of the first bi-observable, and along
the bisector $\vec m$ for the other one. More precisely, we introduce the sharp
bi-observables
\begin{equation}\label{eq:M+-}
\begin{aligned}
\Mo_+ (x,y) & = \left[\frac{1}{2}(\id+x\vec {n}\cdot\vec {\sigma})\right]
\left[\frac{1}{2}(\id+y\vec {n}\cdot\vec {\sigma})\right] \equiv \Ao_{\vec n}(x) \Bo_{\vec n}(y), \\
\Mo_- (x,y) & = \left[\frac{1}{2}(\id-x\vec {m}\cdot\vec {\sigma})\right]
\left[\frac{1}{2}(\id+y\vec {m}\cdot\vec {\sigma})\right] \equiv \Ao_{-\vec m}(x) \Bo_{-\vec m}(y) .
\end{aligned}
\end{equation}
Then, we have
\begin{equation}\label{gamma,pm}
\Mo_0 = \frac{1}{2}(\Mo_+ + \Mo_-) , \qquad \Mo_\gamma = \frac{1+\gamma}{2}\, \Mo_+ + \frac{1-\gamma}{2}\, \Mo_-.
\end{equation}
In terms of $\Mo_0$, the bi-observable $\Mo_\gamma$ can be expressed also as the mixture
\[
\Mo_\gamma = \begin{cases}
\gamma\Mo_++ (1-\gamma)\Mo_0 & \quad \text{if $\gamma \geq 0$}, \\
\abs{\gamma}\Mo_- +(1-\abs{\gamma})\Mo_0 & \quad \text{if $\gamma \leq0$} .
\end{cases}
\]

\subsubsection{Numerical and analytical results for nonorthogonal components}\label{sez:nuova}
In the case of two arbitrarily oriented spin components, the
minimax problem \eqref{eq:optispin2},
giving $\Icomp$ and $\gamma$ for the optimal bi-observable $\Mo_\gamma$, is hard to be solved analytically. Nevertheless, the double optimization over the angle $\phi$ and the parameter $\gamma$ can be tackled numerically, and the resulting $\Icomp(\Ao,\Bo)$
for $100$ equally distant values $\alpha$ in the interval $[0,\pi/2]$ are
plotted in Figure \ref{mod}.
\begin{figure}[h]
\centering
\includegraphics*[scale=0.85]{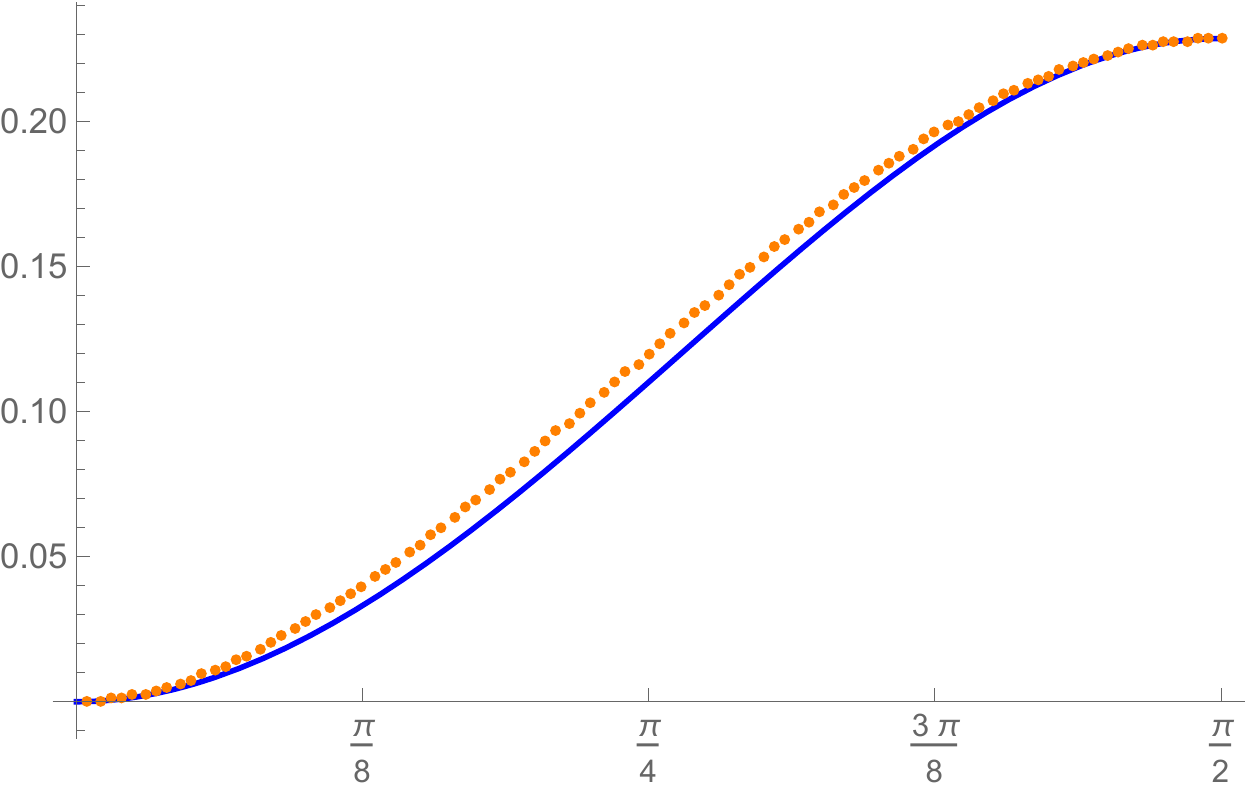}
\caption{Dots: numerical evaluations of $\Icomp(\Ao,\Bo)$ as a function of $\alpha$.
Continuous line:
the analytical lower bound $LB(\alpha)$ in \eqref{s:c2min}.}\label{mod}
\end{figure}

A good analytical lower bound for $\Icomp(\Ao,\Bo)$ can be found by fixing a trial state $\rho(\phi)$, considering the bi-observables $\Mo_\gamma$ of
\eqref{eq:optispinobs}, and then minimizing the error
function $\sddiv{\Ao,\Bo}{\Mo_{\gamma}}(\rho(\phi))$ with respect to $\gamma\in [-1,1]$. Indeed, equation \eqref{eq:optispin2} yields the inequality $\Icomp(\Ao,\Bo)\geq\min_{\gamma\in [-1,1]}\sddiv{\Ao,\Bo}{\Mo_{\gamma}}(\rho(\phi))$ for all $\phi\in [0,2\pi)$. A convenient choice for $\phi$, suggested by the results in the case of two orthogonal components, is to take $\phi\in\{\pi/4\pm\alpha/2,\,5\pi/4\pm\alpha/2\}$, so that the corresponding state $\rho(\phi)$ is any eigenprojection of $\vec a\cdot\vec\sigma$ or $\vec b\cdot \vec\sigma$; say we take the eigenprojection $\rho_e = \rho(\pi/4-\alpha/2)$ of $\vec a\cdot \vec\sigma$  with positive eigenvalue. Then, we get
\begin{equation}\label{lbspin}
\Icomp(\Ao,\Bo)\geq \min_{\gamma \in [-1,1]} \sddiv{\Ao,\Bo}{\Mo_{\gamma}}(\rho_e)=: LB(\alpha) .
\end{equation}
In Appendix \ref{sec:lb}, the explicit expression of $\sddiv{\Ao,\Bo}{\Mo_\gamma}(\rho_e)$ is given in
\eqref{s:c2}, its minimum over $\gamma$ is computed and, for $\alpha\neq \pi/2$, it is found at the
point
\begin{equation}\label{c2min}
\gamma =\frac{\sqrt{2}\ell - a_2}{a_1},
\end{equation}
where
\begin{equation}\label{ellmin}
\ell=\frac{1}{2\sqrt 2\,
a_2}\left(\sqrt{u^2+8(1+u)a_2^{\,2}}-u\right),
\end{equation}
\begin{equation}\label{def:u}
u=\left(a_1+\frac 1 {\sqrt 2}\right)\frac{a_1^{\,2} -a_2^{\,2}}{\sqrt 2}
=\left(1+\sqrt{1+\sin\alpha}\right)\frac{\sin\alpha}2\,.
\end{equation}
In particular, the value \eqref{c2min} for $\gamma$, together with the fact that the bi-observable $\Mo_\gamma$ has marginals $\Mo_{\gamma\,[1]} = \Ao_{\vec{c}(\gamma)}$ and $\Mo_{\gamma\,[2]} = \Bo_{\vec{c}(\gamma)}$, show that the marginals of the bi-observable giving the lower
bound \eqref{lbspin} are not  noisy versions of the target observables $\Ao$ and $\Bo$; indeed, $\vec{c}(\gamma)\not\propto\vec{a}$ in this case. Finally, the lower bound turns out to be
\begin{equation}\label{s:c2min}
LB(\alpha)=- \log {w} + \frac 12
\left(1+ \cos\alpha\right)\log \frac {1+ \cos\alpha} {1+ \ell}
+ \frac 12 \left(1-\cos\alpha\right)\log \frac {1- \cos\alpha} {1-\ell}\,,
\end{equation}
with
\begin{equation}
w=\frac 12+\frac{\sqrt{u^2+8(1+u)a_2^{\,2}}}{4\sqrt 2 \, a_1}+
\frac{\sin\alpha}{8}\left(\frac 3 {\sqrt 2 \,a_1}-1\right).
\end{equation}
The plot of $LB(\alpha)$ is the continuous line in Figure \ref{mod}.

For $\alpha=0$, the target observables are compatible and $\Icomp(\Ao,\Bo)=0$. For $\alpha\to 0$ the previous formulae give $u=0$, $\ell=1$, $\vec{c}(\gamma) = \vec{a} \equiv \vec{n}$, and one can check that also the lower bound
\eqref{s:c2min} vanishes, as it must be.

For two orthogonal components, i.e.\ $\alpha=\pi/2$, the expression
\eqref{s:c2min} gives the exact value \eqref{corth} of the entropic
incompatibility degree, and it is not only a lower bound. This value can be
computed by going to the limit $\alpha\to \pi/2$ in \eqref{s:c2min}, or directly by Remark \ref{rem:app} in Appendix \ref{sec:lb}.

For $\alpha\in (0,\pi/2)$,  Figure \ref{mod} shows that the analytical lower
bound \eqref{s:c2min} is not so far from the numerical value.

We now compare our optimal approximate joint measurements with other proposals coming from different approaches. Of course, every approximate joint
measurement $\Mo$ that is optimal with respect to some other criterium will
have a divergence from the target observables $(\Ao,\Bo)$ larger or equal than
$\Icomp(\Ao,\Bo)$. We stress that the other two  proposals we will consider
yield optimal bi-observables of the form $\Mo_\gamma$, in which however the
parameter $\gamma$ is different from ours.

We have seen that, when $\alpha=\pi/2$, the incompatibility degree of $\Ao$ and
$\Bo$, as well as their unique optimal approximate joint measurement $\Mo_0$,
can be evaluated analytically. In this special case, it turns out that $\Mo_0$
is optimal also with respect to the other criteria we are going to consider in
this section. However, as already said, this is not true for general $\alpha$. In order to show
it, we fix the angle $\alpha=\pi/4$, and compare the results of the different
criteria in Table \ref{tab}. We also add a $LB$ column summarizing the
parameters for the analytical lower bound \eqref{lbspin}. The rows provide: (1)
the parameter $\gamma$ characterizing the measurement $\Mo_\gamma$; (2) the angle characterizing the pure state $\rho(\phi)$ at which $\sddiv{\Ao,\Bo}{\Mo_\gamma}$ is computed, that is the trial angle $\pi/4-\alpha/2$ in the first column, and the angle maximizing $\sddiv{\Ao,\Bo}{\Mo_\gamma}(\rho(\phi))$ in the other ones; (3) the value of $\sddiv{\Ao,\Bo}{\Mo_\gamma}( \rho(\phi))$ for the parameters chosen in (1) and (2),
which gives $LB(\pi/4)$ in the first column and the entropic divergence
$\Div{\Ao,\Bo}{\Mo_\gamma}$ in the other ones.
\begin{table}[h]
\caption{Incompatibility degree and its bounds for $\alpha=\pi/4$.}\label{tab}
\centering
\begin{tabular}{|c||c|c|c|c|}
\hline
criterium& $LB$ & $\Icomp$ &
BLW& NV \\ \hline \hline
measurement: $\gamma\simeq$ & 0.795559 & 0.743999 & 0.541195 & 0.414213
\\ \hline
state: $\phi\simeq$ &  0.392699  &  0.282743 &0.391128 & 0.416889
\\ \hline
 value: $\sddiv{\Ao,\Bo}{\Mo_\gamma}( \rho(\phi))\simeq$& 0.110081 & 0.120035 &
0.160886 & 0.212079
\\ \hline
\end{tabular}
\end{table}
The description of the columns is as follows.

$LB$: The choice of the parameters is the one described in the computation of
the analytical lower bound for $\Icomp$. The parameter $\gamma$ comes from
\eqref{c2min}, the angle $\phi=\pi/8$ corresponds to the trial state  $\rho_e =  \rho(\pi/8)$
(i.e.\ the eigenprojection of $\vec a \cdot \vec \sigma$ for
$\alpha=\pi/4$), and the corresponding value of
$\sddiv{\Ao,\Bo}{\Mo_\gamma}( \rho_e)$ is the lower bound $LB(\pi/4)$.

$\Icomp$: The parameters are chosen following the relative entropy approach to
MURs.  They are the numerical solution of the minimax problem \eqref{eq:optispin2}. Thus,  the value of
$\sddiv{\Ao,\Bo}{\Mo_\gamma}( \rho(\phi))$ is the one found numerically for
$\Icomp(\Ao,\Bo)$, i.e.\ the dot at $\alpha = \pi/4$ in Figure \ref{mod};
$\gamma$ is the corresponding minimum point giving the optimal approximate joint
measurement $\Mo_\gamma$ of $\Ao$ and $\Bo$; the angle $\phi$ corresponds to
the state at which the error function $\sddiv{\Ao,\Bo}{\Mo_\gamma}$ attains its
maximum.

BLW: As discussed  in Section \ref{sec:MUR}, in
\cite{Wer04,BLW14a,BLW14b} a different approach is proposed.
In particular, its application to the case of two spin-1/2
components is given in \cite{BusLW14} (see also \cite{BusH08}, where the same final results are obtained in a slightly different context). There, the authors find a strictly positive
lower bound for the sum $d(\Ao,\Mo_{[1]})^2+d(\Bo,\Mo_{[2]})^2$, which holds for all
approximate joint measurements $\Mo$. Moreover, they find a
couple of compatible observables $(\Ao_{\vec{c}},\Bo_{\vec{c}})$ saturating the
lower bound, and thus optimally approximating the target observables $(\Ao,\Bo)$; this couple is given by a vector $\vec{c}$ yielding compatible $\Ao_{\vec{c}}$ and $\Bo_{\vec{c}}$, and lying as close as possible to the target direction
$\vec{a}$. Referring to Figure \ref{c+q} in Appendix \ref{ex:D4}, this amounts
to requiring that $\vec{c}$ is the orthogonal projection of $\vec{a}$ on the
right edge of the square $Q$ in the $\vec{i}\vec{j}$-plane; such a square is the
region of the plane where the approximating observables $\Ao_{\vec{c}}$ and
$\Bo_{\vec{c}}$ are compatible (see Proposition \ref{prop:cov}, item
(\ref{itemQ}), in Appendix \ref{ex:D4}). Using the parametrization
$\vec{c}(\gamma)$ given in \eqref{eq:rho(phi)c(gamma)} for the right edge of $Q$, we
see that this approach fixes $\gamma = \sqrt{2} a_2$.
The entropic divergence of the corresponding bi-observable $\Mo_\gamma$ from
$(\Ao,\Bo)$ and the angle of the state $\rho(\phi)$ at which it
is attained are the content of the BLW column.

NV: At the end of Section \ref{sec:noise}, we briefly discussed the proposal of
\cite{HeiMZ16,HSTZ14} to use approximating joint measurements whose marginals are
noisy versions (NV) of the two target observables. In this approach, one approximates the target observables by means of a compatible couple $(\Ao_{\vec{c}},\Bo_{\vec{c}})$ with $\vec{c}\parallel\vec{a}$. Still making reference to Figure \ref{c+q} in the appendix, the best choice is then picking $\vec{c}$ as large as possible; in this way, $\vec{c}$ lies where the right edge of the compatibility square $Q$ intersects the line joining $\vec{a}$ and the origin. With our parametrization $\vec{c}(\gamma)$ of the edge, this implies $\gamma = a_2/a_1$. The results for
this choice (together with the corresponding maximizing state) are reported in the last column.

\subsection{Two conjugate observables in prime power dimension}\label{sec:cov.MUBs}
We now consider two complementary observables in prime power dimension,
realized by a couple of MUBs that are conjugated by the Fourier transform of a
finite field. In general, the construction of a maximal set of MUBs in a prime power
dimensional Hilbert space by using finite fields is well known since Wootters
and Fields' seminal paper \cite{WoFi89}; see also \cite[Sect.\ 2]{DuEnBeZy10}
for a review, and \cite{BaBoRoVa02,Ap09,CaScTo15} for a group theoretic
perspective on the topic.

Let $\F$ be a finite field with characteristic $p$. We
refer to \cite[Sect.\ V.5]{LanAlg} for the basic notions on finite fields. Here
we only recall that $p$ is a prime number, and $\F$ has cardinality $|\F|=p^n$
for some positive integer $n$. We need also the field trace \ $\tr: \F\to\Zbb_p$ defined by  $\tr{x} = \sum_{k=0}^{n-1} x^{p^k}$ (see \cite[Sect.\ VI.5]{LanAlg} for its definition and properties).

We consider the Hilbert
space $\Hscr = \ell^2(\F)$, with dimension $d=p^n$, and we let our target observables be the two sharp rank-one observables $\Qo$ and $\Po$ with outcome spaces $\Xscr = \Yscr = \F$, given by
\begin{equation}\label{eq:defcoj}
\Qo(x) = \kb{\delta_x}{\delta_x} , \qquad \Po(y) = \kb{\omega_y}{\omega_y} , \qquad \forall x,y\in\F .
\end{equation}
In this formula, $\delta_x$ is the delta function at $x$, and
\begin{equation}\label{eq:Fourier}
\omega_y(z) = \frac{1}{\sqrt{d}}\,\rme^{\frac{2\pi \rmi}{p}\,\tr yz} \equiv (F^*\delta_y)(z) \quad \text{with} \quad
F\phi(z) = \frac{1}{\sqrt{d}} \sum_{t\in\F} \rme^{-\frac{2\pi \rmi}{p}\,\tr zt} \phi(t).
\end{equation}
Since $|\ip{\delta_x}{\omega_y}| = 1/\sqrt{d}$ for all $x$ and $y$, the two orthonormal bases $\{\delta_x\}_{x\in\F}$ and $\{\omega_y\}_{y\in\F}$ satisfy the MUB condition. In particular, as a consequence of the bound in \cite{MaaU88},
their preparation uncertainty coefficient \eqref{cprep} is
\begin{equation}\label{eq:MaaUff}
\Iprep(\Qo,\Po) = \log d .
\end{equation}

In \eqref{eq:Fourier}, the operator $F:\Hscr\to\Hscr$ is the unitary \emph{discrete Fourier transform} over the field $\F$. The observables $\Qo$ and $\Po$ are  then an example of \emph{Fourier conjugate MUBs}, as $\Po(y) = F^*\Qo(y) F$ for all $y\in\F$.

The definitions \eqref{eq:defcoj} and \eqref{eq:Fourier} should be
compared with the analogous ones for MUBs that are conjugated by means of the
Fourier transform over the cyclic ring $\Zbb_d$, see e.g.~\cite{CHT12}. In the
latter case, the Hilbert space is $\Hscr = \ell^2(\Zbb_d)$, and
the operator $F$ in \eqref{eq:Fourier} is replaced by
\begin{equation}\label{eq:Fourier_brutta}
\mathcal{F}\phi(z) = \frac{1}{\sqrt{d}} \sum_{t\in\Zbb} \rme^{-\frac{2\pi \rmi}{d}\,zt} \phi(t)
\end{equation}
(cf.~\cite[Eq.\ (4)]{CHT12}; note that no field trace appears in this formula).
The two definitions are clearly the same if $\F$ coincides with the cyclic
field $\Zbb_p$ (i.e.\ $n=1$ and so $d=p$), but they are essentially different in
general. Indeed, as observed in \cite[Sect.\ 5.3]{DuEnBeZy10}, they are
inequivalent already for $d=2^2$.

The following theorem is the main result of this section.
\begin{theorem}\label{teo:incoMUB}
For the two sharp observables $\Qo$ and $\Po$ defined in \eqref{eq:defcoj}, we have
\begin{equation}\label{eq:I(Q,P)}
\log\frac{2\sqrt{d}}{\sqrt{d}+1} \leq \Icomp(\Qo,\Po)  = \max_{\substack{\rho\in\Sscr(\Hscr)\\
    \rho\ \mathrm{pure}}} \left[\Srel{\Qo^\rho}{\Qo^\rho_{\lam_0}} + \Srel{\Po^\rho}{\Po^\rho_{\lam_0}}\right]
\leq 2\log\frac{2(d+1)}{d+3}\,,
\end{equation}
where $\Qo_{\lam_0} = \lam_0\Qo + (1-\lam_0)\Uo_\F$ and $\Po_{\lam_0} = \lam_0\Po + (1-\lam_0)\Uo_\F$ are the uniformly noisy versions of the observables $\Qo$ and $\Po$ with noise intensity
\begin{equation}\label{lambm}
1-\lam_0 = \frac{\sqrt{d}}{2(\sqrt{d}+1)} \,.
\end{equation}
An optimal approximate joint measurement $\Mo\in\Mcomp(\Ao,\Bo)$ is given by
\begin{equation}\label{eq:hatI(Q,P)2}
\Mo_0(x,y) = \frac{1}{2(d+\sqrt{d})} \kb{\psi_{x,y}}{\psi_{x,y}} \quad \text{with} \quad \psi_{x,y} = \delta_x + \rme^{-\frac{2\pi\rmi}{p}\,\tr xy} F \delta_{-y}.
\end{equation}
If $p\neq 2$, then $\Mo_0$ is the unique optimal approximate joint measurement of $\Qo$ and $\Po$, i.e.\  $\Mcomp(\Qo,\Po) = \{\Mo_0\}$.
\end{theorem}

As in the case of the two spin-1/2 components, the proof of this theorem relies on a detailed study of the symmetries of the pair of observables $(\Qo,\Po)$, and a subsequent application of Theorem \ref{prop:invarD^}. The symmetries and the proof of the theorem are given in Appendix \ref{app:MUB}. Here we  briefly comment on its statements and provide a simple example.

\begin{remark}\label{rem:uniqMUBS}
\begin{enumerate}
\item Since $\Qo$ and $\Po$ are sharp, the inequality \eqref{eq:I(Q,P)} also gives a bound for the index $\Iad(\Qo,\Po) = \Icomp(\Qo,\Po)$.
\item The two bounds in \eqref{eq:I(Q,P)} are not asymptotically optimal for $d\to\infty$, as the lower bound tends to $1$ while the upper bound goes to $2$.
\item The value in \eqref{lambm} is the minimal noise intensity making the two uniformly noisy observables $\Qo_{\lam_0}$ and $\Po_{\lam_0}$ compatible \cite[Prop.\ 5 and Ex.\ 1]{CHT12}.

\item In the terminology of \cite{BGL97,BLPY16}, the bi-observable in \eqref{eq:hatI(Q,P)2} is the {\em covariant phase-space observable} generated by the state $[\sqrt{d}/(2\sqrt{d}+2)] \kb{\psi_{0,0}}{\psi_{0,0}} = d\Mo_0(0,0)$ (see \eqref{Mtau} and the discussion below it for further details on covariant phase-space observables).
\item Our choice of using the field $\F$ instead of the ring $\Zbb_d$ in defining the Fourier operator in \eqref{eq:Fourier}, and the consequent restriction to only prime power dimensional systems, comes from the fact that the resulting MUBs \eqref{eq:defcoj} share dilational symmetries that are not present in the $\mathcal{F}$-conjugate ones. These extra symmetries drastically reduce the number of parameters to be optimized for finding an element of $\Mcomp(\Ao,\Bo)$ (see Remark \ref{rem:Fourier2}.\ref{rem:Fourier2.symmetry} for further details).
\item \label{it:uniqMUBS} The uniqueness property of the optimal approximate joint measurement $\Mo_0$ in odd prime power dimensions should be compared with the measurement uncertainty region for two qudit observables found in \cite[Sect.\ 5.3]{Wer16}. In particular, we remark that there is a whole family of covariant phase-space observables saturating the uncertainty bound of \cite[Eq.\ (38)]{Wer16}. Our optimal bi-observable $\Mo_0$ just corresponds to one of them, that is, the one generated by the state $\rho= d\Mo_0(0,0)$.
\item When $d=2^n$ with $n\geq 2$, it is not clear whether or not the set $\Mcomp(\Qo,\Po)$ is made up of a unique bi-observable. However, in the simplest case $d=2$ we have already shown that  $\Mcomp(\Qo,\Po) = \{\Mo_0\}$ (see Theorem \ref{teo:incoORTHSPIN}).
\end{enumerate}
\end{remark}

\begin{example}[Two orthogonal spin-1/2 components]
Let us consider as target observables the two sharp spin-1/2 components
$\Xo,\Yo\in\Mscr(\{+1,-1\})$ associated with the first two Pauli matrices,
defined in \eqref{def:XY}. This is the easiest example of two Fourier conjugate MUBs. To see this, take the cyclic field
$\F=\Zbb_2\equiv\{0,1\}$, corresponding to the choice $d=p=2$, $n=1$, $\tr
x=x$, and identify the observables $\Qo(x)=\Xo\left((-1)^x\right)$ and
$\Po(y)=\Yo\left((-1)^y\right)$ ($x,y=0,1$) by setting
$\sigma_1= |\delta_0\rangle \langle \delta_0|- |\delta_1\rangle \langle \delta_1|$,
and
$\sigma_2= |\delta_0\rangle \langle \delta_1|+ |\delta_1\rangle \langle \delta_1|$.
With this identification, the discrete Fourier transform becomes
$
F=\left(\sigma_1+\sigma_2\right)/\sqrt 2\equiv \rmi \exp\{-\rmi \pi
\vec  n \cdot \vec  \sigma/2\}$.
We have already found in \eqref{orthcov} the optimal joint observable of $\Xo$ and
$\Yo$, together with the value of the entropic incompatibility degree. These are precisely the bi-obser\-vable and the lower bound found in Theorem \ref{teo:incoMUB}.
\end{example}

\section{Entropic measurement uncertainty relations for $n$ observables}\label{sec:manyobs}

Uncertainty relations have been studied also in the case of more than two observables, see e.g.\ \cite{AAHB16,ColesBTW17,WehW10} for the case of entropic PURs. Both our entropic coefficients \eqref{def:inc} and \eqref{def:ed} (and the related MURs) can be generalized to the case of $n>2$ target observables.
However, in the case of
$\Iad(\Ao_1,\ldots,\Ao_n)$ an order of observation has to be fixed, and one
needs to point out the subset of the observables for which imprecise
measurements are allowed (the analogues of the observable $\Ao$ in the binary
case of $\Iad(\Ao,\Bo)$) from those observables that are kept fixed and get
disturbed by the other measurements (similar to $\Bo$ in $\Iad(\Ao,\Bo)$).
Thus, different definitions of $\Iad$ are possible in the $n$-observable case.
This leads us to generalize only the entropic incompatibility degree
$\Icomp(\Ao_1,\ldots,\Ao_n)$, whose definition is straightforward and gives a
lower bound for $\Iad$, independently of its possible definitions.

\subsection{Entropic incompatibility degree and MURs}

Let $\Ao_1,\ldots,\Ao_n$ be $n$ fixed observables with outcome sets
$\Xscr_1,\ldots,\Xscr_n$, respectively. As usual, we assume that all the sets
$\Xscr_i$ are finite. The observables with outcomes in the product set
$\Xscr_{1\cdots n}=\Xscr_1\times\cdots\times\Xscr_n$ are called
multi-observables, and we use the notation $\Mscr(\Xscr_{1\cdots n})$ for the
set of all such observables. If $\Mo\in\Mscr(\Xscr_{1\cdots n})$, its {\em
$i$-th marginal observable} is the element $\Mo_{[i]}\in\Mscr(\Xscr_i)$, with
\[
\Mo_{[i]}(x) = \sum_{x_j\in\Xscr_j : j\neq i} \Mo(x_1,\ldots,x_{i-1},x,x_{i+1},\ldots,x_n) .
\]
The notion of compatibility straightforwardly extends to the case of $n$
observables.

As in the $n=2$ case, we regard any $\Mo\in\Mscr(\Xscr_{1\cdots n})$ as an approximate joint measurement of $\Ao_1,\ldots$, $\Ao_n$. For all $\rho\in\Sscr(\Hscr)$, the total amount of information loss in the distribution approximations $\Ao_i^\rho\simeq\Mo_{[i]}^\rho$, $i=1,\ldots,n$, is the sum of the respective relative entropies. Then, we have the following generalization of Definitions \ref{def:erf}, \ref{def:D} and \ref{def:Cinc}.

\begin{definition}
For any multi-observable \ $\Mo\in\Mscr(\Xscr_{1\cdots n})$, \ the \emph{error function} of the approximation \\ $(\Ao_1,\ldots, \Ao_n)\simeq (\Mo_{[1]},\ldots,\Mo_{[n]})$ is the state-dependent quantity
\begin{equation}\label{statedependentdivergencen}
\sddiv{\Ao_1,\ldots,\Ao_n}{\Mo}(\rho)=\sum_{i=1}^n \Srel{\Ao^\rho_i}{\Mo^\rho_{[i]}}.
\end{equation}

The \emph{entropic divergence} of $\Mo\in\Mscr(\Xscr_{1\cdots n})$ from
$(\Ao_1,\ldots,\Ao_n)$ is
\begin{equation}
\Div{\Ao_1,\ldots,\Ao_n}{\Mo} = \sup_{\rho\in\Sscr(\Hscr)}
\sddiv{\Ao_1,\ldots,\Ao_n}{\Mo}(\rho).
\end{equation}

The \emph{entropic incompatibility degree} of the observables
$\Ao_1,\ldots,\Ao_n$ is
\begin{equation}\label{eq:multi_cinc}
\Icomp\big(\Ao_1,\ldots,\Ao_n\big)=\inf_{\Mo\in\Mscr(\Xscr_{1\cdots n})}
\Div{\Ao_1,\ldots,\Ao_n}{\Mo} .
\end{equation}
\end{definition}

We still denote by
$$
\Mcomp(\Ao_1,\ldots,\Ao_n) = \argmin_{\Mo\in\Mscr(\Xscr_{1\cdots n})}
\Div{\Ao_1,\ldots,\Ao_n}{\Mo}
$$
the set of the \emph{optimal approximate joint measurements} of
$\Ao_1,\ldots,\Ao_n$. As in the case with $n=2$, the optimality of a
multi-observable $\Mo$ depends only on its marginals $\Mo_{[i]}$, since the
entropic divergence itself depends only on such marginals.

We have the
following extension of Theorems \ref{prop:2A}, \ref{prop:2B} and \ref{prop:3}.
\begin{theorem}\label{prop:2n}
Let $\Ao_i\in\Mscr(\Xscr_i)$, $i=1,\ldots,n$, be the target observables. The
error function, entropic divergence and incompatibility degree satisfy the following
properties.
\begin{enumerate}[(i)]
\item \label{slscn} The function $\sddiv{\Ao_1,\ldots,\Ao_n}{\Mo}
    :\Sscr(\Hscr)\to[0,+\infty]$ is convex and LSC, $\forall \Mo\in
    \Mscr(\Xscr_{1\cdots n})$.

\item \label{Dlscn} The function $\Div{\Ao_1,\ldots,\Ao_n}{\cdot} :
    \Mscr(\Xscr_{1\cdots n})\to [0,+\infty]$ is convex and LSC.

\item \label{Dfiniten} For any $\Mo\in\Mscr(\Xscr_{1\cdots n})$, the following three statements are equivalent:
\begin{enumerate}[(a)]
\item $\Div{\Ao_1,\ldots,\Ao_n}{\Mo}<+\infty$,
\item
    $\operatorname{ker}\Mo_{[i]}(x)\subseteq\operatorname{ker}\Ao_i(x)\quad
    \forall x, \ \forall i$,
\item $\sddiv{\Ao_1,\ldots,\Ao_n}{\Mo}$ is bounded and continuous.
\end{enumerate}

\item \label{worstrhon}
    $\displaystyle\Div{\Ao_1,\ldots,\Ao_n}{\Mo}=\max_{\rho\in\Sscr(\Hscr), \
    \rho\ \mathrm{pure}} \sddiv{\Ao_1,\ldots,\Ao_n}{\Mo}(\rho)$,
    where the maximum can be any value in the extended interval $[0,+\infty]$.

\item \label{cinvn} The quantities $\sddiv{\Ao,\Bo}{\Mo}(\rho)$, $\Div{\Ao_1,\ldots,\Ao_n}{\Mo}$ and $\Icomp(\Ao_1,\ldots,\Ao_n)$ are invariant under an overall unitary conjugation of the state $\rho$ and the observables $\Ao_1,\ldots,\Ao_n$ and $\Mo$, and they do not depend on the labelling of the outcomes in $\Xscr_1,\ldots,\Xscr_n$.

\item \label{cinvn2} $\Icomp(\Ao_{\sigma(1)},\ldots,\Ao_{\sigma(n)}) =
    \Icomp(\Ao_1,\ldots,\Ao_n)$ for any permutation $\sigma$ of the index
    set $\{1,\ldots,n\}$.

\item \label{cboundn} The entropic incompatibility coefficient
    $\Icomp(\Ao_1,\ldots,\Ao_n)$ is always finite, and it satisfies the bounds
\begin{equation}\label{eq:trivboundCprepn}
\Icomp(\Ao_1,\ldots,\Ao_n) \leq \sum_{i=1}^n\log|\Xscr_i| - \inf_{\rho\in\Sscr(\Hscr)} \sum_{i=1}^n H(\Ao_i^\rho),
\end{equation}
\begin{multline}\label{eq:bound3n}
\Icomp(\Ao_1,\ldots,\Ao_n) \leq
\sum_{i=1}^n \log \frac{n(d+1)}{d+n+(n-1)\min_{x\in\Xscr_i} \Tr\{\Ao_i (x)\}} \\
{}\leq n\log \frac{n(d+1)}{d+n}\leq n\log n .
\end{multline}

\item \label{optMn} The set $\Mcomp(\Ao_1,\ldots,\Ao_n)$ is a nonempty
    convex compact subset of $\Mscr(\Xscr_{1\cdots n})$.

\item \label{jmeasn} $\Icomp(\Ao_1,\ldots,\Ao_n)=0$ if and only if the
    observables $\Ao_1,\ldots,\Ao_n$ are compatible, and in this case
    $\Mcomp(\Ao_1,\ldots,\Ao_n)$ is the set of all the joint measurements
    of $\Ao_1,\ldots,\Ao_n$.

\item \label{cnondecn} If $\Ao_{n+1}\in\Mscr(\Xscr_{n+1})$ is another
    observable, then we have $\Icomp(\Ao_1,\ldots,\Ao_{n+1}) \geq
    \Icomp(\Ao_1,\ldots,\Ao_n)$.
\end{enumerate}
\end{theorem}
\begin{proof}
The proofs of items (\ref{slscn})--(\ref{cinvn2}), (\ref{optMn}) and
(\ref{jmeasn}) are straightforward extensions of the analogous ones for two
observables.

In  item (\ref{cboundn}), the upper bound \eqref{eq:trivboundCprepn}  follows by evaluating the entropic divergence of the uniform observable $\Uo =
(u_{\Xscr_1}\otimes\cdots\otimes u_{\Xscr_n})\,\id$ from $(\Ao_1,\ldots,\Ao_n)$:
\begin{align*}
\Icomp(\Ao_1,\ldots,\Ao_n) & \leq \Div{\Ao_1,\ldots,\Ao_n}{\Uo} = \sup_{\rho\in\Sscr(\Hscr)}
\sum_{i=1}^n \Srel{\Ao^\rho_i}{u_{\Xscr_i}} \\
& = \sup_{\rho\in\Sscr(\Hscr)}
\sum_{i=1}^n \left[\log|\Xscr_i| - H(\Ao^\rho_i)\right] \qquad
\text{by Proposition \ref{prop:propHSrel}, item (\ref{HSrel});}
\end{align*}
this yields \eqref{eq:trivboundCprepn}.

The upper bound \eqref{eq:bound3n} follows by using an
approximate cloning argument, just as in the case of only two observables. Indeed, the optimal approximate $n$-cloning channel is the map
$$
\Phi : \sh\to\Scal(\Hscr^{\otimes n}), \qquad \Phi(\rho) = \frac{d!n!}{(d+n-1)!}
S_n (\rho\otimes\id^{\otimes (n-1)}) S_n ,
$$
where $S_n$ is the orthogonal projection of $\Hscr^{\otimes n}$ onto its
symmetric subspace ${\rm Sym}(\Hscr^{\otimes n})$ \cite{KW99}. Evaluating the
marginals of the multi-observable $\Mo_{\rm cl} =
\Phi^*(\Ao_1\otimes\cdots\otimes\Ao_n)$, we obtain the noisy versions
$$
\Mo_{{\rm cl}\,[i]} = \Ao_{i\,\lam_{\rm cl}}, \qquad \text{where} \qquad \lam_{\rm cl} = \frac{d+n}{n(d+1)}
$$
(see \cite{Wer98}). Since $\Icomp(\Ao_1,\ldots,\Ao_n) \leq \Div{\Ao_1,\ldots,\Ao_n}{\Mo_{\rm cl}}$, a computation similar to the one for obtaining the bound \eqref{eq:bound3} in Section \ref{sec:noise} then yields the bounds \eqref{eq:bound3n}.

Finally, in order to prove item (\ref{cnondecn}), take any
$\Mo'\in\Mscr(\Xscr_1\times\cdots\times\Xscr_{n+1})$, and let
$$
\Mo(x_1,\ldots,x_n) = \sum_{x\in\Xscr_{n+1}}
\Mo'(x_1,\ldots,x_n,x) .
$$
We have $\Mo'_{[i]} = \Mo_{[i]}$ for all $i=1,\ldots,n$, hence
\begin{multline*}
\Icomp(\Ao_1,\ldots,\Ao_n)  \leq \Div{\Ao_1,\ldots,\Ao_n}{\Mo} =
\sup_\rho \sum_{i=1}^n \Srel{\Ao_i^\rho}{\Mo^\rho_{[i]}}\\ {} \leq \sup_\rho \sum_{i=1}^{n+1}
\Srel{\Ao_i^\rho}{\Mo^{\prime\rho}_{[i]}}
= \Div{\Ao_1,\ldots,\Ao_{n+1}}{\Mo'}.
\end{multline*}
Item (\ref{cnondecn}) then follows by taking the infimum over $\Mo'$.
\end{proof}

The monotonicity property (\ref{cnondecn}), which is specific of the many
observable case, is another desirable feature for an incompatibility
coefficient: the amount of incompatibility cannot decrease when an extra
observable is added.

\begin{remark}[MURs]
Theorem \ref{prop:2n} gives the following extension of the entropic MURs \eqref{MUR1} and \eqref{MUR2}:
\begin{equation}
\Div{\Ao_1,\ldots,\Ao_n}{\Mo} \geq \Icomp(\Ao_1,\ldots,\Ao_n), \qquad \forall \Mo \in \Mscr(\Xscr_{1\cdots n}),
\end{equation}
\begin{equation}
\forall \Mo \in \Mscr(\Xscr_{1\cdots n}), \ \ \exists\rho\in \Sscr(\Hscr) : \ \  \sum_{i=1}^n \Srel{\Ao_i^\rho}{\Mo^\rho_{[i]}}\geq\Icomp(\Ao_1,\ldots,\Ao_n).
\end{equation}
\end{remark}

Finally, suppose the product space $\Xscr_{1\cdots n}$ carries the action of a
finite symmetry group $G$, which also acts on the quantum system Hilbert space
$\Hscr$ by means of a projective unitary representation $U$. These actions then
extend to the set of states $\sh$, the set of probabilities
$\Pscr(\Xscr_{1\cdots n})$ and the set of multi-observables
$\Mscr(\Xscr_{1\cdots n})$ exactly as in Section \ref{sec:cov.gen}. Similarly,
for any $\Mo\in\Mscr(\Xscr_{1\cdots n})$, we can define its covariant version
$\Mo_G$. Then, the content of Theorem \ref{prop:invarD^} and Corollary
\ref{prop:invarI^} can be translated to the case of $n$ observables as
follows.

\begin{theorem}\label{prop:invarD+I^n}
Let $\Ao_i\in\Mscr(\Xscr_i)$, $i=1,\ldots,n$, be the target observables.
Suppose the finite group $G$ acts on both the output space $\Xscr_{1\cdots n}$
and the index set $\{1,\ldots,n\}$, and it also acts with a projective unitary
representation $U$ on $\Hscr$. Moreover, assume that $G$ is generated by a
subset $S_G\subseteq G$ such that, for every $g\in S_G$ and
$i\in\{1,\ldots,n\}$, there exists a bijective map $f_{g,i} :
\Xscr_i\to\Xscr_{gi}$ for which
\begin{enumerate}[(a)]
\item $g(x_1,\ldots,x_n)_{gi} = f_{g,i}(x_i)$  for all
    $(x_1,\ldots,x_n)\in\Xscr_{1\cdots n}$,
\item  $U_g\Ao_i(x_i)U_g^* = \Ao_{gi}(f_{g,i}(x_i))$ for all
    $x_i\in\Xscr_i$.
\end{enumerate}
Then,
\begin{enumerate}[-]
\item $\Div{\Ao_1,\ldots,\Ao_n}{g\Mo} = \Div{\Ao_1,\ldots,\Ao_n}{\Mo}$ for
    all $\Mo\in\Mscr(\Xscr_{1\cdots n})$ and $g\in G$;
\item the set $\Mcomp(\Ao_1,\ldots,\Ao_n)$ is $G$-invariant;

\item for any $\Mo\in\Mcomp(\Ao_1,\ldots,\Ao_n)$, we have
    $\Mo_G\in\Mcomp(\Ao_1,\ldots,\Ao_n)$;
\item there exists a $G$-covariant observable in
    $\Mcomp(\Ao_1,\ldots,\Ao_n)$.
\end{enumerate}
\end{theorem}
\begin{proof}
As in the proof of Theorem \ref{prop:invarD^}, it is not restrictive to assume
that $S_G = G$. For all $p\in\Pscr(\Xscr_{1\cdots n})$, condition (a) implies
\begin{align*}
gp_{[i]} (x_i) & = \sum_{\substack{x_j\in\Xscr_j\\ \text{s.t.~} j\neq i}} gp(x_1,\ldots,x_n) = \sum_{\substack{x_{gj}\in\Xscr_{gj}\\ \text{s.t.~} gj\neq i}} p(f_{g^{-1},g1}(x_{g1}),\ldots,f_{g^{-1},gn}(x_{gn}))
\\
& = \sum_{\substack{y_j\in\Xscr_j\\ \text{s.t.~} gj\neq i}} p(y_1,\ldots,y_n) \qquad \qquad \text{where $y_j = f_{g^{-1},gj}(x_{gj})$}
\\
& = \sum_{\substack{y_j\in\Xscr_j\\ \text{s.t.~} j\neq g^{-1}i}} p(y_1,\ldots,y_n) = p_{[g^{-1}i]} (y_{g^{-1}i}) \\
& = p_{[g^{-1}i]}(f_{g^{-1},i}(x_i)) ,
\end{align*}
and hence
\begin{multline*}
(gp_{[1]}\otimes\cdots\otimes gp_{[n]})(x_1,\ldots,x_n) =
\prod_{i=1}^n gp_{[i]}(x_i) = \prod_{i=1}^n p_{[g^{-1}i]}(f_{g^{-1},i}(x_i)) \\ {}=
\prod_{i=1}^n p_{[i]}(f_{g^{-1},gi}(x_{gi}))  = g(p_{[1]}\otimes\cdots\otimes p_{[n]})(x_1,\ldots,x_n) .
\end{multline*}
Therefore,
\begin{equation}\label{eq:invarDn1}
gp_{[1]}\otimes\cdots\otimes gp_{[n]} = g(p_{[1]}\otimes\cdots\otimes p_{[n]}) .
\end{equation}
On the other hand, by condition (b) we have $ \Ao^{g\rho}_i (x_i) =
\Ao^\rho_{g^{-1}i} (f_{g^{-1},i}(x_i)) $, and then
\begin{align*}
& (\Ao_1^{g\rho}\otimes\cdots\otimes\Ao_n^{g\rho})(x_1,\ldots,x_n) =
\prod_{i=1}^n \Ao^\rho_{g^{-1}i} (f_{g^{-1},i}(x_i)) =
\prod_{i=1}^n \Ao^\rho_i (f_{g^{-1},gi}(x_{gi})) \\
& \qquad \qquad \qquad \qquad = g(\Ao_1^\rho\otimes\cdots\otimes\Ao_n^\rho)(x_1,\ldots,x_n) ,
\end{align*}
that is,
\begin{equation}\label{eq:invarDn2}
\Ao_1^{g\rho}\otimes\cdots\otimes\Ao_n^{g\rho} = g(\Ao_1^\rho\otimes\cdots\otimes\Ao_n^\rho) .
\end{equation}
Having established \eqref{eq:invarDn1} and \eqref{eq:invarDn2}, the proof of the equality
$\Div{\Ao_1,\ldots,\Ao_n}{g\Mo} = \Div{\Ao_1,\ldots,\Ao_n}{\Mo}$ follows along
the same lines of the proof of Theorem \ref{prop:invarD^}. The remaining
statements are then proved as in Corollary \ref{prop:invarI^}. \end{proof}

In the next section we will use Theorem \ref{prop:invarD+I^n} to solve the case
of $n=3$ orthogonal target spin-1/2 components. This is the basic example of a maximal set of $d+1$ MUBs in a $d$-dimensional Hilbert space. It is an open problem whether
similar arguments lead to find the incompatibility index
$\Icomp(\Qo_1,\ldots,\Qo_{d+1})$ of a maximal set of $d+1$ MUBs
$\Qo_1,\ldots,\Qo_{d+1}$ whenever such a set of MUBs is known to exist, that is,
for all prime powers $d$.

\subsection{Three orthogonal spin-1/2 components}\label{ex:3spin}
Let the target observables $\Ao_1$, $\Ao_2$, $\Ao_3$ be
three mutually orthogonal spin-1/2 components, that is, the sharp observables
$\Xo$, $\Yo$, $\Zo\in\Mscr(\{+1,-1\})$ associated with the three Pauli
matrices; the observables $\Xo,\Yo$ are given in \eqref{def:XY}, and
\begin{equation}\label{eq:Z}
\Zo(z) = \frac{1}{2}\left(\id+z\sigma_3\right), \qquad \forall z\in\Zscr=\{+1,-1\} .
\end{equation}
Then, we have the following three-spin version of Theorem \ref{teo:incoORTHSPIN}.
\begin{theorem}\label{teo:incoORTH333SPIN}
Let $\Xo$, $\Yo$ and $\Zo$ be the three orthogonal spin-1/2 components \eqref{def:XY}, \eqref{eq:Z}. Then, for the following two tri-observables $\Mo_0, \Mo_1 \in \Mscr(\Xscr\times\Yscr\times\Zscr)$
\begin{gather}
\Mo_0(x,y,z) = \frac{1}{8} \left[\id+\frac{1}{\sqrt 3} (x\sigma_1+y\sigma_2+z\sigma_3)\right], \label{eq:M0x3}\\
\begin{gathered}
\Mo_1(1,1,-1)=2\Mo_0(1,1,-1),\quad
\Mo_1(1,-1,1)=2\Mo_0(1,-1,1),\\
\Mo_1(-1,1,1)=2\Mo_0(-1,1,1),\quad
\Mo_1(-1,-1,-1)=2\Mo_0(-1,-1,-1),\\
\Mo_1(x,y,z)=0 \quad \text{ otherwise,}
\end{gathered}\label{eq:tildeM0x3}
\end{gather}
we have $\Mo_0,\Mo_1\in\Mcomp(\Xo,\Yo,\Zo)$.
If $\rho_e$ is the projection on any eigenvector of $\sigma_1$, $\sigma_2$ or $\sigma_3$, then, for $i=0,1$,
\begin{equation}\label{corth333}
\Icomp(\Xo,\Yo,\Zo)=\sddiv{\Xo,\Yo,\Zo}{\Mo_i}(\rho_{e})=\log \left(3-\sqrt 3\right) \simeq 0.342497\,.
\end{equation}
\end{theorem}
The description of the symmetry group of $\Xo$, $\Yo$ and $\Zo$, and the consequent application of Theorem \ref{prop:invarD+I^n} yielding the proof of Theorem \ref{teo:incoORTH333SPIN}, are provided in Appendix \ref{App:B3}.

Note that, differently from the case with only $n=2$ spins, for $n=3$ orthogonal spin-1/2 components there is not a unique optimal approximate joint measurement. It should be also remarked that, although $\Mo_0 \neq \Mo_1$, both optimal approximate joint measurements given in \eqref{eq:M0x3}, \eqref{eq:tildeM0x3} have the same marginals, so that $\sddiv{\Xo,\Yo,\Zo}{\Mo_0}=\sddiv{\Xo,\Yo,\Zo}{\Mo_1}$. Indeed, they are the equally noisy observbles
\begin{equation}\label{eq:noisy3margins}
\Mo_{0[1]} = \Mo_{1[1]} = \Xo_{\frac{1}{\sqrt{3}}} , \qquad \Mo_{0[2]} = \Mo_{1[2]} = \Yo_{\frac{1}{\sqrt{3}}} , \qquad \Mo_{0[3]} = \Mo_{1[3]} = \Zo_{\frac{1}{\sqrt{3}}} .
\end{equation}
The construction of the tri-observable $\Mo_1$ is taken from \cite[Sect.~VI]{HeJiReZi10}. It is an open question whether  $\left(\Xo_{1/\sqrt{3}}, \Yo_{1/\sqrt{3}}, \Zo_{1/\sqrt{3}}\right)$ is the unique triple of compatible observables optimally approximating $(\Xo, \Yo, \Zo)$; see also Remark \ref{rem:B4} in Appendix \ref{App:B3} for further comments.

\section{Conclusions }\label{sec:CO}

We have formulated and proved entropic MURs for discrete observables in a finite-dimensional Hilbert space. In doing so, we have considered target observables $\Ao$ and $\Bo$ described by general POVMs, not only sharp observables.
Our formulation employs the relative entropy to quantify the total amount of information that is lost when $\Ao$ and $\Bo$ are approximated with the marginals $\Mo_{[1]}$ and $\Mo_{[2]}$ of a bi-observable $\Mo$. Such an information loss is the state-dependent error \eqref{eq:Soprod}; maximizing it over all states $\rho$ and then minimizing the result over the bi-observables $\Mo$, we have derived our MURs \eqref{MUR2}: for every approximating bi-observable $\Mo$, there is always a state $\rho$ such that the total information loss of the approximation $(\Ao^\rho,\Bo^\rho) \simeq (\Mo^\rho_{[1]},\Mo^\rho_{[2]})$ is not less than a minimal threshold $c(\Ao,\Bo)$, independent of $\Mo$ and strictly positive when the target observables are incompatible.

Minimizing over different sets of bi-observables yields MURs with different meanings. If $\Mo$ varies over all the POVMs on the product set of the $\Ao$- and $\Bo$-outcomes, then the resulting index $\Icomp(\Ao,\Bo)$ is the minimal error potentially affecting all possible approximate joint measurements of $\Ao$ and $\Bo$. On the other hand, if $\Mo$ is only allowed to vary over the subset of all the sequential measurements of an approximation of $\Ao$ followed by $\Bo$, we obtain the minimal information loss $\Iad(\Ao,\Bo)$ due to the error/disturbance tradeoff. We have proved that the two indexes remarkably coincide when the second observable $\Bo$ is sharp; we have also given explicit examples, involving general POVMs, where the two indexes actually differ.

The two coefficients $\Icomp$ and $\Iad$ play a double role: on the one side, they are the lower bounds of our entropic MURs, as just described above; on the other side, they also properly quantify the degree of (total or sequential) incompatibility of the target observables. The latter interpretation is justified since $\Icomp$ and $\Iad$ only depend on $\Ao$ and $\Bo$, as well as by the remarkable properties of the two indexes (Theorem \ref{prop:2B}). In particular, the existence of an index allows to establish whether a couple of observables is more or less incompatible than another one. For instance, the incompatibility degree of two spin-1/2 components grows by increasing the angle between their directions, as naturally expected (see Figure \ref{mod}).

Due to the double optimization in the definitions of $\Icomp(\Ao,\Bo)$ and $\Iad(\Ao,\Bo)$, it is not easy to explicitly compute them and their corresponding optimal approximate joint measurements. Anyway, in Theorem \ref{prop:invarD^} we have shown how one can use general symmetry arguments in order to simplify the problem. We have then applied this method to two spin-1/2 components (Theorems \ref{teo:incoORTHSPIN} and \ref{teo:incoNONORTHSPIN}), and two Fourier conjugate MUBs in prime power dimension (Theorem \ref{teo:incoMUB}).

A peculiar feature of our MURs is that in several cases there is actually a unique optimal approximate
joint measurement. Indeed, in the two spin and MUB examples, we have uniqueness for all the cases in which we have managed to completely characterize the sets $\Mcomp$ and $\Mad$ of the optimally approximating joint measurements.
We conjecture that this is still true also for the two nonorthogonal spin-1/2 components, and the Fourier conjugated MUBs in even prime power dimensions, for which up to now we have only partial results.

One nice aspect of our approach is that it naturally and easily generalizes to more than two observables. We have done this extension only for the entropic incompatibility degree $\Icomp$ (Theorems \ref{prop:2n} and \ref{prop:invarD+I^n}), as the multi-observable interpretation of $\Iad$ is less transparent. As an application, we have computed the index $\Icomp$ for three orthogonal spin-1/2 components $\Xo$, $\Yo$ and $\Zo$; although in this case there is still a unique covariant approximate joint measurement, the main difference with the two spin case is that $\Mcomp(\Xo,\Yo,\Zo)$ is not a singleton set now.

Many problems still remain open, as it is not clear how to analytically or at least numerically compute $\Icomp$ and $\Iad$ and the corresponding optimal approximate joint measurements for an arbitrary couple of target observables. Explicit results would be desirable for physically relevant observables other than those considered in Sections \ref{ex:D2}, \ref{sec:cov.MUBs} and \ref{ex:3spin} (e.g.\ two or more spin-$s$ components with $s>1/2$, two or more MUBs in arbitrary dimensions and possibly not Fourier conjugate, etc.). A possible generalization is to include also systems in presence of ``quantum memories''; indeed, this extension has recently been studied in the case of  entropic PURs \cite{Berta+,Frank,ColesBTW17}. More importantly, the theory we have developed is restricted to discrete observables in a finite-dimen\-sional Hilbert space. The bound $\Icomp(\Ao_1,\ldots,\Ao_n)\leq n\log n$ appearing in \eqref{eq:bound3} and \eqref{eq:bound3n}, which is independent of the number of the outcomes and the dimension of $\Hscr$, suggests that it would be possible to generalize the theory to arbitrary observables in a separable Hilbert space. However, this is not a straightforward extension; indeed, the first results on position and momentum \cite{BGT17} already show that the error function \eqref{statedependentdivergence} needs to be restricted to only particular classes of states,  in order to avoid $\Icomp=+\infty$, merely due to classical effects.

\appendix

\section{Examples of compatible but not sequentially compatible observables}\label{app:HW}

\paragraph{First example from \cite{HeiW10}} Apart from an exchange of $\Ao$
and $\Bo$ and some explicit computations, this example is taken from
\cite[Sect.\ III.C, and the end of Sect.\ III.A] {HeiW10}. With $\Hscr=\Cbb^3$,
$\Xscr=\{1,2\}$ and $\Yscr=\{1,\ldots,5\}$, the two target observables are
defined by
\[
\Ao(1)=\frac 12 \begin{pmatrix} 2&0&0 \\ 0&0&0 \\ 0&0& 1\end{pmatrix}, \qquad
\Ao(2)=\frac 12 \begin{pmatrix} 0&0&0 \\ 0&2&0 \\ 0&0& 1\end{pmatrix};
\]
\[
\Bo(1)=\frac 14 \begin{pmatrix} 2&0&-\sqrt 2 \\ 0&0&0 \\ -\sqrt 2&0& 1\end{pmatrix}, \quad
\Bo(2)=\frac 1{10} \begin{pmatrix} 0&0&0 \\ 0&1&-2 \\ 0&-2& 4\end{pmatrix}, \quad
\Bo(3)=\frac 12 \begin{pmatrix} 0&0&0 \\ 0&1&0 \\ 0&0& 0\end{pmatrix},
\]
\[
\Bo(4)=\frac 1{10} \begin{pmatrix} 0&0&0\\ 0&4&2 \\ 0 &2& 1\end{pmatrix}, \qquad
\Bo(5)=\frac 14 \begin{pmatrix} 2&0&\sqrt 2\\ 0&0&0 \\ \sqrt 2&0& 1\end{pmatrix}.
\]
These two observables are compatible, and one can check that a joint observable is
\[
\Mo(1,1)=\Bo(1), \quad \Mo(1,5)=\Bo(5), \quad \Mo(2,2)=\Bo(2), \quad \Mo(2,3)=\Bo(3), \]
\[
\Mo(2,4)=\Bo(4),
\\ \qquad
\Mo(1,2)=\Mo(1,3)=\Mo(1,4)=\Mo(2,1)=\Mo(2,5)=0.
\]
This implies
$\Icomp(\Ao,\Bo)=0$. Moreover, in \cite[Sect.\ III.C]{HeiW10} it is proved
that: (1) there exists an instrument implementing $\Bo$ which does not disturb
$\Ao$; (2) any instrument implementing $\Ao$ disturbs $\Bo$. By Theorem \ref{prop:2B}, item
(\ref{smeas}), this implies $\Iad(\Bo,\Ao)=0$ and
$\Iad(\Ao,\Bo)>0$.

\paragraph{Second example  from \cite{HeiW10}}
This is the first example of \cite[Sect.\ III.A]{HeiW10}, which we report in
the particular case in which the noise parameters are fixed and equal; let us call them
$\lambda$, with $\lambda\in \left(\frac 12\,,\, \frac 23\right]$. The
observables are two-valued $(\Xscr=\Yscr=\{1,2\})$, and they are built up by using
two noncommuting orthogonal projections $P$ and $Q$: \ $[P,Q]\neq 0$.  The
joint observable $\Mo$ and its marginals are given by
\[
\Mo(1,1)=(1-\lambda)\id, \qquad \Mo(1,2)=(2\lambda-1 )P,
\qquad \Mo(2,1)=(2\lambda-1 )Q,
\]
\[
\Mo(2,2)=\left(1-\frac 32\,\lambda\right)(P+Q)+\frac\lambda 2 \left(\id -P+\id -Q\right),
\]
\[
\Ao(1)=
\Mo_{[1]}(1)=\lambda P+(1-\lambda)(\id-P),
\quad
\Ao(2)=
\Mo_{[1]}(2)=\lambda (\id-P)+(1-\lambda)P,
\]
\[
\Bo(1)=
\Mo_{[2]}(1)=\lambda Q+(1-\lambda)(\id-Q), \quad
\Bo(2)=
\Mo_{[2]}(2)=\lambda (\id-Q)+(1-\lambda)Q.
\]
The observables $\Ao$ and $\Bo$ are compatible by construction, and so
$\Icomp(\Ao,\Bo)=0$. In \cite{HeiW10}, it is proved that there does not exist
any instrument implementing $\Ao$ which does not disturb $\Bo$; it follows that
$\Iad(\Ao,\Bo)>0$ and, by exchanging $P$ and $Q$, $\Iad(\Bo,\Ao)>0$.

\section{Symmetries and proofs for target spin-1/2 components}\label{app:spin}

In this appendix, we describe the symmetry groups for two arbitrary and three orthogonal spin-1/2 components. Then, by using Theorem \ref{prop:invarD^}, we prove our main Theorems \ref{teo:incoNONORTHSPIN} (Appendix \ref{ex:D4}), \ref{teo:incoORTHSPIN} (Appendix \ref{app:orth}) and \ref{teo:incoORTH333SPIN} (Appendix \ref{App:B3}), and we provide the missing calculations in Section \ref{sez:nuova}. Since the proof of Theorem \ref{teo:incoORTHSPIN} follows from Theorem \ref{teo:incoNONORTHSPIN} with the angle $\alpha = \pi/2$, here we prefer to reverse the order of the two proofs.

\subsection{Incompatibility degree and optimal measurements for two spin-1/2 components}\label{ex:D4}

In this section, $\Ao$ and $\Bo$ are the spin-1/2 components defined in \eqref{ABspin}, with directions spanning an arbitrary angle $\alpha\in [0,\pi/2]$; the respective outcome spaces are $\Xscr = \Yscr = \{-1,+1\}$. The symmetry group of $\Ao$ and $\Bo$ is the order $4$ dihedral group $D_2\subset SO(3)$ generated by the
rotations $S_{D_2}=\{R_{\vec  n}(\pi),\,R_{\vec  m}(\pi)\}$, i.e.\ the
$180^\circ$ rotations around  the bisectors $\vec n$ and $\vec  m$
of the first two quadrants (see Figure \ref{abfig}). Here and in the following, our reference for the discrete subgroups of the
rotation group is \cite[pp.~77--79]{Weyl}. The natural action of the group
$D_2$ on the outcome space $\Xscr\times\Yscr$ is given by
\begin{subequations}\label{2rot}
\begin{equation}\label{action1}
R_{\vec  n}(\pi)\, (x,y) = (y,x), \quad R_{\vec  m}(\pi)\, (x,y) =
(-y,-x), \qquad \forall(x,y)\in\Xscr\times\Yscr .
\end{equation}
We then see that condition (\ref{8ii}.a) of Theorem \ref{prop:invarD^} is
satisfied for all $g\in S_{D_2}$. As the representation $U$ of $D_2$ on $\Cbb^2$, we take the restriction of the usual spin-1/2 projective representation of $SO(3)$; this gives
\begin{equation}\label{actionU}
U\big(R_{\vec
n}(\pi)\big)=\rme^{-\rmi\pi\, {\vec  n}\cdot \vec  \sigma/2}\equiv
-\rmi\, {\vec  n}\cdot \vec  \sigma, \quad U\big(R_{\vec
m}(\pi)\big)=\rme^{-\rmi\pi\, {\vec  m}\cdot \vec  \sigma/2}\equiv
-\rmi\, {\vec  m}\cdot \vec  \sigma.
\end{equation}
\end{subequations}
It is easy to see that the observables $\Ao$ and $\Bo$ satisfy the relations
\begin{equation}\label{Un}
\begin{split}
& U\big(R_{\vec  n}(\pi)\big) \Ao(x) U\big(R_{\vec  n}(\pi)\big)^*
= \Bo(x), \quad   U\big(R_{\vec  m}(\pi)\big)\Ao(x) U\big(R_{\vec  m}(\pi)\big)^*
= \Bo(-x),
\\
&  U\big(R_{\vec  n}(\pi)\big) \Bo(y) U\big(R_{\vec  n}(\pi)\big)^*
= \Ao(y),  \quad  U\big(R_{\vec  m}(\pi)\big) \Bo(y) U\big(R_{\vec  m}(\pi)\big)^*
= \Ao(-y).
\end{split}
\end{equation}
This implies that also condition (\ref{8ii}.b) of Theorem \ref{prop:invarD^} is
fulfilled for all $g\in S_{D_2}$. Then, because of Remark \ref{rem:covariant},
in order to find $\Icomp(\Ao,\Bo)$, we are led to study the most general form
of a $D_2$-covariant bi-observable and its marginals.

\begin{proposition}\label{prop:cov}
Let the dihedral group $D_2$ act on $\Xscr\times\Yscr$ and $\Hscr$ as in
\eqref{2rot}. Then, the following facts hold.
\begin{enumerate}[(i)]
\item \label{itemM} The most general $D_2$-covariant bi-observable on $\Xscr\times\Yscr$
    is
\begin{equation}\label{gencov}
\Mo(x,y)
= \frac{1}{4}\left[ \left(1 + \gamma xy\right)\id + \left(c_1 x +c_2 y\right)\sigma_1
+ \left(c_2 x +c_1 y\right)\sigma_2\right],
\end{equation}
with $\gamma\in\Rbb$ and $\vec{c} = c_1\vec{i} + c_2\vec{j} \in \Rbb^2$ such that
\begin{equation}\label{2spineq}
\sqrt 2 \abs{c_1+c_2}-1\leq \gamma\leq 1-\sqrt 2\abs{c_1-c_2}.
\end{equation}
The marginals of $\Mo$ are $\Mo_{[1]} = \Ao_{\vec {c}}$ and $\Mo_{[2]} =
\Bo_{\vec {c}}$, with $\Ao_{\vec {c}}$, $\Bo_{\vec {c}}$ defined in \eqref{eq:marginstorte}.

\item \label{itemQ} Equation \eqref{eq:marginstorte} defines the marginals of a $D_2$-covariant bi-observable on $\Xscr\times\Yscr$ if and only if the vector $\vec{c}$
belongs to the square
\begin{equation}\label{squareQ}
Q=\{c_1\vec {i}+c_2\vec {j} : |c_1|\leq 1/\sqrt{2},\ |c_2|\leq
1/\sqrt{2}\}.
\end{equation}
\end{enumerate}
\end{proposition}

\begin{proof}
(i) The set $\Mscr(\Xscr\times\Yscr)$ is a subset of the linear space
$\Lscr(\Cbb^2)^{\Xscr\times\Yscr}=\Cbb^{\Xscr\times\Yscr}\otimes\Lscr(\Cbb^2)$,
where the set of the 16 products between one of the functions $1,\, x, \,y,
\,xy$ and one of the operators $\id, \sigma_1, \sigma_2, \sigma_3$ provides a
basis of linearly independent elements. Then, the most general
bi-observable on $\Xscr\times\Yscr$ is a linear combination of such products;
it is easy to see that the covariance under the rotation $R_{\vec{n}}(\pi) R_{\vec{m}}(\pi)$ implies the vanishing of the coefficients of the products
$x\id$, $y\id$, $xy\sigma_1$, $xy\sigma_2$, $1\sigma_1$, $1\sigma_2$,
$x\sigma_3$, $y\sigma_3$. By taking into account also the normalization and
selfadjointness conditions, we are left with
\begin{equation*}
\Mo(x,y)
= \frac{1}{4}\left[ \left(1 + \gamma xy \right)\id + \left(c_1 x +c_2 y\right)\sigma_1
+ \left(c'_1 x +c'_2 y\right)\sigma_2+ \left(c_3+c_4xy\right)\sigma_3\right],
\end{equation*}
with real coefficients $\gamma$, $c_i$ and $c'_i$. By imposing the covariance
under $R_{\vec  n}(\pi)$, we get $c'_1=c_2$, $c'_2=c_1$,
$c_3=c_4=0$, and \eqref{gencov} is obtained. Finally, since $R_{\vec{m}}(\pi) = R_{\vec{n}}(\pi) R_{\vec{m}}(\pi) R_{\vec{n}}(\pi)$, the bi-observable \eqref{gencov} is covariant
with respect to the whole group $D_2$. To impose the positivity of the operators $\Mo(x,y)$,
it is enough to study the diagonal elements and the determinant of the $2\times
2$--matrix representing \eqref{gencov}. The positivity of the diagonal elements
$\forall(x,y)$ gives $\gamma\in [-1,1]$. By the positivity of the determinant,
\[
\left(1+\gamma xy\right)^2\geq \left(c_1x+c_2y\right)^2+ \left(c_2x+c_1y\right)^2,
\qquad \forall(x,y)\in\Xscr\times\Yscr.
\]
The latter two conditions are equivalent to \eqref{2spineq}. Evaluating the marginals of
\eqref{gencov} immediately yields the observables \eqref{eq:marginstorte}.

(ii) We begin by noticing that $\vec c \in Q$ is equivalent to
\begin{equation}\label{bbineqbis}
\sqrt{2}\abs{c_1+c_2}-1 \leq 1-\sqrt{2}\abs{c_1-c_2} .
\end{equation}
For the marginals $\Ao_{\vec {c}}$ and $\Bo_{\vec {c}}$ of a
$D_2$-covariant bi-observable, inequalities \eqref{2spineq} trivially imply
\eqref{bbineqbis}, and so $\vec c \in Q$ holds; alternatively, the same result follows from \cite[Prop.\ 3]{BusH08}. Conversely, if $\Ao_{\vec {c}}$ and
$\Bo_{\vec {c}}$ are as in \eqref{eq:marginstorte} with $\vec c \in Q$, then by \eqref{bbineqbis} we can always find $\gamma$ as in
\eqref{2spineq}. The $D_2$-covariant bi-observable corresponding to
$\gamma,c_1,c_2$ then has marginals $\Ao_{\vec {c}}$ and $\Bo_{\vec {c}}$.
\end{proof}

Now we tackle the problem of evaluating the lower bound $\Icomp(\Ao,\Bo)$ and finding the
optimal covariant approximate joint measurements of the target spin-1/2 components \eqref{ABspin}. By Remark \ref{rem:covariant} and Proposition \ref{prop:cov},
\begin{equation}\label{eq:optispin}
\begin{aligned}
\Icomp(\Ao,\Bo) & = \min_{\substack{\Mo\in\Mscr(\Xscr\times\Yscr)\\ \Mo\ \text{$D_2$-covariant}}} \,
\max_{\substack{\rho\in\sh\\ \rho\text{ pure}}}\{\Srel{\Ao^\rho}{\Mo_{[1]}^\rho}
+ \Srel{\Bo^\rho}{\Mo_{[2]}^\rho}\}  \\
& = \min_{\vec {c}\in Q} \, \max_{\substack{\rho\in\sh\\ \rho\text{ pure}}}\{\Srel{\Ao^\rho}{\Ao^\rho_{\vec {c}}}
+ \Srel{\Bo^\rho}{\Bo^\rho_{\vec {c}}}\} ,
\end{aligned}
\end{equation}
where $Q$ is the square \eqref{squareQ}. Thus, the value of $\Icomp(\Ao,\Bo)$
can be found by minimizing the function
\begin{equation}\label{eq:defDR2}
D(\vec  c)=\max_{\substack{\rho\in\sh\\ \rho\text{ pure}}}\left\{\Srel{\Ao^\rho}{\Ao_{\vec {c}}^\rho}+ \Srel{\Bo^\rho}{\Bo_{\vec {c}}^\rho}\right\}
\end{equation}
for $\vec {c}$ ranging inside $Q$.

Note that the domain of the function $D$ can be extended to the whole disk $C$
introduced in \eqref{diskC}. In the domain $C$, $D(\vec c)=0$ if and only if
$\Ao_{\vec {c}}=\Ao$ and $\Bo_{\vec {c}}=\Bo$, which is equivalent to $\vec{c}
= \vec{a}$. The regions $C$ and $Q$ in the $\vec {i}\vec {j}$-plane are
depicted in Figure \ref{c+q}.
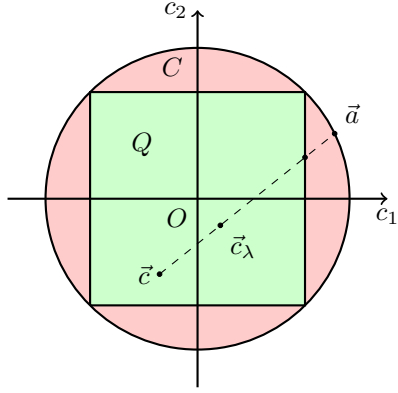
\begin{figure}[h]
\centering
\begin{tikzpicture}[domain=0:4]
\draw[thick,fill=red!20] (0,0) circle (2); \draw[thick,fill=green!20]
(-1.414,-1.414)--(1.414,-1.414)--(1.414,1.414)--(-1.414,1.414)--(-1.414,-1.414);
\draw[thick] (-1,1)node[anchor=north west]{$Q$}; \draw[thick]
(-0.6,2)node[anchor=north west]{$C$}; \draw[->,
thick](-2.5,0)--(2.5,0)node[anchor=north]{$c_1$};
\draw[->, thick](0,-2.5)--(0,2.5)node[anchor=east]{$c_2$}; \draw[thick]
(0,0)node[anchor=north east]{$O$};
\draw[very thick] (1.804,0.863)node[anchor=south west]{$\vec {a}$};
\draw[dashed] (-0.5,-1)--(1.804,0.863); \draw[fill=black] (-0.5,-1) circle (0.03);
\draw[thick] (-0.5,-1)node[anchor=east]{$\vec {c}$};
\draw[fill=black] (1.414,0.548) circle (0.03);
\draw[fill=black] (1.804,0.863) circle (0.03);
\draw[fill=black] (0.3,-0.353) circle (0.03);
\draw[thick] (0.3,-0.353)node[anchor=north west]{$\vec {c}_\lambda$};
\end{tikzpicture}
\caption{The  existence disk \eqref{diskC} for the observables $\Ao_{\vec{c}}$ and $\Bo_{\vec{c}}$, and their compatibility square \eqref{squareQ}. The disk is the domain of the function $D$ defined in \eqref{eq:defDR2}, and the square is the subset over which $D$ is minimized in \eqref{eq:minD}.}\label{c+q}
\end{figure}

We are now ready to prove our main result for the case of $\Ao$ and $\Bo$ being two arbitrary spin-1/2 components.
Indeed, the key point is that, by convexity arguments, the minimization of the function $D$ over the square $Q$ fixes $c_1 = 1/\sqrt2$. This considerably
simplifies the search of an optimal $D_2$-covariant bi-observable, as it
reduces the involved parameters from the number of three (see \eqref{gencov}) to a single one (see \eqref{eq:optispinobs}).

\begin{proof}[Proof of Theorem \ref{teo:incoNONORTHSPIN}]
By \eqref{eq:optispin}, we have
\begin{equation}\label{eq:minD}
\Icomp(\Ao,\Bo) = \min_{\vec{c}\in Q} D(\vec{c}) .
\end{equation}

Let us start with the case $\alpha\neq 0$. For $\vec{c}\in Q$, the observables $\Ao_{\vec{c}}$ and $\Bo_{\vec{c}}$ are compatible, and $D(\vec{c}) = \Div{\Ao,\Bo}{\Mo}$ for any of their joint measurements $\Mo$. By Theorem \ref{prop:2A}, item (\ref{Dfinite}), $D(\vec{c})$ is finite if and only if $\ker \Ao_{\vec{c}}(x) \subseteq \ker \Ao(x)$ and $\ker \Bo_{\vec{c}}(y) \subseteq \ker \Bo(y)$ for all $x,y$. In turn, this is equivalent to $\vec{c}$ not being any of the vertices $V$ of the square $Q$, since $\Ao_{\vec{c}} = \abs{\vec{c}}\Ao_{\vec{c}/\abs{\vec{c}}} + (1-\abs{\vec{c}}) \Uo_{\Xscr}$ and $\Bo_{\vec{c}} = \abs{\vec{c}}\Bo_{\vec{c}/\abs{\vec{c}}} + (1-\abs{\vec{c}}) \Uo_{\Yscr}$.
Therefore, in the minimum \eqref{eq:minD} we can assume that $\vec{c}\notin V$, and so $D(\vec{c})<+\infty$.

The mappings $\vec  c\mapsto \Ao_{\vec {c}}^\rho$ and $\vec  c\mapsto \Bo_{\vec
{c}}^\rho$ are affine on the disk $C$ for all $\rho\in \Sscr(\Hscr)$, which,
together with the convexity of the relative entropy, implies that the
mappings $\vec  c\mapsto \Srel{\Ao^\rho}{\Ao_{\vec {c}}^\rho}$ and $\vec
c\mapsto \Srel{\Bo^\rho}{\Bo_{\vec {c}}^\rho}$ are convex; hence, such are
their sum and the supremum $D$ in \eqref{eq:defDR2}. Moreover, we have already noticed that
$D(\vec c)=0$ if and only if $\vec {c} = \vec {a}$.

Making reference to Figure \ref{c+q}, let us take $\vec{c} \in Q\setminus V$ and introduce
the line segment joining $\vec  c$ and $\vec  a$: $\vec  c_\lambda=(1-\lambda)
\vec c + \lambda \vec  a$, $\lambda\in [0,1]$. By defining $D(\lambda)=D(\vec
c_\lambda)$, a simple convexity argument (see Lemma \ref{lem:nuovo} below) shows that the function $\lambda\mapsto D(\lambda)$ is finite and strictly decreasing on $[0,1]$. Then, the minimum of $D(\vec  c_\lambda)$ with respect to $\vec c_\lambda\in Q$
is attained where the line segment crosses the right side of the square, i.e.\ for
$(\vec c_\lambda)_1 = 1/\sqrt 2$. This is true for every point $\vec {c}$ in
the set $Q\setminus V$. Therefore, the points $\vec{c}$ minimizing \eqref{eq:minD} need to be on the right edge $\{1/\sqrt{2}\vec {i} + c_2\vec {j} : \abs{c_2}\leq 1/\sqrt{2}\}$  $=  \{\vec{c}(\gamma) : \gamma\in [-1,1]\}$ of the square $Q$; in the second equality, we have used the parametrization in \eqref{eq:rho(phi)c(gamma)}. In conclusion,
\begin{equation}\label{eq:optispin2bis}
\Icomp(\Ao,\Bo) = \min_{\gamma\in [-1,1]} D\left(\vec{c}(\gamma)\right).
\end{equation}
Note that \eqref{eq:optispin2bis} is true also in the case $\alpha = 0$ (compatible $\Ao$ and $\Bo$), for which we have $D(\vec{c}(1)) = 0$.

Now, for $\gamma\in [-1,1]$, define $\Mo_\gamma$ as in \eqref{eq:optispinobs}. Then, $\Mo_\gamma$ has the form \eqref{gencov} with $\vec{c} = \vec{c}(\gamma)$. In particular, since $\gamma$, $c_1 = 1/\sqrt{2}$ and $c_2 = \gamma/\sqrt{2}$ satisfy \eqref{2spineq}, item (\ref{itemM}) of Proposition \ref{prop:cov} implies that $\Mo_\gamma$ is a POVM, and $\Mo_{\gamma\,[1]} = \Ao_{\vec{c}(\gamma)}$ and $\Mo_{\gamma\,[2]} = \Bo_{\vec{c}(\gamma)}$. Equation \eqref{eq:dopo_optispin2} then follows from the definition \eqref{statedependentdivergence} of the error function. Moreover, by \eqref{eq:dopo_optispin2} and \eqref{eq:defDR2}, we have $\Div{\Ao,\Bo}{\Mo_\gamma} = D(\vec{c}(\gamma))$, hence $\Mo_\gamma\in\Mcomp(\Ao,\Bo)$ if and only if $\gamma$ attains the minimum in \eqref{eq:optispin2bis}.

In order complete the proof, it only remains to show that the minimization problem \eqref{eq:optispin2bis} is equivalent to \eqref{eq:optispin2}. Indeed, for $\rho = (\id + \vec{v}\cdot\vec{\sigma})/2$ and $\rho' = (\id + \vec{v}'\cdot\vec{\sigma})/2$, with $\vec{v} = (v_1\,,v_2\,,v_3)$ and $\vec{v}' = (v_1\,,v_2\,,0)$, we have $\sddiv{\Ao,\Bo}{\Mo_\gamma}(\rho) = \sddiv{\Ao,\Bo}{\Mo_\gamma}(\rho')$ by \eqref{eq:dopo_optispin2}. Therefore, by defining $\rho(\phi)$ as in \eqref{eq:rho(phi)c(gamma)},
\begin{align*}
D(\vec{c}(\gamma)) & = \max_{\substack{\rho\in\sh\\ \rho\text{ pure}}} \sddiv{\Ao,\Bo}{\Mo_\gamma}(\rho) & & \text{by \eqref{eq:dopo_optispin2}, \eqref{eq:defDR2}} \\
& = \max_{\phi\in[0,2\pi)} \sddiv{\Ao,\Bo}{\Mo_\gamma}(\rho(\phi)) .
\end{align*}
By inserting this expression into \eqref{eq:optispin2bis}, we get the desired equivalence.
\end{proof}

\begin{remark} \label{rem:covgamma}
The last proof shows that the bi-observables $\Mo_\gamma$ given by \eqref{eq:optispinobs} with $\gamma$ yielding the minimum in \eqref{eq:optispin2} actually exhaust all $D_2$-covariant elements in $\Mcomp(\Ao,\Bo)$. Indeed, for the most general $D_2$-covariant bi-observable $\Mo$ parameterized with $\gamma$ and $\vec{c}$ as in \eqref{gencov}, we have $\Div{\Ao,\Bo}{\Mo} = D(\vec{c}) > \Icomp(\Ao,\Bo)$ if $c_1\neq 1/\sqrt{2}$, or, equivalently, $\Mo\neq\Mo_\gamma$. However, it is not clear whether any optimal bi-observables needs to be $D_2$-covariant, and, if this is the case, the minimum \eqref{eq:optispin2} is attained at a unique $\gamma$.
\end{remark}

In the proof Theorem \ref{teo:incoNONORTHSPIN}, we have made use of the following lemma, which will turn out  useful also later.
\begin{lemma}\label{lem:nuovo}
Let $a\leq 0$, and suppose $D:[a,1]\to [0,+\infty]$ is a convex function such that $D(1) = 0$ and $D(0)<+\infty$. Then, $D$ is nonincreasing on the interval $[a,0]$, and it is finite and strictly decreasing on $[0,1]$.
\end{lemma}
\begin{proof}
For $a\leq x<y<1$, the convexity of $D$ implies
\begin{align}\label{eq:Dconv}
D(y) & \leq \frac{1-y}{1-x} D(x) + \frac{y-x}{1-x} D(1)
= \frac{1-y}{1-x} D(x) .
\end{align}
In particular, $D(y)\leq D(x)$, and, choosing $x=0$, $D(y)<+\infty$ for all $y\in (0,1)$. Then, another application of \eqref{eq:Dconv}, now with $0\leq x<y<1$, yields $D(y)<D(x)$. Since the latter inequality implies $D(x)>0$, for all $x\in [0,1)$, its extension to $y=1$ is clear.
\end{proof}

\subsection{The case of two orthogonal components} \label{app:orth}
When the target observables are the orthogonal spin-1/2 components $\Xo$ and $\Yo$ in \eqref{def:XY}, the symmetries of our system increase from $D_2$ to the enlarged dihedral group $D_4$. Here we recall that $D_4\subset SO(3)$ is the order $8$
group of the $90^\circ$ rotations around the $\vec  k$-axis, together with the
$180^\circ$ rotations around $\vec  i$, $\vec  j$, $\vec  n$ and $\vec m$;
clearly, $D_2\subset D_4$. Now, the two rotations $S_{D_4}=\{R_{\vec
i}(\pi),R_{\vec  n}(\pi)\}$ generate $D_4$; for instance, we have $R_{\vec
j}(\pi)=R_{\vec  n}(\pi)R_{\vec  i}(\pi)R_{\vec  n}(\pi)$, $R_{\vec
m}(\pi)=R_{\vec  i}(\pi)R_{\vec  n}(\pi)R_{\vec  i}(\pi)$, $R_{\vec
k}(\pi/2)=R_{\vec  m}(\pi)R_{\vec  j}(\pi)$.

The action of the group element $R_{\vec  n}(\pi)$ on $\Xscr\times \Yscr$, $\Hscr$,
$\Ao=\Xo$ and $\Bo=\Yo$ is still given by \eqref{2rot} and \eqref{Un};
we have already seen that these actions satisfy condition (ii) of
Theorem \ref{prop:invarD^}. Further, by introducing the natural actions
\begin{equation}\label{trasf+}
R_{\vec  i}(\pi)\, (x,y) = (x,-y), \qquad U\big(R_{\vec
i}(\pi)\big)=\rme^{-\rmi\pi \,{\vec  i}\cdot \vec  \sigma/2}\equiv
-\rmi\, {\vec  i}\cdot \vec  \sigma,
\end{equation}
we have
\begin{equation*}
U\big(R_{\vec  i}(\pi)\big) \Xo(x) U\big(R_{\vec  i}(\pi)\big)^*
= \Xo(x),\qquad  U\big(R_{\vec  i}(\pi)\big) \Yo(y) U\big(R_{\vec  i}(\pi)\big)^*
= \Yo(-y).
\end{equation*}
In particular, we see that $R_{\vec  i}(\pi)$ fulfills condition (i) of the
same theorem. Therefore, all $g\in S_{D_4}$ satisfy the hypotheses of Theorem
\ref{prop:invarD^}.

Again, in view of Remark \ref{rem:covariant}, now we look for the general expression of
a $D_4$-covariant bi-observable.

\begin{proposition}\label{cor:cov}
Let the dihedral group $D_4$ act on $\Xscr\times\Yscr$ and $\Hscr$ by
\eqref{2rot} and \eqref{trasf+}. Then, the most general $D_4$-covariant
bi-observable on $\Xscr\times\Yscr$ is given by \eqref{gencov} with $\gamma=0$,
$c_2=0$ and $\abs{c_1}\leq 1/\sqrt 2$, that is,
\begin{equation}\label{gencovD4}
\Mo(x,y)
= \frac{1}{4}\left[\id + c_1\left(x\sigma_1 + y\sigma_2\right) \right] , \qquad \abs{c_1}\leq 1/\sqrt 2 .
\end{equation}
\end{proposition}
\begin{proof}
By applying the extra transformation \eqref{trasf+} to the $D_2$-covariant
bi-observable \eqref{gencov} we get
\begin{align*}
R_{\vec  i}(\pi)\Mo(x,y) & = U\big(R_{\vec  i}(\pi)\big)\Mo(x,-y)U\big(R_{\vec  i}(\pi)\big)^*
\\
& = \frac{1}{4}\left[ \left(1 - \gamma xy\right)\id + \left(c_1 x -c_2 y\right)\sigma_1
- \left(c_2 x -c_1 y\right)\sigma_2\right].
\end{align*}
In order to have covariance also under this transformation, it must be $\gamma
=0$ and $c_2=0$; then,  condition \eqref{2spineq} reduces to the  inequality  in
\eqref{gencovD4}. \end{proof}

We are now ready to prove our main theorem for two orthogonal spin components.

\begin{proof}[Proof of Theorem \ref{teo:incoORTHSPIN}]
By Theorem \ref{prop:invarD^}, there is at least one $D_4$-covariant bi-observable $\Mo\in\Mcomp(\Xo,\Yo)$, which is necessarily of the form \eqref{gencovD4} by Proposition \ref{cor:cov}. Comparing it with \eqref{eq:optispinobs}, we see that they coincide if and only if $c_1 = 1/\sqrt{2}$ and $\gamma = 0$, and in this case both of them equal $\Mo_0$ in \eqref{orthcov}. Thus, by Theorem \ref{teo:incoNONORTHSPIN}, $\gamma = 0$ solves the minimization problem \eqref{eq:optispin2}, and $\Mo_0$ is the unique $D_4$-covariant element in $\Mcomp(\Xo,\Yo)$. In particular, by \eqref{eq:optispin2} and \eqref{eq:dopo_optispin2} we have
\begin{equation}\label{eq:cinc_con_sv}
\begin{gathered}
\Icomp(\Xo,\Yo) = \max_{\phi\in [0,2\pi)} \sddiv{\Ao,\Bo}{\Mo_0}(\rho(\phi)) \\
\sddiv{\Ao,\Bo}{\Mo_0}(\rho(\phi)) = \tilde s(\cos \phi)+\tilde s(\sin \phi),
\end{gathered}
\end{equation}
where we have introduced the function
$$
\tilde{s}(v)=\frac 12 \sum_{k=\pm 1}(1+kv)\log \frac{1+kv}{1+kv/\sqrt 2}, \qquad \abs{v}\leq 1 .
$$
In \eqref{eq:cinc_con_sv}, the best way to maximize $\tilde s(\cos \phi)+\tilde s(\sin \phi)$ is by means of a suitable integral
representation. Namely, by direct inspection, we have
\begin{equation*}
\tilde s(v)= \frac 1{2\ln 2}\int_{\frac{1}{\sqrt 2}}^1 \frac {2v^2(1-\lambda)}{1-\lambda^2v^2}\,\rmd \lambda.
\end{equation*}
Then, by differentiation and simple computations, we get
\[
f(\phi)=\frac{\rmd \ }{\rmd \phi}\left(\tilde s(\cos \phi)+ \tilde s(\sin \phi)\right)
= -\frac{\sin(4\phi)}{2\ln 2}\int_{1/\sqrt 2}^1 \frac {\lambda^2(1-\lambda)(2-\lambda^2)}{\left(1-\lambda^2(\sin\phi)^2\right)^2
\left(1-\lambda^2(\cos\phi)^2\right)^2}\,\rmd \lambda.
\]
The integrand is nonnegative for all $\lam\in[1/\sqrt{2},1]$ and $\phi\in [0,2\pi)$. We then see that $f(\phi)<0$ for $0<\phi <\pi/4$, $f(\pi/4)=0$,
$f(\phi)>0$ for $\pi/4<\phi<\pi/2$. So, for $\phi\in [0,\pi/2]$, the point $\phi=\pi/4$ gives a minimum
of $\tilde s(\cos \phi)+ \tilde s(\sin \phi)$, while we have two equal maxima at $\phi=0$ and
$\phi=\pi/2$; as $\tilde{s}$ is a continuous even function on $[-1,1]$, the maximum \eqref{eq:cinc_con_sv} is attained at $\phi=0,\pi/2,\pi,3\pi/2$. Such angles correspond to $\rho(\phi)$ being the eigenprojections of $\sigma_1$ or $\sigma_2$; this gives the first equality in \eqref{corth}. Then, in the last two ones, the numerical values follow by direct computation.

Finally, we still have to prove the uniqueness of $\Mo_0$ in the set $\Mcomp(\Xo,\Yo)$. Let $\Mo$ be any bi-observable in $\Mcomp(\Xo,\Yo)$. By Corollary
\ref{prop:invarI^}, its covariant version $\Mo_{D_4}$ is still in
$\Mcomp(\Xo,\Yo)$, and hence $\Mo_{D_4} = \Mo_0$ since $\Mo_0$ is the unique $D_4$-covariant element of $\Mcomp(\Xo,\Yo)$.
Definition \eqref{eq:covariantization} implies $g\Mo(x,y)\leq
|D_4|\Mo_{D_4}(x,y)$ for all $g$ and $x,y$, hence in particular $\Mo(x,y)\leq
|D_4|\Mo_{D_4}(x,y) = 8 \Mo_0(x,y)$ for all $x,y$. Since $\Mo_0(x,y)$ has rank
$1$, it must then be $\Mo(x,y) = f(x,y) \Mo_0(x,y)$, $\forall x,y$,
for some nonnegative coefficients $f(x,y)$. Writing $f$ in the linear basis
$1$, $x$, $y$, $xy$ of $\Cbb^{\Xscr\times\Yscr}$, the normalization constraint
$\sum_{x,y} \Mo(x,y) =\sum_{x,y}f(x,y)\Mo_0(x,y) =  \id$ gives
$f(x,y)=1+\epsilon xy$ for some real parameter $\epsilon$. For all $x,y$, we have the positivity constraint $\Mo(x,y) =f(x,y)\Mo_0(x,y)\geq0$, which implies $f(x,y)\geq 0$; this
gives $-1\leq\epsilon\leq1$.

Summing up, if $\Mo\in\Mcomp(\Xo,\Yo)$, then $\Mo(x,y) =(1+\epsilon
xy)\Mo_0(x,y)$ for some $\epsilon\in[-1,1]$. Let us show that the only possible
parameter is $\epsilon=0$. Indeed, the marginals of $\Mo$ are
$$
\Mo_{[1]} = \Ao_{\vec{c}(\epsilon)}, \qquad \Mo_{[2]} = \Bo_{\vec{c}(\epsilon)}, \qquad \text{with} \qquad \vec{c}(\epsilon) = \frac{\vec{i} + \epsilon \vec{j}}{\sqrt{2}}.
$$
Their distributions in the state $\rho_{e}=\left(\id + \sigma_1\right)/2$ are
$$
\Mo_{[1]}^{\rho_{e}} = \frac{1}{\sqrt{2}} \, \delta_1 + \left(1 - \frac{1}{\sqrt{2}}\right) u_{\Xscr} ,
\qquad \Mo_{[2]}^{\rho_{e}} = \frac{\epsilon}{\sqrt{2}} \, \delta_1 + \left(1 - \frac{\epsilon}{\sqrt{2}}\right) u_{\Yscr} .
$$
On the other hand, we have $\Xo^{\rho_{e}} = \delta_1$ and $\Yo^{\rho_{e}} = u_{\Yscr}$,
so that
\begin{multline*}
\Icomp(\Xo,\Yo)=\Div{\Xo,\Yo}{\Mo}  \geq \sddiv{\Xo,\Yo}{\Mo}(\rho_{e})=\Srel{\Xo^{\rho_{e}}}
{\Mo_{[1]}^{\rho_{e}}}+ \Srel{\Yo^{\rho_{e}}}{\Mo_{[2]}^{\rho_{e}}}\\
{}=\log\frac{2\sqrt 2}{1 +\sqrt2}+\Srel{\Yo^{\rho_{e}}}{\Mo_{[2]}^{\rho_{e}}}
=\Icomp(\Xo,\Yo)+ \Srel{\Yo^{\rho_{e}}}{\Mo_{[2]}^{\rho_{e}}},
\end{multline*}
which implies $\Srel{\Yo^{\rho_{e}}}{\Mo_{[2]}^{\rho_{e}}}=0$. Hence,
$\Yo^{\rho_{e}}=\Mo_{[2]}^{\rho_{e}}$, and $\epsilon=0$ then follows.
\end{proof}

\subsection{A lower bound for the incompatibility degree}\label{sec:lb}
In order to compute the lower bound \eqref{lbspin}, we have to minimize the following quantity over $\gamma$:
\begin{multline}\label{s:c2}
\sddiv{\Ao,\Bo}{\Mo_\gamma}(\rho_e) = \log \frac {2} {1+ (a_1 + a_2\gamma)/\sqrt 2 } + \frac{1+ 2a_1a_2}{2}\,\log \frac {1+ 2a_1a_2}{1+ (a_1\gamma + a_2)/\sqrt 2 } \\
{} + \frac{1-2a_1a_2}{2}\,\log \frac {1- 2a_1a_2}{1- (a_1\gamma + a_2)/\sqrt
2 }\,.
\end{multline}
By setting $ \ell=(a_1\gamma + a_2)/\sqrt 2$ and $f(\ell)= (\ln 2)
\sddiv{\Ao,\Bo}{\Mo_\gamma}(\rho_e)$, we get
\begin{multline}\label{f(x)}
f(\ell)= \ln \frac {2\sqrt 2\, a_1} {\sqrt2 \, a_1+\sqrt 2 \, a_2 \ell +a_1^{\,2}-a_2^{\,2}}\\
{}+ \frac12\left(1+2a_1a_2\right)\ln \frac {1+2a_1a_2} {1+\ell} +
\frac12\left(1- 2a_1a_2\right)\ln \frac {1-2a_1a_2}{1-\ell}\,,
\end{multline}
whose derivative is
\[
f'(\ell)=-\frac{\sqrt 2\, a_2}{\sqrt2 \, a_1+\sqrt 2 \, a_2 \ell
+a_1^{\,2}-a_2^{\,2}} +\frac{\ell-2a_1a_2} {1-\ell^2}\,.
\]

\begin{remark}\label{rem:app}
For $\alpha=\pi/2$, i.e.\ $a_1=1$ and $a_2=0$, we immediately get that the expression \eqref{f(x)} has a unique minimum at $\ell=0$, which gives  $\gamma=0$ and the value \eqref{corth} for the incompatibility degree.
\end{remark}

For $\alpha\neq \pi/2$, the zeros of $f'(\ell)$ satisfy the algebraic
equation $\ell^2+u\ell/(\sqrt2\,a_2)- 1-u=0$,
where $u$ is defined in \eqref{def:u}. By solving the algebraic equation and
studying the sign of the derivative, we find that the minimum of \eqref{f(x)}
is at the point \eqref{ellmin} and that the corresponding value of $\gamma$ is
\eqref{c2min}. By using this result and $2a_1a_2=\cos\alpha$, we get the lower
bound \eqref{s:c2min}.

\subsection{Incompatibility degree and optimal measurements for three orthogonal spin-1/2 components}\label{App:B3}

Here the target observables are $\Xo$, $\Yo$ and $\Zo$ defined in \eqref{def:XY} and \eqref{eq:Z}. Their symmetry group is the order $24$ octahedron group $O\subset SO(3)$, generated by the
$90^\circ$ rotations around the three coordinate axes: $S_{O}=\{R_{\vec
i}(\pi/2),\,R_{\vec  j}(\pi/2),\,R_{\vec  k}(\pi/2)\}$. Note that for the
dihedral groups introduced before we have $D_2\subset D_4
\subset O$. Let us denote the three generators of $O$ by $g_1=R_{\vec
i}(\pi/2)$, $g_2=R_{\vec j}(\pi/2)$, $g_3=R_{\vec  k}(\pi/2)$. By using again
the spin-1/2 projective representation of $SO(3)$, which we now restrict to
$O$, we have the relations
\begin{equation*}\begin{split}
U_{g_1}
\Xo(x) U_{g_1}^{\;*} = \Xo(x), \qquad
&U_{g_1}\Yo(y) U_{g_1}^{\;*}  = \Zo(y), \qquad \ \
U_{g_1} \Zo(z) U_{g_1}^{\;*}= \Yo(-z),
\\
U_{g_2}\Xo(x)U_{g_2}^{\;*}  = \Zo(-x), \qquad
&U_{g_2} \Yo(y) U_{g_2}^{\;*} = \Yo(y), \qquad \ \ U_{g_2}\Zo(z)U_{g_2}^{\;*}  = \Xo(z),
\\
U_{g_3} \Xo(x) U_{g_3}^{\;*}= \Yo(x), \qquad  &U_{g_3}  \Yo(y) U_{g_3}^{\;*} = \Xo(-y),
\qquad  U_{g_3} \Zo(z)U_{g_3}^{\;*} = \Zo(z).
\end{split}
\end{equation*}
Moreover, the natural action of $O$ on the outcome space $\Xscr\times
\Yscr\times \Zscr = \{+1,-1\}^3$ is
\[
g_1\,(x,y,z)=(x,-z,y),\qquad g_2\,(x,y,z)=(z,y,-x),
\qquad g_3\,(x,y,z)=(-y,x,z),
\]
and the action on the index set is
\[
g_ii=i, \quad g_12=3, \quad g_13=2,\quad g_21=3, \quad g_23=1, \quad g_31=2,\quad g_32=1.
\]
Then, the hypotheses of Theorem \ref{prop:invarD+I^n} are satisfied by setting
\[
f_{g_1,1}(x)=x, \quad f_{g_1,2}(y)=y, \quad f_{g_1,3}(z)=-z,  \quad f_{g_2,1}(x)=-x,
\]
\[ f_{g_2,2}(y)=y, \quad f_{g_2,3}(z)=z, \quad f_{g_3,1}(x)=x, \quad
f_{g_3,2}(y)=-y, \quad  f_{g_3,3}(z)=z.
\]
Therefore, we can apply Theorem \ref{prop:invarD+I^n} in order to prove the main result of Section \ref{ex:3spin}.
\begin{proof}[Proof of Theorem \ref{teo:incoORTH333SPIN}]
By similar arguments as in the proofs of Propositions \ref{prop:cov} and \ref{cor:cov}, one can prove that the most general $O$-covariant
tri-observable in $\Mscr(\Xscr\times\Yscr\times\Zscr)$ has the form
\begin{equation}\label{3cov}
\Mo(x,y,z) = \frac{1}{8} \left[\id+c(x\sigma_1+y\sigma_2+z\sigma_3)\right] \qquad\text{with}
\qquad \abs c \leq \frac 1 {\sqrt 3}\,.
\end{equation}
Writing its marginals as
\begin{equation*}
\Mo_{[1]} = \Xo_c , \qquad \Mo_{[2]} = \Yo_c , \qquad \Mo_{[3]} = \Zo_c ,
\end{equation*}
we have
$$
\Div{\Xo,\Yo,\Zo}{\Mo} = \max_{\substack{\rho\in\Sscr(\Hscr)\\
    \rho\ \mathrm{pure}}} \left[\Srel{\Xo^\rho}{\Xo_c^\rho} + \Srel{\Yo^\rho}{\Yo_c^\rho} + \Srel{\Zo^\rho}{\Zo_c^\rho}\right]
$$
for all $c\in[-1/\sqrt{3},1/\sqrt{3}]$. Denote by $D(c)$ the right hand side of the previous equation; then, the function $D$ can be extended to all $c$'s such that $\Xo_c$, $\Yo_c$ and $\Zo_c$ define three POVMs on $\{-1,+1\}$. In particular, it is naturally defined also in the interval $(1/\sqrt{3},1]$, where $\Xo_c$, $\Yo_c$ and $\Zo_c$ are the equally noisy versions of the sharp observables $\Xo$, $\Yo$ and $\Zo$ (cf.\ \eqref{eq:defOlam}). We thus obtain a function $D:[-1/\sqrt{3},1]\to [0,+\infty]$. The mappings $c\mapsto\Xo^\rho_c$, $c\mapsto\Yo^\rho_c$ and $c\mapsto\Zo^\rho_c$ are affine on the interval $[-1/\sqrt{3},1]$, which, together with the convexity of the relative entropy, implies that such are the sum and the supremum in $D$.
Moreover, $D(0) = \Div{\Xo,\Yo,\Zo}{\Uo_{\Xscr\times\Yscr\times\Zscr}} < +\infty$ and $D(1) = 0$. Then, by Lemma \ref{lem:nuovo}, the divergence $\Div{\Xo,\Yo,\Zo}{\Mo}$, with $\Mo$ given by \eqref{3cov}, attains its unique minimum when $c = 1/\sqrt{3}$; for such $c$, $\Mo = \Mo_0$ defined in \eqref{eq:M0x3}. Since $\Mcomp(\Xo,\Yo,\Zo)$ contains at least one $O$-covariant tri-observable by Theorem \ref{prop:invarD+I^n}, then $\Mo_0$ is the unique $O$-covariant element in $\Mcomp(\Xo,\Yo,\Zo)$. The fact that also $\Mo_1$ given by \eqref{eq:tildeM0x3} is optimal follows since $\Mo_0$ and $\Mo_1$ have the same marginals (see \eqref{eq:noisy3margins}).

For the optimal approximate joint measurements $\Mo_0$ and $\Mo_1$, we have
\begin{multline}\label{eq:max333}
\Icomp(\Xo,\Yo,\Zo)  = \Div{\Xo,\Yo,\Zo}{\Mo_i} = \max_{\substack{\rho\in\Sscr(\Hscr)\\ \rho\ \mathrm{pure}}} \sddiv{\Xo,\Yo,\Zo}{\Mo_i}(\rho) \\
{} = \max_{\substack{\rho\in\Sscr(\Hscr)\\ \rho\ \mathrm{pure}}} \left[ \Srel{\Xo^{\rho}}{\Xo^{\rho}_{1/\sqrt{3}}} + \Srel{\Yo^{\rho}}{\Yo^{\rho}_{1/\sqrt{3}}} + \Srel{\Zo^{\rho}}{\Zo^{\rho}_{1/\sqrt{3}}} \right] \\
{}= \max_{\substack{\phi\in [0,2\pi)\\ \theta\in [0,\pi)}} \left[ \tilde{s}(\cos\phi\sin\theta)+\tilde{s}(\sin\phi\sin\theta)+\tilde{s}(\cos\theta) \right] ,
\end{multline}
where we have used the parametrization $\rho = (\id + \cos\phi\sin\theta\,\sigma_1 + \sin\phi\sin\theta\,\sigma_2 + \cos\theta\,\sigma_3)/2$, inserted the marginals \eqref{eq:noisy3margins} of $\Mo_0$, and introduced the function
\begin{equation*}
\tilde{s}(v)=\frac 12 \sum_{k=\pm 1}(1+kv)\log \frac{1+kv}{1+kv/\sqrt 3} = \frac 1{2\ln 2}\int_{\frac{1}{\sqrt 3}}^1 \frac {2v^2(1-\lambda)}{1-\lambda^2v^2}\,\rmd \lambda, \qquad \abs{v}\leq 1 .
\end{equation*}
By using the integral representation of $\tilde{s}$,
\begin{multline*}
\frac{\partial \ }{\partial \phi}\bigl(\tilde{s} (\cos \phi\,\sin \theta) + \tilde{s} (\sin
\phi\,\sin \theta) + \tilde{s} (\cos \theta)\bigr) \\ {}= - \frac{\sin(4\phi)(\sin
\theta)^4}{2\ln 2} \int_{1/\sqrt 3}^1 \frac
{\lambda^2(1-\lambda)(2-\lambda^2)}{\left(1-\lambda^2v_1^{\,2}\right)^2
\left(1-\lambda^2v_2^{\,2}\right)^2}\,\rmd \lambda;
\end{multline*}
similar computations give the derivative with respect to $\theta$. By the same
arguments as in the case of two components, we obtain that in \eqref{eq:max333} the maximum is attained at all angles $\phi,\theta$ corresponding to $\rho$ being an eigenprojection of $\sigma_1$, $\sigma_2$ or $\sigma_3$. This fact and a final straightforward computation  give \eqref{corth333}.
\end{proof}

\begin{remark}\label{rem:B4}
The last proof actually shows that $\Mo_0$ given in \eqref{eq:M0x3} is the unique $O$-covari\-ant optimal approximate joint measurement of $\Xo$, $\Yo$ and $\Zo$.
\end{remark}

\section{Symmetries and proofs for two Fourier conjugate MUBs} \label{app:MUB}

The natural symmetry group for the two Fourier conjugate observable $\Qo$ and $\Po$ of \eqref{eq:defcoj} is the
group of the translations in the finite phase-space of the system, together with all its symplectic transformations; as usual, we identify the latter symplectic
group with the group $SL(2,\F)$ of the $2\times 2$ matrices with entries in
$\F$ and unit determinant. However, just a smaller subgroup of $SL(2,\F)$ will
be enough for us. Namely, for all $a\in\F_* =
\F\setminus\{0\}$, we denote by $\ma{d}(a)$ and
$\ma{f}(a)$ the $SL(2,\F)$-matrices
$$
\ma{d}(a) = \begin{pmatrix} a & 0 \\ 0 & a^{-1} \end{pmatrix},
\qquad \ma{f}(a) = \begin{pmatrix} 0 & a \\ -a^{-1} & 0 \end{pmatrix} .
$$
Then, the set $H=\{\ma{d}(a),\ma{f}(a)\mid a\in\F_*\}$ is an order $2(d-1)$
subgroup of the order $d(d^2-1)$ group $SL(2,\F)$. It naturally acts by left
multiplication on the additive abelian  group $V=\F^2$ of the $\F$-valued $2$-entries
column vectors $\vec u = (u_1,u_2)^T$. We can then form the semidirect product
group $G = H\rtimes V$, whose composition law is $(\ma{h},\vec u) (\ma{k},\vec
v) = (\ma{h}\ma{k},\ma{k}^{-1}\vec u+\vec v)$.

The group $G$ has a natural left action on the joint outcome space
$\Xscr\times\Yscr = \F^2$: by writing the points of $\Xscr\times\Yscr = \F^2$ as columns, we have
\begin{equation}\label{eq:actionG}
(\ma{h},\vec u)\begin{pmatrix}x\\y\end{pmatrix} = \ma{h} \begin{pmatrix}x+u_1\\y+u_2\end{pmatrix}.
\end{equation}
In
this context, the joint outcome space $\Xscr\times\Yscr$ is called the {\em
finite phase-space} of the system, and the subgroup $V\subset G$ is the group
of its translations $(\ma{d}(1),\vec u)$. The elements $(\ma{d}(a),\vec 0)\in
H$ are diagonal symplectic transformations, while $(\ma{f}(1),\vec 0)$ just
reverses the components $x$ and $y$ changing the sign of $x$ (see
e.g.~\cite{CaScTo15} for more details on finite phase-spaces and their
symmetries).

On the other hand, the group $G$ has also a natural projective unitary
representation on $\Hscr$. In order to describe it, we first introduce the
following unitary operators:
\begin{align*}
W(\vec u)\phi(z) & = \rme^{\frac{2\pi \rmi}{p}\,\tr{u_2(z-u_1)}} \phi(z-u_1), \qquad \forall \vec u\in\F^2, \\
D(a)\phi(z) & = \phi(a^{-1}z), \qquad \forall a\in\F_* = \F\setminus\{0\} .
\end{align*}
The operators $W(\vec u)$ constitute the {\em Weyl operators} associated with
the phase-space translations, and $D(a)$ are the {\em squeezing operators} by
the nonzero scalars. Collected together with the Fourier transform $F$, they
satisfy the composition rules
\begin{align*}
& W(\vec u)W(\vec v) = \rme^{\frac{2\pi\rmi}{p}\,\tr{u_2 v_1}} W(\vec u+\vec v), & \qquad & D(a)D(b) = D(ab), \\
& F^2 = D_{-1}, & \qquad & FD(a)F^* = D(a^{-1}), \\
& D(a)W(\vec u)D(a)^* = W(\ma{d}(a)\vec u), & \qquad & FW(\vec u)F^* =
\rme^{-\frac{2\pi \rmi}{p}\,\tr{u_1 u_2}} W(\ma{f}(1)\vec u).
\end{align*}
Setting
$$
U(\ma{d}(a),\vec u) = D(a)W(\vec u), \qquad \qquad U(\ma{f}(a),\vec u) = D(a)FW(\vec u) ,
$$
we obtain a projective unitary representation of $G$ on $\Hscr$. It is easily
checked that
\begin{equation}\label{eq:covarQP}
\begin{aligned}
& U(\ma{d}(a),\vec u )\Qo(x)U(\ma{d}(a),\vec u )^* = \Qo(a(x+u_1)), \\
& U(\ma{d}(a),\vec u )\Po(y)U(\ma{d}(a),\vec u )^* = \Po(a^{-1}(y+u_2)), \\
& U(\ma{f}(a),\vec u )\Qo(x)U(\ma{f}(a),\vec u )^* = \Po(-a^{-1}(x+u_1)), \\
& U(\ma{f}(a),\vec u )\Po(y)U(\ma{f}(a),\vec u )^* = \Qo(a(y+u_2)) .
\end{aligned}
\end{equation}

The action \eqref{eq:actionG} satisfies conditions (\ref{8i}.a) /
(\ref{8ii}.a) of Theorem \ref{prop:invarD^}, with $S_G= G$. Moreover, by
\eqref{eq:covarQP} the two sharp observables $\Ao=\Qo$ and $\Bo=\Po$
satisfy conditions (\ref{8i}.b) / (\ref{8ii}.b) of the same theorem.
Therefore, by Corollary \ref{prop:invarI^} we conclude that the set
$\Mcomp(\Qo,\Po)$ contains a $G$-covariant element $\Mo_0$.

Since in particular the bi-observable $\Mo_0$ is covariant with respect to the
group $V$ of the phase-space translations, it must be of the form
\begin{equation}\label{Mtau}
\Mo_{\tau} (x,y) = \frac{1}{d}\, W((x,y)^T) \tau W((x,y)^T)^* ,\qquad \forall x,y\in\F,
\end{equation}
i.e.\ $\Mo_0 =\Mo_{\tau_0} $ for some state $\tau_0\in\sh$  \cite[Theor.\ 4.5.3]{DavQTOS}.
According to \cite{BGL97,BLPY16}, we call an observable $\Mo_\tau$ of the form \eqref{Mtau} the
{\em $V$-covariant phase-space observable generated by the state $\tau$}. Since
$\Mo_0$ is also $H$-covariant and $H$ is the stability subgroup of $G$ at
$(0,0)$, we see that $\tau_0=d\,\Mo_0(0,0)$ can be any state commuting with the
restriction $\left.U\right|_H$ of the representation $U$ to $H$.

By \cite[Props.\ 1 and 2]{CHT11}, the marginals of a $V$-covariant
phase-space observable $\Mo_\tau$ are
\begin{equation}\label{eq:convol}
\Mo_{\tau\,[1]}(x) = \sum_{z\in\F} \Qo^\tau(z-x)\Qo(z),  \qquad
\Mo_{\tau\,[2]}(y) = \sum_{z\in\F} \Po^\tau(z-y)\Po(z) .
\end{equation}
Now, the fact that $\tau_0$ commutes with $\left.U\right|_H$ and the covariance
relations \eqref{eq:covarQP} imply
$$
\Qo^{\tau_0}(x) = \Tr\left[\tau_0 U(\ma{f}(-1),\vec 0 )\Qo(x)U(\ma{f}(-1),\vec 0 )^*\right]
= \Po^{\tau_0}(x), \qquad \forall x\in\F,
$$
$$
\Qo^{\tau_0}(x) = \Tr\left[\tau_0 U(\ma{d}(a),\vec 0 )\Qo(x)U(\ma{d}(a),\vec 0 )^*\right]
= \Qo^{\tau_0}(ax), \qquad \forall x\in\F,\,a\in\F_* .
$$
By the second relation, the probability $\Qo^{\tau_0}$ is constant on the two
subsets $\{0\}$ and  $\F_*$ of $\F$, which are the orbits of the action of the
multiplicative group $\F_*$ on $\F$. Therefore, we can write $\Qo^{\tau_0}$ as
a linear combination of the two functions $\delta_0$ and $u_{\F}-\delta_0/d$.
The normalization of $\Qo^{\tau_0}$ requires
$$
\Qo^{\tau_0} = \lambda_0\delta_0 + (1-\lambda_0) u_{\F}
$$
for some real $\lambda_0$. On the other hand, we must have
$\lambda_0\in[-1/(d-1)\,,\,1]$ by the positivity constraint. Equations
\eqref{eq:convol} with $\tau=\tau_0$ then give
\begin{equation*}
\Mo_{0\,[1]} = \lambda_0\Qo + (1-\lambda_0) \Uo_\F =: \Qo_{\lambda_0}\,, \qquad
\qquad \Mo_{0\,[2]} = \lambda_0\Po + (1-\lambda_0) \Uo_\F =: \Po_{\lambda_0},
\end{equation*}
where $\Uo_\F$ is the trivial uniform observable on $\F$. If $\lambda_0\geq 0$, then
$\Qo_{\lambda_0}$ and $\Po_{\lambda_0}$ have the simple physical interpretation as
uniformly noisy versions of $\Qo$ and $\Po$ with noise intensities $1-\lambda_0$,
as it was explained in Section \ref{sec:noise} (cf.\ \eqref{eq:defOlam}).
However, we can not exclude that $\lambda_0$ takes its value in the negative
interval $[-1/(d-1)\,,\,0)$, where this interpretation does not apply.

We finally come to the proof of our main result for two Fourier conjugate target observables.

\begin{proof}[Proof of Theorem \ref{teo:incoMUB}]
For $\lambda\in [0,1]$, a straightforward extension of the argument in \cite[Prop.\ 5]{CHT12} from the cyclic field $\Zbb_p$ to the finite field $\F$ yields that the minimal noise intensity making the two noisy observables $\Qo_\lambda$ and $\Po_\lambda$ compatible is
\begin{equation}\label{eq:comp_lam}
1-\lambda \geq 1-\lambda_* = \frac{\sqrt{d}}{2(\sqrt{d}+1)}
\end{equation}
(see also Example 1 therein). Moreover, the same extension also proves that when in the previous bound the equality is attained, $\Qo_{\lambda_*}$ and $\Po_{\lambda_*}$ have a unique joint measurement in the whole set $\Mscr(\Xscr\times\Yscr)$; it is the $V$-covariant phase-space observable $\Mo_{\tau_*}$ generated by the pure state
\begin{equation*}
\tau_* = \frac{\sqrt{d}}{2(1+\sqrt{d})} \kb{\psi_{0,0}}{\psi_{0,0}} ,
\end{equation*}
with $\psi_{0,0}$ given in \eqref{eq:hatI(Q,P)2}.
As a consequence, for the two marginals $\Qo_{\lam_0}$ and $\Po_{\lam_0}$ of
the optimal approximate joint measurement $\Mo_0$, the inequalities $-1/(d-1)\leq\lam_0\leq\lambda_*$ must hold. Note that the state $\tau_*$ commutes with $\left.U\right|_H$, hence it is a valid
candidate for generating the $G$-covariant phase-space observable $\Mo_0$.

Now, by optimality of $\Mo_0$ we have
$$
\Icomp(\Qo,\Po)=\Div{\Qo,\Po}{\Mo_0} = \sup_{\rho}\left[\Srel{\Qo^\rho}{\Qo^\rho_{\lam_0}}
+ \Srel{\Po^\rho}{\Po^\rho_{\lam_0}}\right] =: D(\lam_0) .
$$
The map $\lambda\mapsto D(\lambda) =
\sup_{\rho}\left[\Srel{\Qo^\rho}{\Qo^\rho_\lambda} +
\Srel{\Po^\rho}{\Po^\rho_\lambda}\right]$ is defined for all $\lambda\in\Rbb$ such that $\Qo_\lambda$ and $\Po_\lambda$ are two POVMs.
By affinity, these $\lambda$'s form an interval $I$, which necessarily contains the subinterval $[0,1]$. On the interval $I$, the function $D$ is nonnegative; moreover, the mappings $\lam\mapsto\Qo^\rho_\lam$ and $\lam\mapsto\Po^\rho_\lam$ are affine on $I$, which, together with the convexity of the relative entropy, implies that such are the sum and the supremum in $D$.
Since $D(0) = \Div{\Qo,\Po}{\Uo_{\Xscr\times\Yscr}} < +\infty$ and $D(1) = 0$, by Lemma \ref{lem:nuovo} the function $D$ is nonincreasing on $I$, and finite and strictly decreasing on $[0,1]$. This fact and inequality \eqref{eq:comp_lam} for compatible $\Qo_\lam$ and $\Po_\lam$ then imply $\lam_0 = \lam_*$.
Moreover, the fact that $\Mo_{\tau_*}$ is the unique joint observable of $\Qo_{\lam_*}$ and $\Po_{\lam_*}$ imposes $\tau_0=\tau_*$, that is $\Mo_0 = \Mo_{\tau_*}$, which is \eqref{eq:hatI(Q,P)2}. Therefore, $\Mo_{\tau_*}$ is the unique $G$-covariant observable in $\Mcomp(\Qo,\Po)$, and
\begin{align*}
\Icomp(\Qo,\Po) & = D(\lam_*) = \sup_{\rho}\left[\Srel{\Qo^\rho}{\Qo^\rho_{\lam_*}}
+ \Srel{\Po^\rho}{\Po^\rho_{\lam_*}}\right] .
\end{align*}
The first inequality in \eqref{eq:I(Q,P)} then follows by evaluating the sum inside the $\sup$ at any eigenprojection $\rho=\kb{\delta_x}{\delta_x}$ of $\Qo$. On
the other hand, the second inequality is the general bound for
$\Icomp(\Qo,\Po)$ given in \eqref{eq:bound3}.

We finally prove the uniqueness of the optimal approximate joint measurement \eqref{eq:hatI(Q,P)2} in the case $p\neq 2$. If $\Mo$ is any observable in
the optimal set $\Mcomp(\Qo,\Po)$, its covariant version $\Mo_G$ is still in
$\Mcomp(\Qo,\Po)$ by Corollary \ref{prop:invarI^}, hence $\Mo_G = \Mo_{\tau_*}$ by the previous part. By \eqref{eq:covariantization},
$\Mo(x,y)\leq |G|\Mo_G(x,y) = |G| \Mo_{\tau_*}(x,y)$ for all $x,y$.
Since $\Mo_{\tau_*}(x,y)$ has rank $1$, it must then be $\Mo(x,y) =
f(x,y) \Mo_{\tau_*}(x,y)$ for some function $f:\F^2\to [0,|G|]$. The
two normalization requirements $\sum_{x,y} \Mo_{\tau_*}(x,y) = \id$ and
$\sum_{x,y} f(x,y) \Mo_{\tau_*}(x,y) = \sum_{x,y} \Mo(x,y) = \id$ impose
constraints on the coefficients $f(x,y)$. If $d=p^n$ is odd, these constraints
are enough to imply that $f(x,y) = 1$ for all $x,y$. Indeed, this follows since
in this case the observable $\Mo_{\tau_*}$ is informationally complete.
For $d=p$ odd, this is proved in \cite[Prop.\ 9]{CHT12}. In the more
general case $d=p^n$ odd, the same proof still holds, as it relies on the fact
that the inverse Weyl transform of $\tau_*$
\[ 
\hat{\tau}_* (\vec u ) := \Tr\left\{\tau_* W(\vec u )\right\}
= \frac{\sqrt{d}}{2(1+\sqrt{d})}
\left[ \delta_0(u_1) + \delta_0(u_2) + \frac{1}{\sqrt{d}}\left( \rme^{-\frac{2\pi i}{p}\tr{u_1u_2}} + 1 \right) \right]
\]
is nonzero for all $\vec u \in\F^2$ (see \cite[Prop. 12]{CHST14}). The uniqueness statement is thus proved, and this concludes the proof of Theorem \ref{teo:incoMUB}. \end{proof}

\begin{remark}\label{rem:Fourier2}
\begin{enumerate}
\item In the case $p=2$, the above proof only shows that $\Mo_0$ defined in \eqref{eq:hatI(Q,P)2} is the unique $G$-covariant observable in the set $\Mcomp(\Qo,\Po)$.
\item \label{rem:Fourier2.symmetry} In the proof of Theorem \ref{teo:incoMUB}, the dilational symmetries $\{\ma{d}(a) \mid a \in\F_*\}$ simplified the problem of characterizing the set $\Mcomp(\Qo,\Po)$, reducing it to the optimization of the single parameter $\lam$.
\end{enumerate}
\end{remark}

\end{document}